\documentclass[11pt]{article}

\usepackage{amsfonts}
\usepackage{amsmath}
\usepackage{amssymb}
\usepackage{amsthm}
\usepackage[mathscr]{eucal}
\usepackage{bbold}
\usepackage{braket}
\usepackage{color}
\usepackage{comment}
\usepackage{dsfont}
\usepackage{enumitem}
\usepackage{framed}
\usepackage{jheppub}
\usepackage{mathtools}
\usepackage{physics}
\usepackage[normalem]{ulem}
\usepackage{tensor}
\usepackage{thmtools}
\usepackage{thm-restate}
\usepackage{ulem}
\usepackage{tikz}
\usetikzlibrary{math} 
\newtheorem{lemma}{Lemma}
\newtheorem{axiom}{Axiom}
\newtheorem{definition}{Definition}
\newtheorem{corollary}{Corollary}
\def \A {\mathcal{A}}
\def \H {\mathcal{H}}
\def \tr {\text{tr}}
\def\lf{\left\lfloor}
\def\rf{\right\rfloor}

\def\Phantom{\vphantom{\beta}}

\definecolor{acolor}{rgb}{0,.48,0.65}

\definecolor{rcolor}{rgb}{0.9,0.1,0.1}

\definecolor{bcolor}{rgb}{0.1,0,1}

\definecolor{dcolor}{rgb}{0.8,.1,.6}

\definecolor{mcolor}{rgb}{.9,.5,0.5}

\newcommand{\U}[1]{{\underline{#1}}}

\title{Algebras and Hilbert spaces from gravitational path integrals: \\  Understanding Ryu-Takayanagi/HRT as entropy without AdS/CFT}

\author[a]{Eugenia Colafranceschi,} \emailAdd{ecolafranceschi@ucsb.edu }

\author[a]{Xi Dong,} \emailAdd{xidong@ucsb.edu }

\author[a]{Donald Marolf} \emailAdd{marolf@ucsb.edu}

\author[a,b]{and Zhencheng Wang} \emailAdd{zcwang1@illinois.edu}

\affiliation[a]{Department of Physics, University of California, Santa Barbara, CA 93106, USA}
\affiliation[b]{Department of Physics, University of Illinois, Urbana-Champaign, 1110 W. Green St., Urbana, IL 61801, USA}

\abstract{Recent works by Chandrasekaran, Penington, and Witten  have shown in various special contexts that the quantum-corrected Ryu-Takayanagi (RT) entropy (or its covariant Hubeny-Rangamani-Takayanagi (HRT) generalization) can be understood as computing an entropy on an algebra of bulk observables. These arguments do not rely on the existence of a local holographic dual field theory. We show that analogous-but-stronger results hold in any UV-completion of asymptotically anti-de Sitter quantum gravity with a Euclidean path integral satisfying a simple and (largely) familiar set of axioms. We consider a quantum context in which a standard Lorentz-signature classical bulk limit would have Cauchy slices with asymptotic boundaries $B_L \sqcup B_R$ where both $B_L$ and $B_R$ are compact manifolds without boundary. Our main result is then that (the UV-completion of) the quantum gravity path integral defines type I von Neumann algebras ${\cal A}^{B_L}_L$, ${\cal A}^{B_R}_{R}$ of observables acting respectively at $B_L$, $B_R$ such that ${\cal A}^{B_L}_L$, ${\cal A}^{B_R}_{R}$ are commutants. The path integral also defines entropies on ${\cal A}^{B_L}_L, {\cal A}^{B_R}_R$. Positivity of the Hilbert space inner product then turns out to require the entropy of any projection operator to be quantized in the form $\ln N$ for some $N \in {\mathbb Z}^+$ (unless it is infinite). As a result, our entropies can be written in terms of standard density matrices and standard Hilbert space traces. Furthermore, in appropriate semiclassical limits our entropies are computed by the RT-formula with quantum corrections. Our work thus provides a Hilbert space interpretation of the Ryu-Takayanagi entropy. Since our axioms do not severely constrain UV bulk structures, it is plausible that they hold equally well for successful formulations of string field theory, spin-foam models, or any other approach to constructing a UV-complete theory of gravity.
}

\begin{document}

\maketitle

\section{Introduction}
\label{sec:intro}

The last few years have seen significant progress in our understanding of gravitational entropy.  Some important steps forward were the discovery of non-trivial quantum-extremal surfaces in the context of black hole evaporation \cite{Penington:2019npb,Almheiri:2019psf}, as well as the understanding of their relation to gravitational replica calculations \cite{Almheiri:2019qdq,Penington:2019kki}.  These results in turn relied on the general connections between gravitational replicas and (quantum) extremal surfaces derived in \cite{Lewkowycz:2013nqa,Dong:2016hjy,Dong:2017xht}.  As is by now well-known, these observations led to gravitational computations consistent with the  so-called Page curve \cite{Page:1993wv,Page:2013dx}, which is expected from the ideas that black holes are unitary quantum systems with a finite number of internal states and that the number of such states is well-approximated by the exponential of the appropriate Bekenstein-Hawking entropy $S_{BH}$.

The analysis of Hawking radiation is particularly clean in settings where the emitted Hawking radiation is transferred from an asymptotically locally anti-de Sitter (AlAdS) gravitational system to a {\it non-gravitational} quantum mechanical system; i.e., to a system which can depend on a metric only as a fixed non-dynamical background.  Such systems have often been called `baths' in the recent literature.  In this context, and in appropriate semiclassical limits following \cite{Lewkowycz:2013nqa,Dong:2016hjy,Dong:2017xht}, the above results imply that the usual von Neumann entropy of the bath can be studied using quantum extremal surfaces describing what  \cite{Almheiri:2019hni} termed `islands', and that it is given by a formula which is a special case of the quantum-corrected Ryu-Takayanagi/Hubeny-Rangamani-Takayanagi (RT/HRT) formula \cite{Ryu:2006bv,Ryu:2006ef,Hubeny:2007xt}, with quantum corrections understood in the sense of \cite{Engelhardt:2014gca}.

While such arguments were motivated by considerations related to the AdS/CFT correspondence \cite{Maldacena:1997re} (or equivalently from gauge/gravity duality or gravitational holography), the final versions of the arguments rely only on properties of the gravitational path integral.  In particular, at least for bath entropies described by the Island Formula, one may safely interpret the result in terms of standard von Neumann entropies {\it without} assuming the gravitational bulk system to admit a holographic local field theory dual.  The only subtlety here is that (see e.g.\ \cite{Saad:2019lba,Marolf:2020xie,Marolf:2020rpm,Blommaert:2022ucs}) the semiclassical bulk gravitational theory appears to allow baby-universe superselection sectors (often called $\alpha$-sectors) of the form described in \cite{Coleman:1988cy,Giddings:1988cx}, and that the Island Formula in fact characterizes the von Neumann entropy $S_\alpha$ of the bath state $\rho_\alpha$ in a typical $\alpha$-sector \cite{Marolf:2020xie,Marolf:2020rpm} by describing an average over such bulk $\alpha$-sectors.  This explains the observation of
\cite{Giddings:2020yes} that the computation fails to take the form expected for the von Neumann entropy of the bath computed in the total bath state $\rho_{total}$ (which in the above notation takes the form $\rho_{total} = \oplus_\alpha \rho_\alpha$).

The fact that purely bulk arguments suffice to safely interpret quantum extremal surface computations for a bath in terms of standard bath entropies suggests that this lesson may also hold more generally.   In particular, in order to avoid divergences, let us consider a boundary region $B_L$ (in the sense of Ryu and Takayanagi \cite{Ryu:2006bv,Ryu:2006ef})  that is both compact and  without boundary ($\partial B_L = \emptyset$); see figure \ref{fig:fullbound}.  Here the notation $B_L$ denotes the fact that, in the main text below, we will refer to $B_L$ as the `left' part of the boundary while the complementary boundary region $B_R$ will be called the `right' part of the boundary (which we will also require to be compact and satisfy $\partial B_R = \emptyset$).  In this context we might expect that purely-bulk arguments can be used to construct a Hilbert space ${\cal H}_L$ associated with $B_L$, or perhaps a set of such Hilbert spaces
${\cal H}^\mu_L$ (labelled by some index $\mu$), such that the associated RT/HRT formula can be understood in terms of\footnote{The RT/HRT entropy will also include a Shannon term built from the probabilities $p_\mu$ for the system to be in Hilbert space ${\cal H}^\mu_L$.}
\begin{equation}
\label{eq:SvNidef}
S_{\text{vN}}(\rho^\mu_L) : = -\Tr_\mu \rho^\mu_{L} \ln \rho^\mu_{L},
\end{equation}
 where  $\rho^\mu_{L}$ is the density matrix describing the bulk quantum state on ${\cal H}^\mu_L$ and $\Tr_\mu$ is the standard Hilbert space trace on ${\cal H}^\mu_L$.  This is the challenge to be addressed below.

 \begin{figure}[t]
        \centering
\includegraphics[width=0.35\linewidth]{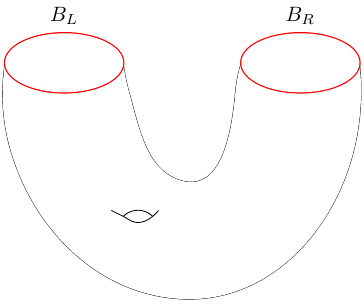}\caption{We consider boundaries $B_L,B_R$ that are complete in the sense that $\partial B_L = \emptyset = \partial \bar B_R$. We also require $B_L, B_R$ to be compact. }
\label{fig:fullbound}
\end{figure}

In certain limiting cases, related results were recently established by Chandrasekaran, Penington, and Witten \cite{Chandrasekaran:2022eqq} (building on \cite{Witten:2021unn} and \cite{Leutheusser:2021qhd,Leutheusser:2021frk}), and especially by Penington and Witten  \cite{Penington:2023dql}; see also \cite{Chandrasekaran:2022cip,Jensen:2023yxy,Kudler-Flam:2023qfl}.  However, the fact that their von Neumann algebras were type II rather than type I meant that their entropies were not given directly by standard Hilbert space traces.  A related comment is that the results of \cite{Chandrasekaran:2022eqq,Penington:2023dql} were valid only in a bulk semiclassical limit in which Hilbert space densities of states diverge and thus that their entropies correspondingly agree with the quantum-corrected Ryu-Takayanagi formula only up to an additive constant.

In contrast, we wish to consider a context in which Hilbert space densities of states are finite so that the above entropies will not require renormalization.   This should allow all constants to be determined.  However, it would also require appropriate couplings to be finite.   In this finite-coupling regime, a primary question will be to understand how the choice of boundary region $B_{L}$ can define the desired Hilbert spaces ${\cal H}_L^\mu$.  In particular, we will be far from the semiclassical regime in which entanglement wedges are well-defined; see e.g.\ the discussion in the final paragraph of \cite{Harlow:2016vwg}.

Of course, the bulk path integral at finite coupling is poorly understood.  Rather than attempt to find and study a UV-completion for any specific model, we instead proceed by simply supposing that we are given a UV-complete finite-coupling bulk asymptotically-locally-AdS (AlAdS) theory with an object that can be called a `Euclidean path integral,' and that this path integral  satisfies a simple set of axioms\footnote{\label{foot:Comp}The Euclidean setting is convenient, but need not be fundamental.  We are hopeful that Euclidean path integrals can generally be derived from Lorentzian path integrals using arguments along the lines of those described in \cite{Marolf:2022ybi}; see e.g.\ \cite{Blommaert:2022lbh} and \cite{Usatyuk:2022afj} for recent comments on such ideas. Our axioms in fact allow the possibility of considering complex metrics.}:

\begin{enumerate}
\item{} {\bf Finiteness:} The path integral gives a well-defined map $\zeta$ from boundary conditions defined by smooth manifolds to the complex numbers ${\mathbb C}$.

\item{} {\bf Reality:} This $\zeta$ is a real function of (possibly complex) boundary conditions; see again footnote \ref{foot:Comp}.

\item{}  {\bf Reflection-Positivity:} $\zeta$ is reflection-positive.

\item{} {\bf Continuity:} $\zeta$ satisfies a rather weak continuity condition described in section \ref{subsec:ax} below.

\item{}  {\bf Factorization:}  For closed boundary manifolds $M_1,M_2$ and their disjoint union $M_1 \sqcup M_2$ we have $\zeta(M_1\sqcup M_2) = \zeta(M_1)\zeta(M_2)$.
\end{enumerate}
The first three axioms are commonly assumed for asymptotically-AdS gravitational theories\footnote{\label{foot:real}The reality condition implies a symmetry of the path integral under reversal of time on the boundary.  It seems likely that this assumption can be dropped without significantly affecting the results, though we leave such a study for future work.}, and were in particular used in \cite{Marolf:2020xie}. In addition, the continuity axiom will be seen in section \ref{subsec:ax} to be extremely weak.  In practice, it seems uncontroversial.  The main subtlety is thus that we assume factorization of the path integral over disconnected closed boundaries, which means that any effects due to spacetime wormholes must either lead to the above-mentioned superselection sectors (in which factorization holds sector-by-sector, so that our analysis can still be applied in that sense), or that such effects must be cancelled by other contributions (as in e.g.\ \cite{Eberhardt:2020bgq,Eberhardt:2021jvj,Saad:2021rcu,Benini:2022hzx}).  As described in \cite{Marolf:2020xie}, it is plausible that the former option holds in general theories of quantum gravity with positive definite inner products.  The technical subtleties in this argument will be reviewed in section \ref{subsec:ax}.

Adopting this axiomatic framework allows us to answer the  challenge associated with equation \eqref{eq:SvNidef} by constructing von Neumann algebras ${\cal A}^{B_L}_L$, ${\cal A}^{B_R}_{R}$ of observables associated with the boundary region $B_L$ and the complementary boundary $B_R$ and by then showing these algebras to contain only type I factors. (We will refer to this property by saying that the entire algebra is of type I.)  The elements of these algebras may be called `boundary observables' in the sense of \cite{Marolf:2008mf}, though we again emphasize that they are defined without assuming the existence of a local dual field theory.  Indeed, as noted above it is plausible that the required axioms to hold for successful UV-completions of general asymptotically-AdS gravitational systems, whether the completion be called string field theory, spin-foam loop quantum gravity, or by some other name.   An important role in our analysis will turn out to be played by the trace inequality recently discussed in \cite{Colafranceschi:2023txs}, which we show to be a consequence of our axioms.

Our construction also leads to associated von Neumann entropies on ${\cal A}^{B_L}_L, {\cal A}^{B_R}_R$ which can be studied using a standard gravitational replica trick.  As usual, when the bulk has an appropriate semiclassical limit that can be described in terms of a local metric theory of gravity, in typical baby-universe superselection sectors\footnote{See e.g. \cite{Marolf:2020xie,Marolf:2020rpm} for discussion of this qualification.} this entropy is given by an RT/HRT-like formula with corrections from both quantum \cite{Faulkner:2013ana} and higher-derivative effects (see e.g.\ \cite{Dong:2013qoa,Camps:2013zua,Miao:2014nxa}).  Furthermore, since ${\cal A}^{B_L}_L, {\cal A}^{B_R}_{R}$ are of type I, they decompose
into direct sums/integrals of type I von Neumann factors.  As a result, the Hilbert space on which these algebras act must also decompose into a sum/integral of terms ${\cal H}^\mu$ (say, labelled by an index $\mu$), each of which is a tensor product ${\cal H}^\mu_{L} \otimes {\cal H}^\mu_{R}$ such that
${\cal A}^{B_L}_L$ acts only on ${\cal H}_{L}$ and  ${\cal A}^{B_R}_{R}$ acts only on ${\cal H}_{R}$.
We also show that ${\cal A}^{B_R}_R$ and ${\cal A}^{B_L}_{L}$ are commutants of each other.
It will then follow that (in typical supserselection sectors) the RT/HRT prescription computes appropriate semiclassical limits of an entropy defined by the
$S_{\text{vN}}(\rho^\mu_L)$ of \eqref{eq:SvNidef}, a Shannon term built from the probabilities $p_\mu$ to be in the Hilbert space ${\cal H}^\mu$, and a set of positive constants $n_\mu$ (which parameterize the general ambiguity in defining an entropy functional on a type I von Neumann algebra).  Finally, a quantization argument will show that the index $\mu$ must be discrete, and that the appropriately-defined $n_\mu$ are positive integers.  As a result, by adding certain finite-dimensional `hidden sectors,' one can construct a modified bulk Hilbert space in which the usual Hilbert space computations of subsystem entropy agree with the result given by RT/HRT.  In this sense we find a Hilbert space interpretation of RT/HRT-entropy.

We  begin in section \ref{sec:PIHS} with an overview of our axioms for (a UV-completion of) a Euclidean gravitational path integral and the construction of the relevant sectors of the gravitational Hilbert space.  The von Neumann algebras ${\cal A}^{B_L}_L$, ${\cal A}^{B_R}_{R}$ are defined in section \ref{sec:algebras} for the case where $B_L$ is chosen for simplicity to be diffeomorphic to $B_R$. The type I structure and the associated decomposition of appropriate sectors of the bulk Hilbert space are then derived in section \ref{sec:typeI}.  Some examples are briefly discussed in
section \ref{sec:ex}.  We conclude in section \ref{sec:disc} by summarizing results, describing potential generalizations, and discussing open issues.

\section{The Path Integral and the Hilbert Space}\label{sec:PIHS}

The goal of this section is to
write down a set of axioms for an object that we will call the Euclidean path integral for a UV-completion of an AlAdS theory of gravity, and to then use those axioms to construct the sectors of the Hilbert space that we will study in sections \ref{sec:algebras} and \ref{sec:typeI} below.    We emphasize that we will require {\it only} that such axioms hold, and that {\it any} object satisfying the axioms may be called a Euclidean path integral, regardless of whether it is in fact computed as an integral over anything resembling Euclidean geometries.  We also emphasize that there may well be many other properties that a good bulk theory should satisfy and which are not captured by our axioms. In other words, we suggest our axioms to be necessary, though probably not at all sufficient, for a theory to be satisfactory.  What we find to be of most interest is just how much can be derived from the simple Axioms \ref{ax:finite}-\ref{ax:factorize} below.

Section \ref{subsec:motivations} describes some brief motivations that stem from considering path integrals that might be defined by sums over Euclidean geometries.  This discussion is informal, but it inspires us to formulate certain axioms that we record in section \ref{subsec:ax}, and which then become the starting point for careful analysis in the remainder of this work.  The relevant Hilbert space sectors are then constructed in section \ref{subsec:Hilbsec}.  Much of the analysis below follows \cite{Marolf:2020xie}.

\subsection{Motivations from sums over geometries}
\label{subsec:motivations}

While in the end we will not actually require that our path integral be formulated as a sum over Euclidean geometries, we would like our axioms to apply to any such cases that might exist.  We dedicate this section to brief comments on such
hypothetical objects, which we take as motivations for axioms to be stated below.  We emphasize that this phase of our discussion is informal and, due to the dearth of examples, it is necessarily imprecise.  Formal discussion will commence in section \ref{subsec:ax} with the formulation of our axioms, from which the rest of our analysis will follow.

Let us thus briefly consider a path integral that actually integrates over a set of fields, among which is the (Euclidean-signature) metric. The bulk fields may also include scalars, fermions, gauge fields, etc.
We will take the above-mentioned sum over metrics to include a sum over all possible topologies.  The bulk fields will be collectively denoted $\phi$, for which the corresponding Euclidean action will be $S[\phi]$.
To every smooth \textit{closed}\footnote{Recall that in mathematics a manifold is said to be closed when it is compact and has no boundary (i.e., $\partial M = \emptyset$).  We will use the term in this sense everywhere below, though we will sometimes add the redundant qualifier compact in order to remind the reader of this property.}   AlAdS  boundary $M$ at which appropriate (potentially complex) boundary conditions are specified,
a Euclidean path integral would then assign the complex number
    \begin{equation}\label{PI}
      \zeta(M) \coloneqq \int_{\phi \sim M} \mathcal{D}\phi e^{-S[\phi]} .
    \end{equation}
    Here we use the symbol $M$ to denote not just the boundary manifold, but also the relevant boundary conditions for the bulk fields $\phi$.     The notation $\phi \sim M$ in \eqref{PI} indicates that we integrate only over bulk fields $\phi$ satisfying such conditions.

In order to avoid overuse of terms involving the word `boundary,' we henceforth refer to the boundary conditions on bulk fields as {\it sources}, and we refer to $M$ as a (boundary) {\it source manifold} to remind the reader of our inclusive terminology.   This terminology will seem natural to practitioners of AdS/CFT, though long experience in that context has established that, even without invoking such a duality, the boundary conditions for bulk fields play precisely the same role as sources for familiar non-gravitational quantum field theories.  In the AlAdS$_d$ context with $d$ even, the appropriate notion of sources/boundary conditions will typically be given by equivalence classes under Weyl transformations.

It is reasonable to expect $\zeta(M)$ to be finite for smooth $M$, and for $\zeta(M)$ to enjoy some degree of continuity under appropriately-small deformations of the boundary conditions described by $M$.  For the present purposes we allow the sources (i.e., the boundary metric and other boundary fields) described by $M$ to be complex, though one can also restrict the discussion to real boundary conditions (or to complex linear combinations thereof).  For complex sources, due to the expected reality property $[S(\phi)]^* = S(\phi^*)$, expression \eqref{PI} suggests that $[\zeta(M)]^* = \zeta(M^*)$  where ${}^*$  denotes complex conjugation and, in particular,  $M^*$ is the same manifold as $M$ but with complex-conjugated sources.

Let us imagine that we cut the path integral \eqref{PI} into two parts along a slice $\Sigma_{bulk}$ through the bulk spacetime. By this we mean that we slice each configuration $\phi$ that enters into the path integral into two parts, and that in all cases we call the cut $\Sigma_{bulk}$ even though the geometry of $\Sigma_{bulk}$, and in fact the topology of $\Sigma_{bulk}$, will depend on $\phi$. We will, however, require the intersection $\partial \Sigma_{bulk}$ of $\Sigma_{bulk}$ with the AlAdS boundary $M$ to be independent of $\phi$.  In the usual way, it is natural to take each of the two resulting pieces of the path integral to compute the wavefunction (or the complex conjugate of a wavefunction) of a state in a Hilbert space ${\cal H}_{\partial \Sigma_{bulk}}$ defined by the choice of $\partial \Sigma_{bulk}$.  The original (uncut) path integral then computes the inner product in ${\cal H}_{\partial \Sigma_{bulk}}$ of the two states thus defined; see figure \ref{fig:cut}.  In particular, when the states are identical, the original uncut path integral computes the norm squared of the state and thus should be required to give a non-negative result.

 \begin{figure}[ht!]
        \centering
\includegraphics[width=0.35\linewidth]{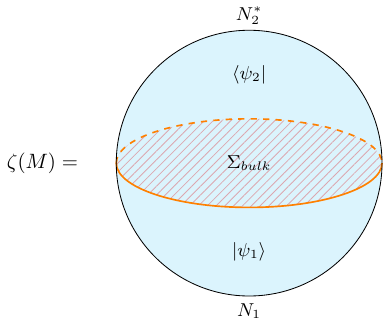}\caption{A slice $\Sigma_{bulk}$ (red) of the path integral intersects the (here, spherical) AlAdS boundary $M$  at a codimension-2 surface $\partial \Sigma_{bulk}$ (red circle) which splits $M$ into two hemispheres $N_1$ and $N_2^*$.  Each half of the path integral defines a quantum state $|\psi_i\rangle$ by computing the wavefunction of $\psi_i$ on $\Sigma_{bulk}$.  These wavefunctions can be thought of as the result of Euclidean evolution from the boundary conditions $N_i$, and the full path integral defined by $M$ can then be regarded as computing the inner product $\langle\psi_2|\psi_1 \rangle$.}
\label{fig:cut}
\end{figure}

Furthermore, it is natural to generalize the above discussion by replacing $M$ with a finite formal linear combination of source manifolds
\begin{equation}
\label{eq:flc}
M : = \sum_{I=1}^n \gamma_I M_I,
\end{equation} for some $n\in {\mathbb Z}^+$ with $\gamma_I \in {\mathbb C}$, in which case we simply use linearity to define
\begin{equation}
\label{eq:zlin}
\zeta(M) : = \sum_{I=1}^n \gamma_I \zeta(M_I).
\end{equation}
In this case we also define
\begin{equation}
\label{eq:starlin}
M^* : = \sum_{I=1}^n \gamma_I^* M_I^*.
\end{equation}
In particular, such formal sums $M$ can sometimes again be `sliced' (or, perhaps better, factorized) into two pieces (factors) and, when the two pieces are isomorphic up to the appropriate complex conjugation, we again expect $\zeta(M)$ to compute a non-negative norm squared.  Below, we will use the notation $X^d$ to denote the set of {\it smooth} $d$-dimensional closed (i.e., compact and without boundary) source manifolds $M$ appropriate to some given theory.  We then use the underlined notation $\U{X}^d$ to denote formal finite linear combinations of such manifolds with coefficients in ${\mathbb C}$ as in \eqref{eq:flc} (with $M_I \in X^d$).  Members of both $X^d$ and $\U{X}^d$ will be denoted $M$ to avoid cumbersome notation.   This should not cause confusion since, as above, we will extend any function $\zeta: X^d \rightarrow {\mathbb C}$ to the domain $\U{X}^d$ by linearity; i.e., via \eqref{eq:zlin}.

\subsection{Some axioms for the UV-completion of a bulk path integral}
\label{subsec:ax}

The above brief discussion motivates the following axioms for the UV-completion of any $(d+1)$-dimensional AlAdS
Euclidean quantum gravity path integral $\zeta(M)$.
We also expect our axioms to apply to UV-completions of bulk gravitational theories of spacetimes asymptotic to ${\cal M}_{d+1} \times X$ where ${\cal M}_{d+1}$ is AlAdS$_{d+1}$ and $X$ is a fixed compact manifold of arbitrary dimension, as well as to other asymptotic structures such as those described in \cite{Itzhaki:1998dd}. However, for simplicity of discussion we will refer only to the AlAdS context below.   We also emphasize again that we make no attempt to state a complete set of such axioms.  Thus, while we expect our axioms to be satisfied in well-behaved contexts, they are almost certainly insufficient to fully characterize the desired UV-completions.

Our first four axioms are as follows:

\begin{axiom}
\label{ax:finite}  {\bf Finiteness:}  For some space of $d$-dimensional closed (and thus compact) source manifolds $X^d$, we are given a function $\zeta: X^d \rightarrow {\mathbb C}$; i.e., $\zeta(M)$ is well-defined and finite for every  $M \in X^d$.  Although we do not specify the detailed nature of the allowed sources, the sources should be given by fields (or equivalence classes thereof) on an underlying manifold.  Furthermore, $X^d$ should include any smooth closed manifold with smooth source fields of the allowed types.
\end{axiom}

\begin{axiom}
\label{ax:real}  {\bf Reality\footnote{As stated in footnote \ref{foot:real}, this condition implies a symmetry of the path integral under reversal of time on the boundary.  It thus seems likely that this assumption can be dropped without significantly affecting the results, though we leave such a study for future work.}}:  For every $M \in \U{X}^d$, we have both $M^* \in \U{X}^d$ and $[\zeta(M)]^* = \zeta(M^*)$, where $M^*$ is defined by \eqref{eq:starlin}.
\end{axiom}

\begin{figure}[ht!]
        \centering
\includegraphics[scale=0.8]{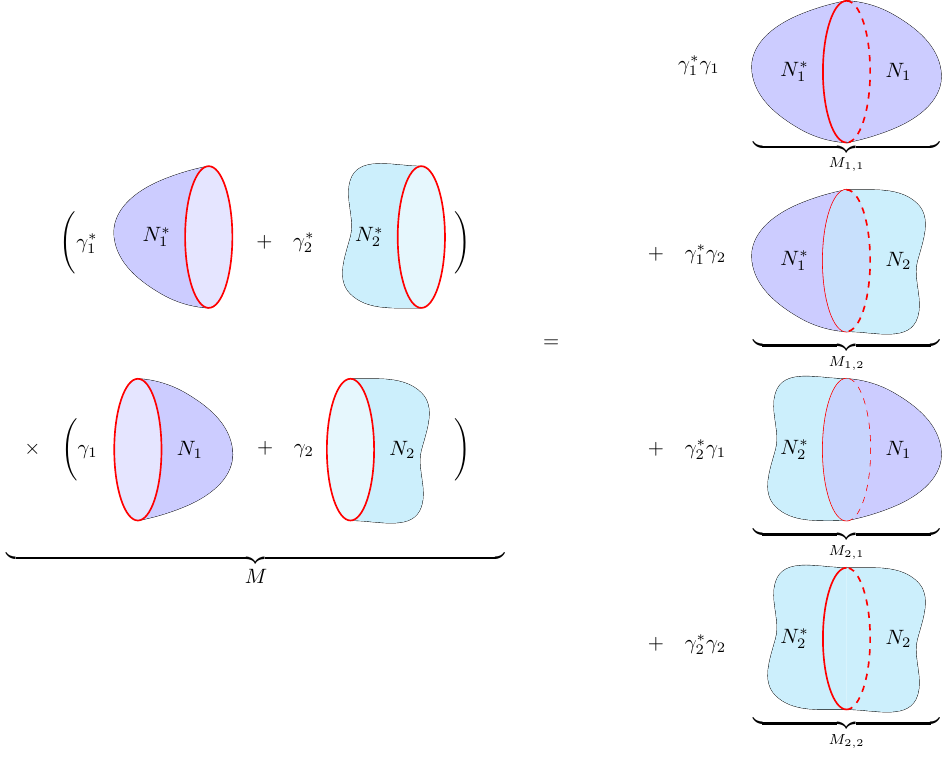}\caption{ A reflection-symmetric linear combination $M\in \U{X}^d$ of smooth source manifolds. The representation on the  left side of the equality makes reflection-symmetry manifest.   On the right side of the equality, the same $M$ is shown as an explicit sum of  terms, each proportional to a source manifold $M_{I,J}$ that can be cut into $N_I^*$ and $N_J$, and with coefficients of form $\gamma_I^* \gamma_J$. Here and in all figures below, we make no attempt to distinguish a given source from its complex conjugate.  Thus $N_I^*$ appears as simply a reflected version of $N_I$. The `diagonal' manifolds $M_{I,I}$ are individually reflection-symmetric.  Axiom \ref{ax:RP} requires that such reflection-symmetric $M$ have $\zeta(M) \ge 0$.}
\label{fig:RP}
\end{figure}

\begin{axiom}
\label{ax:RP} {\bf Reflection Positivity:}    Suppose for some $n \in {\mathbb Z}^+$ that $M\in \U{X}^d$ can be written in the form $ M = \sum_{I,J=1}^n  \gamma_I^* \gamma_J M_{I,J}$  where $\gamma_I \in {\mathbb C}$, $\gamma_I^*$ denotes the complex conjugate of $\gamma_I$,  and where each $M_{I,J}$ can be sliced into two parts $N_I^*$, $N_J$; see figure \ref{fig:RP}.  By such a slicing, we mean that there is a smooth codimension-1 hypersurface $\Sigma_{I,J}$ in $M_{I,J}$ that partitions $M_{I,J}$ into $N_I^*$ and $N_J$, so that $N_I^*$ and $N_J$ are source manifolds with boundaries.
Specifically,  the above notation requires that the same source-manifold-with-boundary $N_I^*$ is obtained from slicing $M_{I,J}$ for each $J$, and the same source-manifold-with-boundary $N_J$ is obtained by slicing $M_{I,J}$ for each $I$.  In particular, slicing the diagonal closed manifold $M_{I,I}$ along $\Sigma_{I,I}$ yields $N_I^*$ and $N_I$. The notation $N_I^*$ indicates that each diagonal source manifold $M_{I,I}$ admits a diffeomorphism $\phi_{I,I}$ that both acts as a reflection about $\Sigma_{I,I}$ and  complex-conjugates all sources\footnote{In particular, if our theory requires the source manifolds to be oriented and if $\zeta$ is sensitive to the choice of orientation (so that the orientation is a non-trivial source), then we must declare the orientation to be an imaginary source, which would change sign under the complex-conjugation inherent in the operation ${}^*$. This change of sign can be seen by noting that $N_I^*$ and $N_I$ must have consistent orientations when we glue them along $\Sigma_{I,I}$ back into $M_{I,I}$.}.  When these conditions hold, $\zeta(M)$ is a non-negative real number.
\end{axiom}

\begin{figure}[ht!]
        \centering
\includegraphics[width=0.7\linewidth]{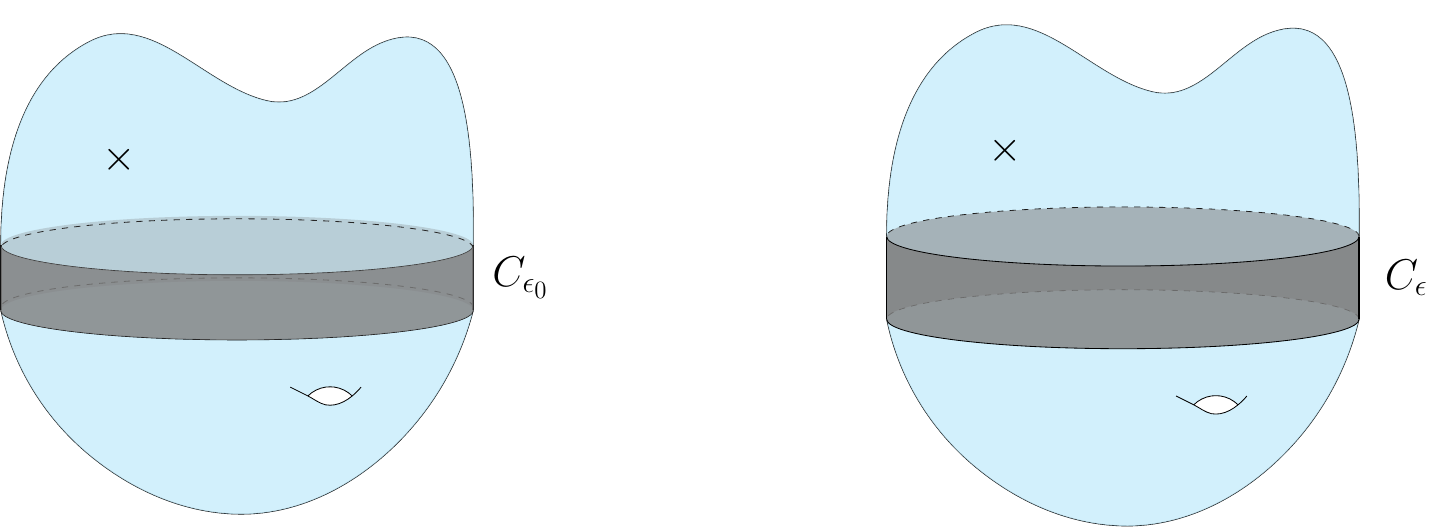}\caption{ The source manifold $M_{\epsilon_0}$ contains a cylinder $C_{\epsilon_0}$ of length $\epsilon_0$.  Changing the length of this cylinder to $\epsilon$ defines a new source manifold $M_\epsilon$. Here and in all figures below, the symbols $\times$ denote potential features of smooth sources whose details are not shown but which serve to distinguish certain boundary points.  For example, such features might be local peaks in the Kretchman scalar of the boundary metric, extrema of a smooth scalar source, or points where a fermionic source becomes large.  The main role of these features in our figures is to provide a simple and clean visualization of effects that arise when boundary-sources break symmetries of the simple cases that we choose to depict.}
\label{fig:cont}
\end{figure}

\begin{axiom}
\label{ax:continuity}  {\bf Continuity:}
Suppose that a source manifold $M\in X^d$ contains a region diffeomorphic to an
(orthogonal) cylinder source-manifold-with-boundary $C_{{\epsilon_0}}$ of some length ${\epsilon_0} >0 $; see figure \ref{fig:cont}.  The term ``(orthogonal) cylinder source-manifold-with-boundary'' indicates that $C_{\epsilon_0}$ is topologically of the form $B \times [0,{\epsilon_0}]$, and that $C_{\epsilon_0}$ admits a Killing field $\xi$ which generates a local symmetry of $C_{\epsilon_0}$ and the sources that it represents, with $\xi$ orthogonal to $\partial C_{\epsilon_0}$.  By a local symmetry, we mean that at each point in the interior of $C_{\epsilon_0}$ the flow along $\xi$ is well-defined for at least some finite values of the Killing parameter, though of course the boundaries $\partial C_{\epsilon_0}$ prohibit $C_{\epsilon_0}$ from enjoying a translational symmetry along these flows.  The statement that $C_{\epsilon_0}$ has length ${\epsilon_0}$ means that the two copies of $B$ in $\partial C_{\epsilon_0}$ are related by flow through a Killing parameter ${\epsilon_0}$, though the actual value of ${\epsilon_0}$ is not meaningful since we have not fixed a preferred normalization for the Killing field $\xi$.  For simplicity we will drop the qualifier `orthogonal' when discussing cylinders $C_\epsilon$ below.

Let us now write the above $M$ as $M_{\epsilon_0}$ and define a related family of manifolds $M_\epsilon$ by replacing the $C_{\epsilon_0}$  contained in $M_{\epsilon_0}$ with the analogous cylinder $C_\epsilon$.  The resulting $\zeta(M_\epsilon)$ is then required to be a continuous function of $\epsilon$ at all $\epsilon >0$.
\end{axiom}
The reader will note that our continuity condition is extremely weak, and that one generally expects rather stronger continuity conditions to hold.  However, Axiom \ref{ax:continuity} has the benefit of being simple to state for general boundary dimension $d$, and it will turn out to be sufficient for our purposes below.

Axioms \ref{ax:finite}, \ref{ax:real}, and \ref{ax:RP} are requirements explicitly stated in \cite{Marolf:2020xie}. Continuity was not discussed in \cite{Marolf:2020xie}, but the mild hypothesis stated in Axiom \ref{ax:continuity} is a natural addition.   However, in addition to Axioms \ref{ax:finite}-\ref{ax:RP}, in its discussion of general theories Ref.~\cite{Marolf:2020xie} also implicitly used
additional assumptions to deal with spacetime wormholes. In particular,  as explained in \cite{Marolf:2020xie}, Axioms \ref{ax:finite}-\ref{ax:RP} imply the set of real source manifolds $M\in X^d$ to be associated with a collection of symmetric operators defined on a common dense domain in a natural quantum gravity Hilbert space.  These axioms also further imply that any two such operators commute on this domain. It was then suggested in \cite{Marolf:2020xie} that each of these symmetric operators should have a self-adjoint extension to the full Hilbert space, and that the physically correct such extensions were again mutually commuting.   This outcome will seem natural to many physicists, though the above results are not in fact sufficient to prove that it actually occurs. See in particular \cite{RS} for an example where symmetric operators that commute on a common invariant dense domain are essentially self-adjoint (so that they have a unique self-adjoint extension to the full Hilbert space) but where their self-adjoint extensions nevertheless fail to commute.  Furthermore, while this issue may seem to some like an abstruse technical concern, it may have some connection to the ambiguity in constructing ensembles dual to Jackiw-Teitelboim (JT) gravity discussed in e.g.\ \cite{Stanford:2019vob,Johnson:2019eik,Johnson:2020exp,Johnson:2020mwi,Johnson:2021owr,Johnson:2021zuo,Johnson:2022wsr,Johnson:2023ofr,
Turiaci:2023jfa}.  It would be very interesting to understand the potential instabilities discussed in such works in terms of the algebraic language used here.

Despite the above caveats, if the suggestion of \cite{Marolf:2020xie} {\it does} hold, then the self-adjoint extensions can be simultaneously diagonalized on the full quantum gravity Hilbert space.  The resulting simultaneous eigenspaces of these operators are then called ``baby universe superselection sectors,'' and they have the property that any $M\in X^d$ defines an operator proportional to the identity on each such sector.  As a result, one may show (see again \cite{Marolf:2020xie} as well as the explicit discussion of \cite{Blommaert:2022ucs} for JT gravity) that considering any given sector on its own may be thought of as working with a modified path integral that exhibits the factorization property
\begin{equation}
\label{eq:factorize}
\zeta(M_1\sqcup M_2) = \zeta(M_1) \zeta(M_2)
\end{equation}
for closed source manifolds $M_1, M_2 \in X^d$.  Here the symbol $\sqcup$ denotes the disjoint union of source-manifolds-without-boundary. Such superselection sectors are often called $\alpha$-sectors, and they play the role of the $\alpha$-states described in \cite{Coleman:1988cy,Giddings:1988cx}.  Once this structure is established, it is then sufficient to deal with each baby universe superselection sector individually.

It is tempting to expect bulk path integrals of UV-complete theories to be equivalent to collections of such superselection sectors. Here we include the case suggested in \cite{McNamara:2020uza} where there is only one such superselection sector in the collection,  so that each $M\in X^d$ defines an operator proportional to the identity on the entire Hilbert space.  However, we emphasize that in the presence of multiple such superselection sectors, it would be natural to simply work with each such sector separately. Doing so would allow us to frame arguments in terms of path integrals satisfying \eqref{eq:factorize}.  As a result, rather than introduce further complicated axioms which would require the path integral to be a sum over superselection sectors as described above, we will simply assume that we start with a path integral satisfying  the factorization property \eqref{eq:factorize}.  In particular, we include the following axiom:

\begin{axiom}
\label{ax:CW}
\label{ax:factorize}  {\bf Factorization:}
For source manifolds $M_1, M_2 \in X^d$, the function $\zeta$ satisfies \eqref{eq:factorize}.  Note that this is equivalent to requiring \eqref{eq:factorize} to hold for $M_1, M_2 \in \U{X}^d$.
\end{axiom}

We will investigate consequences of our axioms below.

\subsection{Sectors of the quantum gravity Hilbert space}
\label{subsec:Hilbsec}

As noted at the beginning of this section, one expects to be able to obtain states of any quantum gravity theory by `cutting open' the associated path integral. The associated formal construction from our axioms will be described shortly.  This construction is standard, and in particular closely parallels the quantum field theory (QFT) case described in e.g.\ \cite{Glimm:1987ng}. As remarked in the introduction, our approach will be to remain agnostic about the inner workings of the path integral, and simply to view it as a function $\zeta : X^d \rightarrow {\mathbb C}$ satisfying Axioms \ref{ax:finite}-\ref{ax:factorize}.

We refer the reader to the literature for further discussion of what it means to cut open a quantum gravity path integral; see e.g.\ \cite{Marolf:2020xie}.  However, at an abstract level it is clear that doing so requires that we cut any closed AlAdS boundary $M$ into two pieces $N_1, N_2$ with $\partial N_1 = \partial N_2$.  We should then associate quantum gravity states with these two pieces such that the inner product of the two states is $\zeta(M)$.

However,  there are several subtleties in this process that merit discussion. The first such subtlety arises when there are open sets in $N_1, N_2$ that contain $\partial N_1 = \partial N_2$ and which admit non-trivial symmetries.  In that case, there is more than one way to glue the pieces $N_1, N_2$ back together to obtain a smooth manifold.  Furthermore, each such gluing $g$ generally leads to a different closed manifold $M_g$, only one of which can be the original $M$ from which the pieces $N_1, N_2$ were cut.  As a result, it is not sufficient to think of $N_1, N_2$  as diffeomorphism equivalence classes of source manifolds with boundaries.  Instead, we see that we should think of the points on $\partial N_1 = \partial N_2$ as being labelled, so that $M$ is reconstructed by gluing $N_1$ to $N_2$ along their boundaries in the manner dictated by matching identical labels.  As a result, we will henceforth use the notation $N$ to denote a manifold with boundary $\partial N$, together with a labelling of points on $\partial N$.

Before proceeding, it may be useful to illustrate the labelling of points on $\partial N$ with some simple examples.  The simplest case occurs for $d=1$, where $\partial N$ has dimension $d-1=0$ and so consists only of discrete points.   Each such point must be assigned a distinct label.  For example, for a given source-manifold-with-boundary $N_I$, this number might be $2$ and $N_I$ might simply be a line segment (say, of some fixed length $\beta$) with the two points in $\partial N_I$ labelled $0$ and $1$.  If $N_{II}$ is a diffeomorphic source-manifold-with-boundary but with boundary points labelled $1,2$, then it is considered to be a different  source-manifold-with-boundary and we write $N_I \neq N_{II}$.  Similarly, suppose that $N_{III}, N_{IV}$ are  again line segments of length $\beta$ with boundary points $0,1$, and that we in fact introduce a coordinate $\theta$ on both line segments that measures $\beta^{-1}$ times the proper distance from the boundary point $0$.  Let us also suppose that $N_{III}$ comes equipped with some scalar source $\phi(\theta)$ which increases monotonically for $\theta \in [0,1]$, and that $N_{IV}$ has a corresponding scalar source $\phi(1-\theta)$.  Then while $N_{III}$ and $N_{IV}$ are related by a source-preserving diffeomorphism $\theta \rightarrow 1-\theta$, this diffeomorphism fails to preserve the labelling of points in $\partial N_{III}, \partial N_{IV}$.  We will thus again write $N_{III} \neq N_{IV}$.

Note that we can also upgrade the above examples to $d=2$ by replacing any boundary-point (say, that is labelled by some $i$ above) with a circle, and by labeling points on this circle $(i, \theta)$ where $\theta \in [0, 2\pi)$ is a standard angular coordinate while $i$ now effectively labels the circle as a whole.   We can also use a similar notation to define a disjoint union operation $\sqcup$ that will be of frequent use below.  For any boundary $\partial N$ for which its points have labels $\{I\}$, we define the disjoint union $\partial N \sqcup \partial N$ to have labels $(i, I)$ where $i=1$ on the first copy of $\partial N$ and $i=2$ on the second copy, and where the labels $I$ are again assigned just as in $\partial N$.  In particular, if $\partial N= S^1$, then we may label $S^1 \sqcup S^1$ with $(i, \theta)$ for $i=1,2$ and $\theta\in [0, 2\pi)$ as above.

Returning to our main discussion, let us now suppose that we are given two manifolds $N_1, N_2$ with labelled boundaries $\partial N_1, \partial N_2$, such that the boundary labels define a diffeomorphism $\phi: \partial N_1 \rightarrow \partial N_2$. (Recall that diffeomorphisms are required to be surjective.) We can then use this $\phi$ to glue $N_1$ to $N_2$ to define a closed manifold $M$ without boundary.  However, there is no guarantee that the resulting boundary fields on $M$ will be smooth, or indeed even that they will be continuous.  As a result, $\zeta(M)$ may not be well-defined.

We will deal with this issue by using the following simple expedient: rather than attempting to construct the entire quantum gravity Hilbert space, we will instead construct only sectors that are associated with certain types of data on the codimension-2 boundaries $\partial N$.  In particular, we will consider only source-manifolds-with-boundary $N$ that are {\it rimmed} in the following sense:

\begin{definition}
\label{def:rim}
A source manifold $N$ with boundary $\partial N$ will be said to be {\it rimmed} when there is a neighborhood $N_\epsilon$ of $\partial N$ such that $N_\epsilon$ is diffeomorphic to some cylinder source manifold $C_\epsilon$ of the form defined in Axiom \ref{ax:continuity} and which satisfies the reality condition $(C_\epsilon^*)^t = C_\epsilon$ with the operation ${}^*$ defined as in Axiom \ref{ax:RP} and the transpose operation ${}^t$ defined by simply swapping the labels attached to the two boundaries of the cylinder.  Such cylinders will be called `self-adjoint.'   The region $N_\epsilon$ is then called a {\it rim} of $N$.
\end{definition}
We note that, in order for rimmed manifolds to exist, `self-adjoint' cylinders (satisfying $(C_\epsilon^*)^t = C_\epsilon$) must also exist\footnote{For certain choices of $\partial N$, it can happen that there are no rimmed manifolds $N$ with the desired boundary.  This can occur, for example, if the path integral is defined only for closed manifolds $M$ with non-negative Ricci scalar (perhaps to avoid instabilities like those described in \cite{Seiberg:1999xz}) and if the Ricci scalar of $\partial N$ is negative.  In this case, the Hilbert space sector $\mathcal H_{\partial N}$ described below is trivial.}.

We also make the following definitions:
\begin{definition}
\label{def:bndy}
We will say that two rimmed source-manifolds $N_1,N_2$ with boundaries $\partial N_1$, $\partial N_2$ {\it agree on their boundaries} when they admit rims $N_{1\epsilon}, N_{2\epsilon}$ that are related by a diffeomorphism that preserves sources and which also preserves the labels on $\partial N_1, \partial N_2$.  By the local translation symmetry, the data on all of $N_{1\epsilon}, N_{2\epsilon}$ is determined by data at $\partial N_1, \partial N_2$, so we will write $\partial N_1 = \partial N_2$ to denote the above agreement on the rims $N_{1\epsilon}, N_{2\epsilon}$.    We will similarly use the symbol $\partial N$ to denote the manifold at the boundary of the source-manifold-with-boundary $N$ together with enough information about the sources on $N$ to reconstruct sufficiently small rims $N_\epsilon$.
\end{definition}

The utility of  restricting to rimmed source-manifolds is that, when two rimmed source-manifolds $N_1,N_2$ agree at their boundary (in the sense that $\partial N_1 = \partial N_2$), it is then clear that a reflection of $N_1$ can be glued to $N_2$ to define a new smooth source-manifold-without-boundary.   However, since $N_1^*$ already incorporates the required reflection of $N_1$, and since Definition \ref{def:rim} required the rims to be self-adjoint, it is natural to simply discuss gluing $N_1^*$ to $N_2$.  For $\partial N_1 = \partial N_2$, we denote the result of this gluing by $M_{N_1^*N_2}$.   For future use, we note that our gluing operation acts symmetrically on $N_1,N_2$ up to the action of ${}^*$ so that we have
\begin{equation}
\label{eq:Msym}
\left( M_{N_1^* N_2} \right)^* = M_{N_2^* N_1} = M_{N_1 N_2^*},
\end{equation}
where ``='' means that the two source manifolds are related by a source-preserving diffeomorphism;
see figure \ref{fig:conjprod}.
Due to this symmetry, from Axiom \ref{ax:real} we also have
\begin{equation}
\label{eq:zeta*}
\zeta(M_{N_1^* N_2}) = \zeta(M_{N_2^* N_1})^*.
\end{equation}

\begin{figure}[t]
        \centering
\includegraphics[scale=0.8]{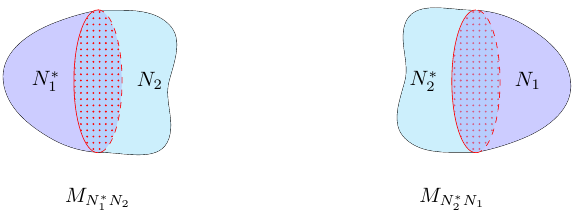}\caption{  Gluing $N^*_1$ to $N_2$  and gluing  $N_2^*$ to $N_1$ defines source-manifolds-without-boundary $M_{N_1^* N_2}$ and $M_{N_2^* N_1}$ that are related by a diffeomorphism that complex-conjugates sources. As depicted here, the relevant diffeomorphism acts as a reflection across the shaded plane.   Thus $(M_{N_1^* N_2})^* = M_{N_1 N_2^*}$, where ``='' means that the two are related by a source-preserving diffeomorphism.}
\label{fig:conjprod}
\end{figure}

In particular, for a given such choice of $\partial N$ (in the sense of Definition \ref{def:bndy}) we can define a sector ${\cal H}_{\partial N}$ of the quantum gravity Hilbert space. To do so, define $Y^d_{\partial N}$ as the space of compact rimmed source manifolds $N$ having the given boundary $\partial N$.  From $Y^d_{\partial N}$, we can then construct the space $\U{Y}^d_{\partial N}$ of finite formal linear combinations $N = \sum_{I=1}^n \gamma_I N_I$ with $\gamma_I \in {\mathbb C}$ and $N_I \in Y^d_{\partial N}$.
The next step is to associate a (not necessarily distinct) state
$|N \rangle$ with each $N \in \U{Y}^d_{\partial N}$.  Two such states $|N_1 \rangle, |N_2 \rangle$ are defined to have (pre-)inner product
\begin{equation}
\label{eq:ip}
\langle N_1 | N_2 \rangle : = \zeta(M_{N_1^* N_2}),
\end{equation}
where $M_{N_1^* N_2} \in \U{X}^d$ is defined by using the distributive law $M_{N_1(\alpha N_2 + \beta N_3)} = \alpha M_{N_1N_2} + \beta M_{N_1N_3}$ and similarly for  $M_{(\alpha N_1^* + \beta N_3^*)N_2}$.
The (pre-)inner product is Hermitian due to \eqref{eq:zeta*}, and is positive semi-definite by Axiom \ref{ax:RP}.   We may then say that \eqref{eq:ip} defines a pre-Hilbert space $H_{\partial N}$. Taking the quotient by the space ${\cal N}_{\partial N}$ of null vectors and completing\footnote{We complete $H_{\partial N}/{\cal N}_{\partial N}$ in the standard way using equivalence classes of Cauchy sequences $\{\ket{N_m}\}$, where two such sequences are equivalent if the norm of their difference approaches $0$.} the result then yields a Hilbert space ${\cal H}_{\partial N}$ that we call the $\partial N$-sector of the full quantum gravity Hilbert space.    Below, we will use the notation $|N\rangle$ to denote both elements of the pre-Hilbert space $H_{\partial N}$ and the associated equivalence class in ${\cal H}_{\partial N}$, though the distinction should always be clear from the context.  Indeed, since $\U{Y}^d_{\partial N}$ allows only finite linear combinations, it may often be the case that ${\cal N}_{\partial N}$ is empty and the quotient is trivial.

The above expedient will allow us to proceed quickly to constructing and studying algebras of operators on ${\cal H}_{\partial N}$ without characterizing in detail the degree of differentiability of sources on $M$ that is required for $\zeta(M)$ to be finite, and also without analyzing the manner in which divergences arise when such conditions fail.  If our goal is to construct states associated with static Lorentz-signature boundary conditions, then one may expect that our restriction to rimmed surfaces gives the full such Hilbert space.  One argument for this comes from AdS/CFT, in which case the rims correspond to insertions of $e^{-\epsilon H}$ for some $\epsilon$.  Since $e^{-\epsilon H}$ is invertible, even at fixed finite $\epsilon$ the rimmed surfaces will generate a complete set of states.  However, even without relying on AdS/CFT,
since we allow the rim $N_\epsilon$ to be arbitrarily small, if our path integral is sufficiently continuous in $\epsilon$ then rimmed surfaces may still provide full information about the sector of the theory associated with a given $\partial N$ of the type described above\footnote{Note, however, that Axiom \ref{ax:continuity} requires only continuity when the length $\epsilon >0$ of some cylinder is slightly deformed and, in particular, it does not necessarily require continuity when the limiting source-manifold no longer contains a cylinder of non-zero length.}.

Nonetheless, two shortcomings to our approach should be discussed.  The first is that we obtain no information about inner products $\langle N_1 | N_2\rangle$ when $\partial N_1 \neq \partial N_2$.  Such inner products do not necessarily vanish, especially in low dimensions. Indeed, for $d=1$ both AdS/CFT and the associated semiclassical bulk computations suggest  that the inner product can be non-zero even when the sources on $M_{N_1^* N_2}$ are discontinuous (so long as $M_{N_1^* N_2}$ is a well-defined topological manifold).

The second shortcoming is that, while we expect the above restrictions to allow us to construct all states associated with static Lorentz-signature boundaries, at least in high dimensions we expect to miss sectors of the quantum gravity Hilbert space associated with non-static boundaries. Based on both the AdS/CFT context and the divergences that manifest themselves in the associated semiclassical bulk computations, we expect this issue to be related to what one finds when studying quantum fields on curved spacetime, where in high dimensions the space of states on a given Cauchy slice $\Sigma$ (say, specified by the correlation functions of fields and their derivatives on $\Sigma$) can depend not only on the metric induced on $\Sigma$ but also on various normal derivatives of background fields (sources) evaluated at $\Sigma$.  It would be interesting to return to both of these issues in the future, though a full investigation of the second issue seems likely to require a Lorentz-signature analysis.


\section{Operator Algebras from the Path Integral}
\label{sec:algebras}

We have thus far described how our path integral $\zeta$ can be used to construct sectors ${\cal H}_{\partial N}$ of the quantum gravity Hilbert space.  But it can also be used to construct operators, and this construction will be useful in understanding the further structure of ${\cal H}_{\partial N}$ and the relation to RT entropy.  To define such operators, let us again consider the space of compact rimmed surfaces $Y^d_{\partial N}$ for some choice of codimension-2 boundary $\partial N$. We will now further suppose that $\partial N$ is the disjoint union of two pieces, $\partial N =  B_{in} \sqcup B_{out}$, with both $B_{in}, B_{out}$ being compact and closed (in the sense that $\partial B_{in} = \partial B_{out} = \emptyset$).  Then for any appropriate additional boundary (which for later purposes we call $B_R$), one sees that any $a\in  Y^d_{\partial N}$ defines an operator from $H_{B_{in} \sqcup B_R}$ to $H_{B_{out} \sqcup B_R}$ by gluing surfaces along $B_{in}$. We may thus construct operators that preserve a sector of the form $H_{B_L \sqcup B_R}$ by considering the case $B_{in} = B_{out} = B_L$. Here we refer to $B_L$ as the `left' part of $B_L \sqcup B_R$ while $B_R$ is the `right' part.

For any $B$, we may then endow the surfaces $Y^d_{B \sqcup B}$ with a multiplication operation  which promotes the space  of formal linear combinations $\U{Y}^d_{B \sqcup B}$ to an algebra $A_L^{B}$.  In fact, we will also introduce an analogous `right algebra' $A_R^{B}$ in section \ref{subsec:sa} below.  Choosing $B=B_L$ or $B=B_R$ then allows us to define four algebras $A_L^{B_L}, A_L^{B_R}, A_R^{B_L}, A_R^{B_R}$.  However, as we will see, only the two algebras $A_L^{B_L}, A_R^{B_R}$ have natural actions on $H_{B_L \sqcup B_R}$.  For these algebras, we obtain representations on $H_{B_L \sqcup B_R}$. Our path integral also defines useful notions of trace on each of these algebras.

We will then show in section \ref{subsec:Rep} that the associated representations of the algebras $A_L^{B_L}, A_R^{B_R}$ can be extended from the pre-Hilbert space $H_{B_L \sqcup B_R}$ to its Hilbert space completion ${\cal H}_{B_L \sqcup B_R}$.  In particular, because our axioms turn out to enforce a trace inequality of the form recently discussed in \cite{Colafranceschi:2023txs}, all operators in these representations must be bounded.  As a result, we may use the associated representations to construct von Neumann algebras.

We then specialize to the case $B_L= B_R = B$, in which case we denote the resulting von Neumann algebras by ${\cal A}^B_L$ and ${\cal A}^B_R$.  Some key properties of these algebras are then studied in section \ref{subsec:prop}.  In particular, we show there that the above trace operations can be extended to both ${\cal A}^B_L$ and ${\cal A}^B_R$.

\subsection{Surface algebras}
\label{subsec:sa}

Consider for the moment a given compact closed boundary $B$ (with $\partial B = \emptyset$), which might represent either $B_L$ or $B_R$ above. For each such $B$ we will define two algebras,  $A^B_{L}$  and $A^B_{R}$.  If we thus allow a choice of $B= B_L$ or $B= B_R$, we could in fact define {\it four} such algebras, though only the two choices
 $A^{B_L}_{L}$  and $A^{B_R}_{R}$ will play an important role in our construction below.

Let us now consider general such $B$. To understand the difference between $A^B_{L}$  and $A^B_{R}$, recall that points on ${B \sqcup B}$ are labelled, which in particular means that the two copies of $B$ are distinguished. We will refer to the first copy as the `left boundary' and the second copy as the `right boundary.'

On the set $Y^d_{B \sqcup B}$ we may define the \textit{left product} $(\cdot_L)$ as the operation that takes as input an ordered pair of rimmed surfaces $a$ and $b$, and which constructs the surface $a \cdot_{L} b$ that results from gluing the \textit{left} boundary of $b$ to the right boundary of $a$ (see figure \ref{fig:LRprod}).

\begin{figure}[t]
        \centering
\includegraphics[width=\linewidth]{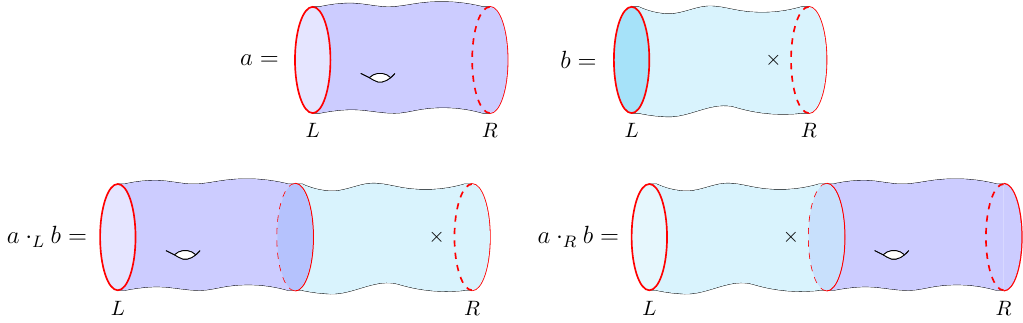}\caption{For two elements $a,b\in Y^d_{B \sqcup B}$ as shown in the top row, we define the left and right products $a \cdot_{L} b$ and $a \cdot_{R} b$ by the gluing procedures shown in the bottom row.}
\label{fig:LRprod}
\end{figure}
For simplicity, we will adopt the notation
\begin{equation}
ab \coloneqq a \cdot_{L} b \ \ \ .
\end{equation}
We similarly define the \textit{right product} $(\cdot_R)$ as the operation that, given an ordered pair of surfaces $a$ and $b$, glues the \textit{right} boundary of $b$ to the left boundary of $a$. Note that we have
\begin{equation}\label{Lprodct}
   a \cdot_{R} b = b \cdot_{L} a = ba \ .
\end{equation}

We can also extend this product to linear combinations $a, b\in \U{Y}^d_{B \sqcup B}$ by defining it to satisfy the distributive law.
The set $\U{Y}^d_{B \sqcup B}$ equipped with the left product then forms an algebra $A^B_{L}$ which we call the \textit{left $B$-surface algebra}, or simply the left surface algebra where confusion will not arise.
Similarly, the right product on $\U{Y}^d_{B \sqcup B}$ leads to the \textit{right $B$-surface algebra} $A^B_{R}$.
Since every element of $\U{Y}^d_{B \sqcup B}$ has a finite rim at each boundary, gluing two surfaces $a,b$ together always results in a surface larger than either $a$ or $b$, so that neither of these algebras can contain an identity element.

However, the algebras $A^B_{L}$ and $A^B_{R}$ do admit a natural involution $\star$ satisfying
\begin{equation}
\label{eq:starop}
    (a \cdot_L b)^\star= b^\star \cdot_L a^\star =a^\star \cdot_R b^\star,
\end{equation}
so that $\star$ defines an anti-linear isomorphism between the left and right algebras.
To define the operation $\star$, recall that Axiom \ref{ax:RP} introduced a complex conjugation operation $*$ (which is different from the $\star$ that we are about to define) on  $N \in Y^d_{B \sqcup B}$.
In particular, $N^*$ was defined so that $M_{N^*N}$ has a reflection symmetry that complex conjugates all sources.  This means that  $N^*$ is the same manifold as $N$ (with the same labels on $\partial N$), and that $*$ acts on scalar sources by standard complex-conjugation (though the operation on vector,  tensor, and spinor sources is more complicated due to the reflection).
In addition, $Y^d_{B \sqcup B}$ admits a natural transpose operation ${}^t$ that simply swaps the labels `left' and `right' attached to the boundaries of any $N \in Y^d_{B \sqcup B}$ while  preserving all sources and leaving the labels on $\partial N$ otherwise unchanged.  The transpose and complex conjugation operations commute, and for any $a$ in either algebra we may then define
\begin{equation}
a^\star : = (a^*)^t.
\end{equation}
Due to the inclusion of the transpose operation, we then immediately find \eqref{eq:starop}.

\subsection{A trace and a trace inequality for surface algebras}
\label{subsec:tr1}

An important consequence of the labelling of points on $B$ is that,
by writing $\partial N = B \sqcup B$, we also mean that the labels on the two copies of $B$ agree up to the distinction between the left and right boundaries.  To be precise, we mean that these labels define a diffeomorphism $\phi_{LR}$ from the left boundary to the right boundary that preserves enough information about sources near each boundary to reconstruct sufficiently small rims at each $B$.  This $\phi_{LR}$ can then be used to identify the left boundary of any $a\in Y^d_{B \sqcup B}$ with its right boundary, and thus to define a closed source manifold\footnote{Though there are similarities, this is a different gluing operation than the one used to construct $M_{N_1^*N_2}$, so we use a correspondingly similar-but-different notation $M(a)$; see e.g.\ \eqref{eq:mrelation}.}. $M(a) \in X^d$ (i.e., without boundary) from any $a\in  Y^d_{B \sqcup B}$.  We can also extend this operation to linear combinations ${a}\in \U{Y}^d_{B \sqcup B}$ by linearity, so that we then find $M(a) \in \U{X}^d$.

This observation allows the path integral to define a useful trace operation $\tr$ on both $A^B_{L}$ and $A^B_{R}$ for any $B$.  This trace associates to any $a\in \U Y^d_{B \sqcup B}$ the finite number
\begin{equation}
\text{tr}(a) :=  \zeta\left(M(a)\right).
\end{equation}
Note that since $M({ab}) = M({ba})$, we clearly have
\begin{equation}
\label{eq:cyclic}
\tr \, ab=\tr \, ba,
\end{equation}
and similarly for the right product; see figure \ref{fig:refpos}.  While the trace operation is defined directly for any $a\in \U Y^d_{B \sqcup B}$ (without using properties of either the left or right algebras), the result \eqref{eq:cyclic} makes it reasonable to refer to this operation as a trace on both $A^B_{L}$ and $A^B_{R}$.

\begin{figure}[t]
        \centering
\includegraphics[width=0.8\linewidth]{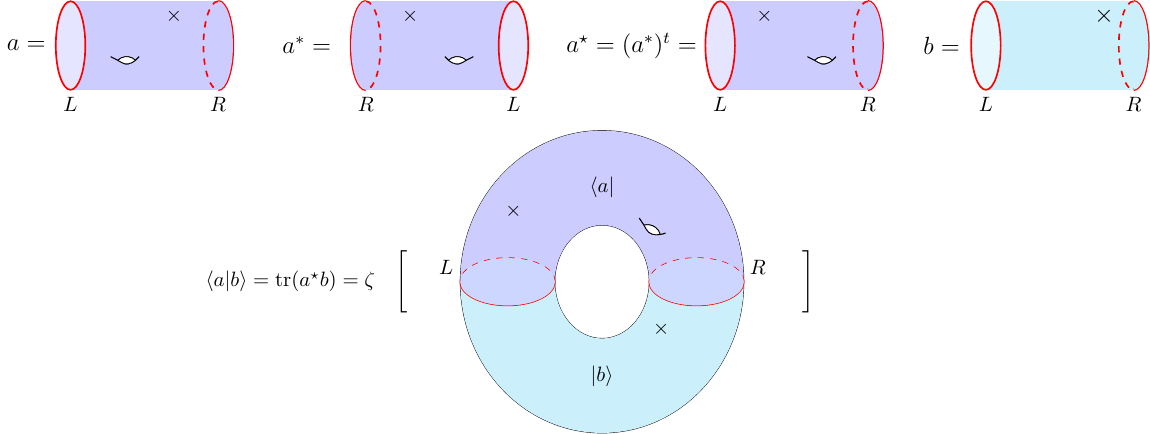}
\caption{ For $a,b \in Y^d_{B \sqcup B}$, we can construct $a^\star$ and compute $\tr(a^\star b)$ as shown in the lower panel.   Note that $\tr(a^\star b) = \tr(b a^\star)$, and that this relation is equivalent to \eqref{eq:cyclic} with  $a$ replaced by $a^\star$.}
\label{fig:refpos}
\end{figure}

Before proceeding to the next step of our analysis, it will be useful to note that, as also shown in figure \ref{fig:refpos}, for $a, b \in \U Y^d_{B \sqcup B}$ we have
\begin{equation}
\label{eq:mrelation}
M_{a^* b} = M({a^\star b}),
\end{equation}
and thus using \eqref{eq:ip} we find
\begin{equation}
\label{eq:iptr}
\langle a | b \rangle = \zeta\left(M({a^\star b})\right) = \text{tr}(a^\star b).
\end{equation}
 This relation will be used to translate certain Hilbert space statements into operator statements and vice versa.
In particular, for $a=b \in \U{Y}^d_{B \sqcup B}$ we have
\begin{equation}
\label{eq:iptr2}
\text{tr}(a^\star a) =  \zeta\left(M({a^\star a})\right) =  \langle a | a \rangle \ge 0,
\end{equation}
where we remind the reader that the inequality on the right follows from Axiom \ref{ax:RP} (reflection positivity).

The inequality \eqref{eq:iptr2} turns out to have interesting and important consequences when we apply it to Hilbert space sectors defined by boundaries with multiple connected components.  We will come back to this idea several times during our work, but we start with a simple case that involves choosing two sums-of-surfaces $a, b \in \U{Y}^d_{B \sqcup B}$. Note that such sums can be used to define an element $a \sqcup b \in \U{Y}^d_{(B \sqcup B) \sqcup (B \sqcup B)}$.  Here the parentheses in the subscript on
$\U{Y}^d_{(B \sqcup B) \sqcup (B \sqcup B)}$ are intended to indicate that points of $\partial a$ are labelled to match the first pair of boundaries $(B \sqcup B)$, while points of $\partial b$ are labelled to match the second pair of boundaries $(B \sqcup B)$.   We might also write such an $a \sqcup b$ using the more explicit notation $a_{L_1 R_1} \sqcup b_{L_2 R_2} \in \U{Y}^d_{B_{L_1},B_{R_1},B_{L_2},B_{R_2}}$, which indicates that the left and right boundaries of $a$ are associated with the first two copies of $B$ (in the specified order) while the boundaries of $b$ are associated with the second two copies of $B$. In writing $\U{Y}^d_{B_{L_1},B_{R_1},B_{L_2},B_{R_1}}$ we have replaced the usual disjoint union symbols $\sqcup$ by commas for notational simplicity.  At this point we also generalize the setting to allow  $a \in \U{Y}^d_{B_{L_1} \sqcup B_{R_1}}$, $b \in \U{Y}^d_{B_{L_2} \sqcup B_{R_2}}$ for arbitrary $B_{R_1}, B_{R_2}$, though we retain the constraint $B_{L_1} = B_{L_2}$.

Note in particular that there are distinct states $|a_{L_1 R_1}, b_{L_2 R_2} \rangle$ and $|a_{L_2 R_1}, b_{L_1 R_2}\rangle$ in the pre-Hilbert space $H_{B_{L_1},B_{R_1},B_{L_2},B_{R_1}}$, where the notation is analogous to that used for
$\U{Y}^d_{B_{L_1},B_{R_1},B_{L_2},B_{R_2}}$ above; see the lower panel of figure \ref{fig:4bndy}.  These states are related by the action of the
`swap' operator ${\cal S}_{L_1, L_2}$ that exchanges the labels $L_1,L_2$ on the relevant two copies of $B$.
Since the (pre-)inner product on $H_{B_{L_1},B_{R_1},B_{L_2},B_{R_2}}$ is defined by identifying corresponding boundaries in the bra- and ket-surfaces, as shown in figure \ref{fig:4bndy}
the norms of these states are computed by path integrals defined by the disconnected source manifold $M(a^\star a) \sqcup M(b^\star b)$. We thus find
\begin{equation}
\label{eq:4bndynorms}
\langle a_{L_1 R_1}, b_{L_2 R_2} | a_{L_1 R_1}, b_{L_2 R_2} \rangle =
\langle a_{L_2 R_1}, b_{L_1 R_2} | a_{L_2 R_1}, b_{L_1 R_2} \rangle =
\langle a | a \rangle \langle b | b \rangle = \tr(a^\star a) \tr(b^\star b).
\end{equation}
In contrast, the (pre-)inner product between $|a_{L_1 R_1}, b_{L_2 R_2} \rangle$ and $|a_{L_2 R_1}, b_{L_1 R_2}\rangle$  is computed by a path integral having the connected source manifold $M(a^\star b b^\star a) = M(a a^\star b b^\star)$.  Thus we have
\begin{eqnarray}
\label{eq:4bndyip}
\langle a_{L_1 R_1}, b_{L_2 R_2} | a_{L_2 R_1}, b_{L_1 R_2} \rangle &=& \langle b^\star a  | b^\star a \rangle = \tr ( a^\star b b^\star a) = \tr (a a^\star b b^\star ) \cr &=&
\langle a_{L_2 R_1}, b_{L_1 R_2} | a_{L_1 R_1}, b_{L_2 R_2} \rangle.
\end{eqnarray}
Here the first equality shows that \eqref{eq:4bndyip} is real and non-negative (since the second form shown is a norm squared), which then implies the final equality (since the final form is manifestly the complex conjugate of the left-hand side).

\begin{figure}[ht!]
        \centering
\includegraphics[width=\linewidth]{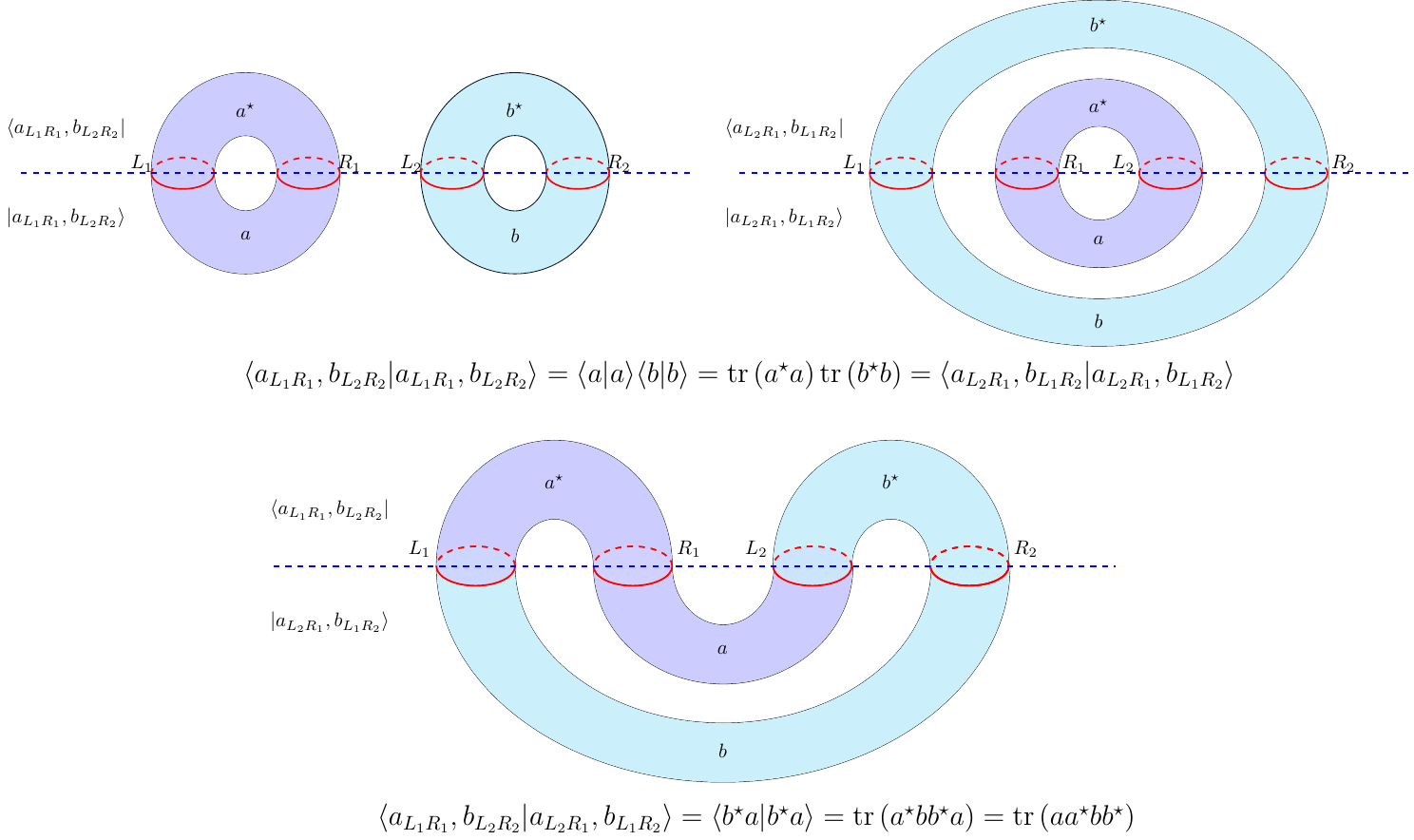}\caption{
{\bf Upper panel:}  Path integrals that compute the norm squared $\langle a_{L_1 R_1}, b_{L_2 R_2} | a_{L_1 R_1}, b_{L_2 R_2} \rangle$ (left) and $\langle a_{L_2 R_1}, b_{L_1 R_2} | a_{L_2 R_1}, b_{L_1 R_2} \rangle$ (right).  As shown, the two norms agree.  {\bf Lower panel:}  The path integral that computes
$\langle a_{L_1 R_1}, b_{L_2 R_2} | a_{L_2 R_1}, b_{L_1 R_2}\rangle$.  Tracing the diagram verifies that this inner product is equal to both $\tr(a^\star b b^\star a) = \tr (a a^\star b b^\star)$ and $\langle  b^\star a|b^\star a\rangle$.  Since the latter is manifestly real (and non-negative), this diagram also computes $\langle a_{L_1 R_1}, b_{L_2 R_2} | a_{L_2 R_1}, b_{L_1 R_2}\rangle^* = \langle   a_{L_2 R_1}, b_{L_1 R_2}| a_{L_1 R_1}, b_{L_2 R_2} \rangle$.}
\label{fig:4bndy}
\end{figure}

On the other hand, the Cauchy-Schwarz inequality requires
\begin{equation}
\label{eq:Cauchy1234}
\Big|\langle a_{L_1 R_1}, b_{L_2 R_2} | a_{L_2 R_1}, b_{L_1 R_2} \rangle\Big| \le
\Big| | a_{L_1 R_1}, b_{L_2 R_2} \rangle\Big| \,\,
\Big| | a_{L_2 R_1}, b_{L_1 R_2} \rangle \Big|,
\end{equation}
where $\big|\ket\psi\big|:= \sqrt{\braket{\psi}{\psi}}$ is the norm.
Combining \eqref{eq:4bndynorms} and \eqref{eq:4bndyip} with \eqref{eq:Cauchy1234} immediately yields
\begin{equation}
\label{eq:TrIn1}
\tr (a^\star b b^\star a) = \tr (a a^\star b b^\star) \le \tr(a^\star a) \tr(b^\star b),
\end{equation}
which is the trace inequality recently discussed in \cite{Colafranceschi:2023txs}.  Here we see that it holds for our surface algebras as a consequence of Axioms \ref{ax:finite}-\ref{ax:factorize}.    The inequality \eqref{eq:Cauchy1234} will play a key role in our discussion below.  We will also return to higher analogues of this inequality in section \ref{subsec:normoftrace}.

\subsection{Representation of the surface algebras on  \texorpdfstring{$\mathcal{H}_{B_L \sqcup B_R}$}{}}
\label{subsec:Rep}

So far we have defined the surface algebras as abstract vector spaces equipped with multiplication; now we define how they act on states. In particular, the operations on surfaces described above can now be used to define a representation $\hat{A}^{B_L \sqcup B_R}_{L}$ of the surface algebra $A^{B_L}_{L}$  that acts on the Hilbert space $\mathcal{H}_{B_L \sqcup B_R}$. In the notation of sections \ref{subsec:sa}-\ref{subsec:tr1} we now specialize to the case $B= B_L$.  Though the notation $\hat{A}^{B_L \sqcup B_R}_{L}$ is somewhat awkward, it emphasizes the important point that this representation depends {\it both} on the choice of $B_L$ and on the choice of $B_R$, even though $A^{B_L}_{L}$ was defined by $B_L$ alone. This point will be discussed further in section \ref{subsec:diag}.

The first steps of our construction are to consider $a \in \U{Y}^d_{B_L \sqcup B_L}$ and to define an associated operator $\hat{a}_L$ that acts on the pre-Hilbert space $H_{B_L \sqcup B_R}$ as
\begin{equation}
    \hat{a}_L \ket{b} =\ket{ab}, \quad \forall \ket{b} \in H_{B_L \sqcup B_R};
\end{equation}
see figure \ref{LRact}. Here we have used the condensed notation $ab \coloneqq a \cdot_L b$ defined in section \ref{subsec:sa} above. When $a$ is a simple surface $a\in Y^d_{B_L \sqcup B_L}$, the associated $\hat{a}_L$ acts on $|b\rangle \in {H}_{B_L \sqcup B_R}$ by just gluing the surface $a$ to the left boundary of $b$. Note that this action is a representation of the left surface algebra $A^{B_L}_{L}$ (on $H_{B_L \sqcup B_R}$ at the moment), because
\begin{equation}
\label{eq:rep}
     \widehat{(ab)}_L \ket{c} =\ket{a b c} = \hat{a}_L (\hat{b}_L\ket{c}),
\end{equation}
where $\widehat{(ab)}_L$ is the operator associated to $(ab)$.

\begin{figure}[ht!]
        \centering
\includegraphics[width=0.7\linewidth]{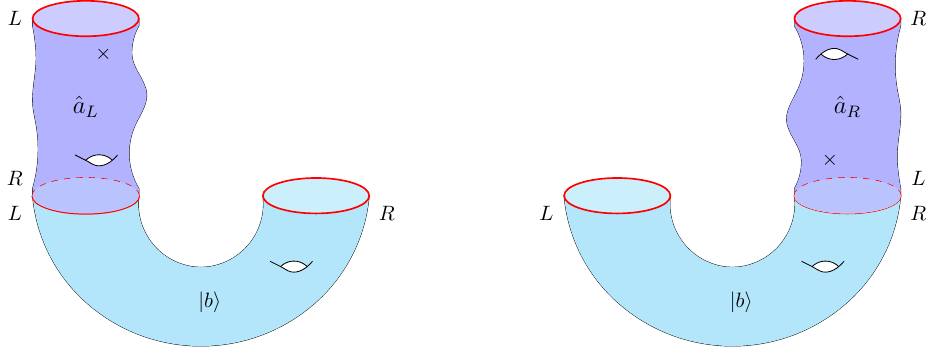}\caption{For surfaces $a,b$, the left panel shows the action of $\hat a_L$ (in the representation of the left algebra $A^{B_L}_{L}$) on $|b \rangle \in H_{B_L \sqcup B_R}$, while the right panel shows the associated action of $\hat a_R$ (in the representation of the right algebra $A^{B_R}_{R}$).  The figure in fact depicts the special case $B_L=B_R$ in which a given element $a\in Y_{B_L \sqcup B_L} =  Y_{B_R \sqcup B_R}$ defines members of both algebras, but this is not generic.}
\label{LRact}
\end{figure}

The next step in showing that our representation acts on the Hilbert space ${\cal H}_{B_L \sqcup B_R}$ is to establish that $\hat{a}_L$ preserves any null space ${\cal N}_{B_L \sqcup B_R}$ of pre-Hilbert space states with vanishing norm, so that $\hat{a}_L$ yields a well-defined operator on $H_{B_L \sqcup B_R}/{\cal N}_{B_L \sqcup B_R}$.  We will also need to extend the definition of $\hat{a}_L$ to the full Hilbert space ${\cal H}_{B_L \sqcup B_R}$ in a manner consistent with \eqref{eq:rep}.  Both of these steps are straightforward due to the trace inequality.  To see this, recall that
 for $a \in \U{Y}^d_{B_L \sqcup B_L}$, $b \in \U{Y}^d_{B_L \sqcup B_R}$ we have $\hat{a}_L|b\rangle = |ab\rangle$, and thus
\begin{equation}
\label{eq:bound}
|\hat{a}_L|b\rangle|^2 = \langle ab|ab\rangle = \tr(b^\star a^\star ab) = \tr(a^\star a bb^\star) \le \tr(a^\star a) \tr(bb^\star ) = \tr(a^\star a) \langle b|b\rangle.
\end{equation}
In the second step of \eqref{eq:bound} we have used \eqref{eq:iptr} with $a$ and $b$ both replaced by $ab$.  The third step used cyclicity of the trace \eqref{eq:cyclic}, and the fourth and fifth steps then follow directly from \eqref{eq:TrIn1} and another use of \eqref{eq:iptr}.  The result is that $\hat a_L$ is bounded by $\sqrt{\tr \, {a^\star a}}$ on $H_{B_L \sqcup B_R}$.

In particular, if $|b\rangle \in {\cal N}_{B_L \sqcup B_R}$ then $\langle b | b \rangle =0$.  The result \eqref{eq:bound} then clearly requires $\hat{a}_L|b\rangle$ to have zero norm as well.  Thus $\hat{a}_L$ preserves ${\cal N}_{B_L \sqcup B_R}$ and induces an operator on the quotient ${H}_{B_L \sqcup B_R}/{\cal N}_{B_L \sqcup B_R}$.  Applying \eqref{eq:bound} on this quotient tells us that the operator $\hat{a}_L$ is bounded by $\sqrt{\tr \, {a^\star a}}$ on the dense subspace ${H}_{B_L \sqcup B_R}/{\cal N}_{B_L \sqcup B_R}$ of ${\cal H}_{B_L \sqcup B_R}$.   It thus admits a unique continuous extension to the entire space ${\cal H}_{B_L \sqcup B_R}$, which is again bounded by $\sqrt{\tr \, {a^\star a}}$; see e.g.\ \cite{RS}.  We will continue to use the symbol
$\hat{a}_L$  for this extension; thus we write the bound on its operator norm as
\begin{equation}
\label{eq:anormbound}
\| \hat{a}_L \| \le \sqrt{\tr \, {a^\star a}}.
\end{equation}
Continuity implies that such extensions also satisfy \eqref{eq:rep}, which makes clear that we have constructed a representation of $A^{B_L}_{L}$ on $\mathcal{H}_{B_L \sqcup B_R}$ as desired.    We call this representation  $\hat{A}^{B_L \sqcup B_R}_{L}$ or, in what we hope is an obvious shorthand, we say that we have constructed a representation  $\hat A_L :=\hat{A}^{B_L \sqcup B_R}_{L}$.

Since the operators in $\hat A_L$ are bounded, it is easy to discuss their adjoints (which must exist and are also bounded).  Any such operator is defined by some $a\in \U{Y}^d_{B_L \sqcup B_L}$, and for $b,c \in \U{Y}^d_{B_L \sqcup B_R}$ we must have
\begin{equation}
\label{eq:aadj}
\langle b | \hat{a}_L^\dagger| c \rangle = (\langle c | \hat{a}_L | b \rangle)^* = [\tr\left(c^\star  a b \right)]^* = \tr \left(b^\star  a^\star   c \right) = \langle b| \widehat{a^\star }_L | c\rangle.
\end{equation}
Here the 3rd step follows from \eqref{eq:iptr} and the relation $(\langle c|d\rangle)^* = \langle d|c\rangle$ with $d=ab$.  Since \eqref{eq:aadj} holds on a dense set of states, and since both $\widehat{a^\star}_L$ and $\hat a^\dagger_L$ are bounded, we must in fact have $\widehat{a^\star}_L = a^\dagger_L$ on all of ${\cal H}_{B_L \sqcup B_R}$.

It is worth noting that the representation $\hat A_L$ of $A^{B_L}_{L}$ is not necessarily faithful. This is precisely characterized by a null space $\mathfrak{N}_L$ consisting of all $a \in A^{B_L}_{L}$ whose associated $\hat a_L$ is the zero operator. Thus, when we construct $\hat A_L$ from $A^{B_L}_{L}$, we have effectively taken a quotient $A^{B_L}_{L} /\mathfrak{N}_L$, in the sense that $\hat A_L$ is isomorphic to (and a faithful representation of) $A^{B_L}_{L} / \mathfrak{N}_L$.

The right surface algebra $A^{B_R}_{R}$ admits a similar representation $\hat A_R := \hat{A}^{B_L \sqcup B_R}_{R}$ on $\mathcal{H}_{B_L \sqcup B_R}$ defined by gluing surfaces to the \textit{right} boundary of $H_{B_L \sqcup B_R}$.  For $a \in \U{Y}^d_{B_R \sqcup B_R}$  and $ b \in \U{Y}^d_{B_L \sqcup B_R}$ we have
\begin{equation}
    \hat{a}_R \ket{b}=\ket{a \cdot_R b} =\ket{b \cdot_L a} =\ket{ba}, \quad \forall \ket{b} \in H_{B_L \sqcup B_R}.
\end{equation}
As required for a representation of $A_{R}$, this action satisfies
\begin{equation}
    {\widehat{(b \cdot_R  c)}}_{R} \ket{a}  =\ket{(b \cdot_R c)\cdot_R  a}= \ket{b \cdot_R (c \cdot_R a)} = \hat{b}_R (\hat{c}_R\ket{a}),
\end{equation}
where ${\widehat{(b \cdot_R c)}}_{R}$ is the operator associated to $(b \cdot_R  c) \in A^{B_R}_{R}$.
The extension to the full Hilbert space $\mathcal{H}_{B_L \sqcup B_R}$ then proceeds precisely as above.  The discussion of adjoints is analogous to the left case and we again find
$\widehat{a^\star}_R = a^\dagger_R$.

Perhaps the most interesting point to mention concerning $\hat A_R$ is that its operators commute with operators in the left-representation $\hat A_L$.  In particular, for $a \in A^{B_R}_R$, $b \in A^{B_L}_L$, and $c$ in the pre-Hilbert space $H_{B_L \sqcup B_R}$ we clearly have
\begin{equation}
\label{eq:LRcomm}
\hat a_R \hat b_L |c\rangle = |bca\rangle = \hat b_L \hat a_R |c\rangle.
\end{equation}
Furthermore, the operators $\hat a_R$, $\hat b_L$
are bounded (and thus continuous) on ${\cal H}_{B_L \sqcup B_R}$ and the above states $|c\rangle$ are dense in ${\cal H}_{B_L \sqcup B_R}$.  We may thus take limits to conclude that
 \eqref{eq:LRcomm} in fact holds for all $|c\rangle \in {\cal H}_{B_L \sqcup B_R}$.

\subsection{Diagonal sectors are special}
\label{subsec:diag}

We now restrict to the ``diagonal" special  case $B_L=B_R$, in which context we introduce the shorthand notation ${\cal H}_{LR} = {\cal H}_{B_L \sqcup B_R}$ and similarly for $\U{Y}^d_{LR}$, the pre-Hilbert space $H_{LR}$, and so forth.
It will also be useful to note that our trace operation $\tr$ on $A^{B_L}_{L}$ defines a trace on the representation $\hat{A}_{L}$.  We of course wish to declare that
\begin{equation}
\label{eq:trrespectsquot}
\tr \, \hat{a}_L : = \tr \, a.
\end{equation}
The important property of the definition \eqref{eq:trrespectsquot} (which we show below) is that it is well-defined on  $\hat{A}_{L}$ in the sense that it satisfies $\tr \, a = \tr \, b$ whenever $\hat{a}_L = \hat{b}_L$ on ${\cal H}_{LR}$.  This is equivalent to saying that $\tr \, a =0$ when $\hat{a}_L =0$ on ${\cal H}_{LR}$. Note that this property is non-trivial because the representation $\hat A_L$ is not necessarily faithful; i.e., when $\hat{a}_L=0$, the element $a$ could be non-zero, but nonetheless we claim that $\tr\, a =0$. Note also that this property relies on the action on a particular Hilbert space, so the condition that an analogous trace be well-defined on the representation of the left-algebra on some other ${\cal H}_{B_L \sqcup B_R}$ is generally quite distinct.    In particular, we make no claim that the trace is well-defined on representations defined by ${\cal H}_{B_L \sqcup B_R}$  with $B_L \neq B_R$.

For a diagonal Hilbert space  ${\cal H}_{LR}$ with $B_L=B_R$, the desired property can be established using the continuity axiom \ref{ax:continuity}.  In particular, for any $a \in \U{Y}^d_{LR}$, let $C_\beta \in \U{Y}^d_{LR}$ be the cylinder of length $\beta$ defined by $B_L=B_R$ as in Axiom \ref{ax:continuity}, and consider $aC_{2\beta} \in \U{Y}^d_{LR}$.  Since $C_{2\beta} = C_\beta C_\beta$, and since we have restricted to $B_L=B_R$ for which cylinders are self-adjoint ($C_\beta^\star = C_\beta$), we have
\begin{equation}
\label{eq:alttrace}
\tr \left(a C_{2\beta} \right)  =  \tr \left(C_\beta a C_\beta \right) = \langle C_\beta | \hat{a}_L | C_\beta \rangle.
\end{equation}
Clearly the right-hand side vanishes for all $\beta$ if $\hat{a}_L =0$.  However, Axiom \ref{ax:continuity} requires the $\beta \rightarrow 0$ limit of \eqref{eq:alttrace} to give $\tr (a)$:
\begin{equation}
\label{eq:alttrace2}
\tr(a) = \lim_{\beta \downarrow 0}  \langle C_\beta | \hat{a}_L | C_\beta\rangle,
\end{equation}
where the notation $\beta \downarrow 0$ emphasizes that $C_\beta$ is defined only for $\beta >0$ so that the limit is necessarily taken from above.
Thus, as desired, we find $\tr(a)=0$ when $\hat{a}_L =0$.  The trace is of course also defined on $\hat A_R$ where its properties are analogous.

We can also easily establish the converse of the above property for manifestly positive\footnote{Throughout the paper, we call an operator positive if it is positive semi-definite.} operators $\hat a_L$ of the form $\hat a_L=\hat \gamma_L^\dagger \hat \gamma_L$ with  $\hat\gamma_L \in \hat A_L$; i.e., if $\tr(a)=0$ then $\hat a_L=0$.  This follows immediately from the operator norm bound \eqref{eq:anormbound}: if $0=\tr(a)=\tr(\hat \gamma_L^\dagger \hat \gamma_L) =\tr \big({\widehat{(\gamma^\star \gamma)}}_L\big) =\tr(\gamma^\star \gamma)$, then the operator norm of $\hat \gamma_L$ is bounded by $\sqrt{\tr(\gamma^\star \gamma)} =0$ so that $\hat \gamma_L=0$ and thus $\hat a_L=0$.  In a slight abuse of terminology, we will refer to this property by saying that our trace is {\it faithful} on the representation $\hat A_L$.  The usual definition of the term faithful would require this property to hold for all positive operators (and not just those whose positivity is manifest).  This stronger notion of faithfulness will also turn out to hold, though we defer its discussion to section \ref{subsec:exTrIn}.

The argument above used the fact that ${\cal H}_{B_L \sqcup B_R}$ for $B_L=B_R$ contains cylinder states $|C_\beta\rangle$ of the form that enter into the continuity axiom.  We have no analogous construction for $B_L \neq B_R$, so our trace may not be well-defined on other representations.  At the level of the full von Neumann algebra defined below, it seems natural for the associated left algebras defined by $B_L \neq B_R$ to be projections of the left algebra defined by the diagonal case $B_L= B_R$ (and similarly for the right algebras).  However, we reserve this analysis for future work \cite{Marolf:2024adj}.

\subsection{The von Neumann algebras \texorpdfstring{$\mathcal{A}_{L}^B$}{} and \texorpdfstring{$\mathcal{A}_{R}^B$}{}}
\label{subsec:prop}

We are now ready to define a von Neumann algebra $\mathcal{A}^B_{L}$  using the representation $\hat A_L : = \hat{A}^{B\sqcup B}_{L}$ on a diagonal sector of the Hilbert space with $B_L=B_R=B$ (and to similarly construct a related von Neumann algebra $\mathcal{A}^B_{R}$ from the right representation $\hat A_R : =\hat{A}^{B\sqcup B}_{R}$). Although it is straightforward to define analogous von Neumann algebras using non-diagonal representations $\hat{A}^{B_L \sqcup B_R}_{L}$ with $B_L \neq B_R$, we leave the investigation of such algebras for future work; see again the comment at the end of the previous subsection.

We define the von Neumann algebra $\mathcal{A}^B_{L}$ to be the closure of $\hat{A}_{L}$ within ${\cal B}({\cal H}_{LR})$ in the weak operator topology or, in what is known to be equivalent in the present context, the strong operator topology\footnote{\label{foot:tops} The weak operator topology means that for a net of operators $T_\alpha$, $T_\alpha \rightarrow T$ if and only if $\left\langle x\left|T_\alpha\right| y\right\rangle \rightarrow\langle x|T| y\rangle$ for every $|x\rangle ,|y\rangle \in \mathcal{H}$.  In contrast,
the strong operator topology means that $T_\alpha \rightarrow T$ if and only if $T_\alpha | y \rangle \rightarrow T| y\rangle$ for every $|y\rangle \in \mathcal{H}$.
The two associated notions of closure agree for convex sets of bounded operators; see e.g.\ theorem 5.1.2 of \cite{KR1}. These notions of closure are each equivalent to including limits of all nets of operators.  However, it appears that these closures are not equivalent to merely including limits of all sequences.  In particular,
exercise 1 from section II.2 of \cite{Takesaki} states that the weak and strong operator topologies are not metrizable for operators on infinite-dimensional Hilbert spaces. This suggests that these topologies are also non-sequential (meaning that sequences do not suffice to characterize them) and, indeed the structure of the exercise suggests that this non-sequential nature would be made explicit from a complete solution.  Although this point is not critical for our work here, we would be very happy to receive references that address the issue more explicitly.}. Here ${\cal B}({\cal H}_{LR})$ denotes the algebra of bounded operators on our Hilbert space. Note that the identity operator $\mathbb{1}$ lies in the closure due to Corollary \ref{cor:Cto1} of appendix \ref{app:lemmas}. Due to the von Neumann bicommutant theorem (see e.g.\ section $0.4$ of \cite{Sunder}), we can also equivalently define $\mathcal{A}^B_{L}$ as the double commutant of $\hat{A}_{L}$ within ${\cal B}({\cal H}_{LR})$.  This in particular means that each operator in the resulting von Neumann algebra $\mathcal{A}^B_{L}$ is again bounded. Of course, corresponding statements hold for the right algebras as usual.

For every operator $a$ in a von Neumann algebra, the adjoint $a^\dagger$ also lies in the von Neumann algebra.  So the adjoint operation continues to act as an involution on
$\mathcal{A}^B_{L}$. Note that we previously used the symbols  $a,b,c,\dots$ to denote elements of $\U{Y}^d_{B \sqcup B}$, but that we henceforth also use them to denote generic operators in $\mathcal{A}^B_{L/R}$. We will continue to use $\hat a_L, \hat b_L, \hat c_L, \dots$ to denote operators in $\hat A_L$.

We also introduced a trace operation $\tr$ on the operators in $\hat{A}_{L}$  in \eqref{eq:trrespectsquot}.  In particular, we showed $\tr$ to be well-defined and finite on $\hat{A}_{L}$. We now wish to extend this trace to the $\mathcal{A}^B_{L}$.  In the theory of von Neumann algebras one generally allows traces of some operators to diverge.  Nevertheless, even in this sense, a trace is usually well-defined only on positive elements of the von Neumann algebra, where it takes values in the closed interval $[0, +\infty]$; i.e., allowing $+\infty$.  The restriction to positive elements is closely related to the familiar fact that,
 when an infinite-dimensional square matrix $A^i_j$ is not positive, the infinite sum of the form $\sum_i A^i_i$  can be oscillatory and need not converge in any sense.  In contrast, for positive infinite-dimensional matrices $A^i_j$,
the fact that each diagonal element $A^i_i$ is non-negative means that if $\sum_i A^i_i$ fails to converge to a finite number, then we may say that it `converges' to $+\infty$.  (Of course, the quantity $\sum_i A^i_i$ is manifestly well-defined
for any finite-dimensional square matrix $A^i_j$.)

We will thus attempt to extend our notion of $\tr$ only to positive elements $a \in \mathcal{A}^B_{L}$, which in this context means that $a$ is a positive operator on ${\cal H}_{LR}$.  However, we note for future reference that this condition is equivalent to requiring that $a$ be of the form $\gamma^\dagger \gamma$ for some $\gamma \in \mathcal{A}^B_{L}$ (where we can in fact take $\gamma$ to be the positive square root of $a$, as this operator must also lie in $\mathcal{A}^B_{L}$).

To define a useful extension of our trace, we need to find a function mapping the positive elements $a \in \mathcal{A}^B_{L}$ to $[0, +\infty]$ that agrees with our previous definition of $\tr$ on $\hat{A}_{L}$ and which satisfies other properties to be discussed below. It will thus be productive to consider alternative representations of the operation $\tr$ on $\hat{A}_{L}$.  We begin by  returning to the relation \eqref{eq:alttrace2}, which was argued above to hold for all $a \in \U{Y}^d_{LR}$.   This will turn out to be a step toward the definition of our trace on ${\cal A}^B_L$, though we will now pause briefly to further rewrite the identity \eqref{eq:alttrace2} in order to make certain properties manifest.

It will be convenient to introduce the normalized cylinders $\tilde C_\beta \in \U{Y}^d_{LR}$ defined by
\begin{equation}
\label{eq:tC}
\tilde{C}_\beta := C_\beta/\|C_\beta\|,
\end{equation}
where $\|C_\beta\|$ denotes the operator norm of ${\widehat{C_\beta}_L}$ on ${\cal H}_{LR}$.  This norm should be more properly written $\|\widehat{C_\beta}_L\|$, but for simplicity we will use just $\|C_\beta\|$.   One may expect that the continuity axiom (Axiom \ref{ax:continuity}) requires $\|C_\beta \| \rightarrow 1$ as $\beta \rightarrow 0$ and in fact that $\|C_\beta\| = \left(\|C_1\|\right)^\beta$.  Both expectations are correct, but the proofs are somewhat technical. We thus relegate them to appendix \ref{app:lemmas}. As a further remark, note that $\tilde C_\beta$ is normalized so that $\widehat{\tilde C_\beta}_L$ has operator norm $1$, but that the state $|\tilde C_\beta\rangle$ is typically still {\it not} normalized with respect to the Hilbert space inner product.  In fact, the norm of $|\tilde C_\beta\rangle$ generally diverges as $\beta \rightarrow 0$; see \eqref{eq:traceid}.

For $a \in \U{Y}^d_{LR}$,  we may use \eqref{eq:tC} and \eqref{eq:alttrace2} to write
\begin{equation}
\label{eq:alttrace3}
\lim_{\beta \downarrow 0}  \langle \tilde C_\beta | \hat{a}_L | \tilde C_\beta\rangle = \lim_{\beta \downarrow 0}  \frac{\langle C_\beta | \hat{a}_L | C_\beta\rangle}{\|C_\beta\|^2} = \tr (\hat a_L).
\end{equation}
Here the second step uses the fact that both $\|C_\beta \|^2$ and $\langle C_\beta | \hat a_L | C_\beta \rangle$ have finite limits, and that $\|C_\beta \|^2 \rightarrow 1$.

The formulation in terms of $\tilde C_\beta$ is useful because the operator norm of $\widehat{\tilde C_\beta}_L$ is $1$ (by construction).    We show below that for positive $\hat a_L$ this requires
$\langle \tilde C_\beta | \hat{a}_L | \tilde C_\beta\rangle$ to be a decreasing function of $\beta$, which means that for positive $\hat a_L$  we can also write \eqref{eq:alttrace3} as a supremum over $\beta$:
\begin{equation}
\label{eq:alttrace4}
 \tr (\hat a_L) = \sup_{\beta >0} \langle \tilde C_\beta | \hat{a}_L | \tilde C_\beta\rangle.
\end{equation}
As we will see, this is an improvement over  \eqref{eq:alttrace2} because two supremum operations always commute (while showing that more general limits commute can be notoriously subtle).

To see that $\langle \tilde C_\beta | \hat{a}_L | \tilde C_\beta\rangle$ is a decreasing function of $\beta$, note that for $\beta' >0$ we have
\begin{equation}
\tilde C_{\beta + \beta'} = \tilde C_{\beta} \tilde C_{\beta'},
\end{equation}
where we have used the relation $\left(\|C_\beta\|\right)  \left(\|C_{\beta'}\| \right)= \|C_{\beta+\beta'} \| $ which follows from  Corollary \ref{cor:Cnorms} of appendix \ref{app:lemmas}.
Thus we may write
\begin{equation}
|\tilde C_{\beta + \beta'} \rangle = |\tilde C_{\beta} \tilde C_{\beta'} \rangle = \widehat{\tilde C_{\beta'}}_R |\tilde C_{\beta}\rangle,
\end{equation}
Let us also recall from \eqref{eq:LRcomm} that the right representation $\widehat{\tilde C_{\beta'}}_R$ of $\tilde C_{\beta'}$ commutes with any $\hat a_L$, and thus in particular with the positive $\hat a_L$ of interest here.   Furthermore, since both $\hat a_L$ and $\widehat{\tilde C_{\beta'}}_R = \widehat{\tilde C_{\beta'/2}}_R\widehat{\tilde C_{\beta'/2}}_R^\dagger$ are positive, both operators are self-adjoint.  We may then use the fact that commuting self-adjoint operators can be diagonalized to introduce a complete set of common eigenstates $|\lambda, \kappa \rangle$ where $\lambda\ge0$ is the eigenvalue of $\widehat{C_{\beta'}}_R$ and $\kappa \ge 0$ is the eigenvalue of  $\hat a_L$.  Since the operator norm of $\widehat{\tilde C_{\beta'}}_R$ is 1, the parameter $\lambda$ takes values only in the interval $[0,1]$.  We will also define a measure $d\mu(\lambda, \kappa)$ that gives a resolution of the identity $\mathbb{1} = \int d\mu(\lambda, \kappa) |\lambda, \kappa \rangle \langle \lambda, \kappa|$.

The argument is now straightforward as we may use self-adjointness of $\widehat{\tilde C_{\beta'}}_R$  to  write
\begin{eqnarray}
\langle \tilde C_{\beta + \beta'}  | \hat{a}_L | \tilde C_{\beta + \beta'} \rangle &=& \langle \tilde C_{\beta}  |\widehat{\tilde C_{\beta'}}_R \hat{a}_L \widehat{\tilde C_{\beta'}}_R| \tilde C_{\beta} \rangle \cr
&=& \int d\mu(\lambda, \kappa)   \langle \tilde C_{\beta}  |\widehat{\tilde C_{\beta'}}_R \hat{a}_L \widehat{\tilde C_{\beta'}}_R|\lambda, \kappa \rangle \langle \lambda, \kappa | \tilde C_{\beta} \rangle \cr
&=& \int d\mu(\lambda, \kappa)   \lambda^2 \kappa |\langle \tilde C_{\beta}  |\lambda, \kappa \rangle|^2 \cr &\le& \int d\mu(\lambda, \kappa)   \kappa |\langle \tilde C_{\beta}  |\lambda, \kappa \rangle|^2 = \langle \tilde C_{\beta}  | \hat{a}_L | \tilde C_{\beta} \rangle,
\end{eqnarray}
where we pass from the 3rd to the 4th line by using $\lambda^2 \le 1$.

This shows that $\langle \tilde C_{\beta} | \hat{a}_L | \tilde C_{\beta} \rangle$ increases monotonically as $\beta$ decreases, and thus that \eqref{eq:alttrace4} holds for positive elements $\hat{a}_L$ of $\hat{A}_{L}$.  We may then extend $\tr$ to any positive element in the left von Neumann algebra ${\cal A}^B_L$ via the analogous expression
\begin{equation}
\label{eq:alttrace5}
 \tr (a) := \sup_{\beta >0} \langle \tilde C_\beta | a | \tilde C_\beta\rangle, \ \ \ \text{for positive} \ a \in {\cal A}^B_L,
\end{equation}
and similarly for ${\cal A}^B_R$.  In particular, for all positive operators $a$, the quantity $\langle \tilde C_\beta | a | \tilde C_\beta\rangle$ is non-negative, so that the supremum on the right-hand side must lie in $[0, +\infty]$ as desired. It is worth noting that our argument above actually showed that $\langle \tilde C_{\beta} | a | \tilde C_{\beta} \rangle$ is a decreasing function of $\beta$ for all positive $a\in {\cal A}^B_L$, and therefore if we wish, we may replace the supremum in \eqref{eq:alttrace5} by a limit and write
\begin{equation}
\label{eq:alttrace52}
 \tr (a) = \lim_{\beta\downarrow 0} \langle \tilde C_\beta | a | \tilde C_\beta\rangle, \ \ \ \text{for positive} \ a \in {\cal A}^B_L,
\end{equation}
with the understanding that the limit could be $+\infty$.

Now, in the theory of von Neumann algebras, what we have shown thus far is sufficient to qualify the operation $\tr$ as what is called a {\it weight} on ${\cal A}^B_L$.  For $\tr$ to qualify as what is usually called a {\it trace} requires an additional property, which is that it gives identical results for both $a^\dagger a$ and $a a^\dagger$ for any $a \in {\cal A}^B_L$. This is the form of the familiar cyclic property that is relevant in the context of general von Neumann algebras.

To show this, it will be useful to find yet another characterization of our trace on ${\cal A}^B_L$.  We begin by again recalling that $\widehat{\tilde C_{2\beta'}}_L$ has operator norm $1$, so that $\mathbb{1} - \widehat{\tilde C_{2\beta'}}_L$ is positive and thus $ a^\dagger a - a^\dagger\widehat{\tilde C_{2\beta'}}_L a$ is also positive.  As a result, for any $|b\rangle \in {\cal H}_{LR}$ we have
\begin{equation}
\langle b |a^\dagger a | b \rangle  - \langle b | a^\dagger\widehat{\tilde C_{2\beta'}}_La | b \rangle \ge 0.
\end{equation}
Taking $|b\rangle  =| \tilde C_\beta \rangle$ then gives
\begin{equation}
\langle \tilde C_\beta | a^\dagger a | \tilde C_\beta \rangle \ge \langle \tilde C_\beta | a^\dagger \widehat{\tilde C_{2\beta'}}_L a | \tilde C_\beta \rangle
\end{equation}
for all $\beta, \beta'>0$.   In particular, taking supremums yields
\begin{equation}
\label{eq:sups}
\tr (a^\dagger a) := \sup_{\beta >0} \, \langle \tilde C_\beta | a^\dagger a | \tilde C_\beta \rangle \ge \sup_{\beta, \beta' > 0} \,  \langle \tilde C_\beta | a^\dagger \widehat{\tilde C_{2\beta'}}_L a | \tilde C_\beta \rangle.
\end{equation}

We can in fact show that the inequality in \eqref{eq:sups} is always saturated by using our continuity axiom and the fact that ${\cal A}^B_L$ can be characterized as the closure of $\hat{A}_{L}$ in the {\it strong} operator topology.  This will then give the desired reformulation of our trace that will allow us to prove $\tr\, (a^\dagger a) = \tr\, (a a^\dagger)$.

To establish this result, we first note that the characterization of ${\cal A}^B_L$ as a strong closure means that for any $a\in {\cal A}^B_L$, for fixed $\beta$, and for any $\epsilon >0$ there is an operator $\hat a_L \in \hat{A}_{L}$ such that $a|\tilde C_\beta \rangle - \hat a_L|\tilde C_\beta \rangle$ has magnitude less than $\epsilon$.  Using $\| \tilde C_{2\beta'} \| = 1$, a short computation then yields
\begin{equation}
\label{eq:alim}
|\langle \tilde C_\beta | a^\dagger \widehat{\tilde C_{2\beta'}}_L a | \tilde C_\beta \rangle
- \langle \tilde  C_\beta | \hat a_L^\dagger \widehat{\tilde C_{2\beta'}}_L \hat a_L | \tilde C_\beta \rangle| \le 2 \epsilon \|a\| \, \sqrt{\langle \tilde C_\beta | \tilde C_\beta \rangle} + \epsilon^2.
\end{equation}
Note that this bound also holds if the operators $\widehat{\tilde C_{2\beta'}}_L$ are replaced by $\mathbb{1}$.
Moreover, since Axiom \ref{ax:continuity} requires $\langle \tilde C_\beta | \hat a_L^\dagger \widehat{C_{2\beta'}}_L \hat a_L | \tilde C_\beta \rangle$ to be continuous in $\beta'$, the same continuity holds for $\langle \tilde C_\beta | \hat a_L^\dagger \widehat{\tilde C_{2\beta'}}_L \hat a_L | \tilde C_\beta \rangle$ (as $\|C_{2\beta'}\| = \left(\|C_1\|\right)^{2\beta'}$ is continuous). Thus for small enough $\beta'$ we have
\begin{equation}
\label{eq:blim}
| \langle \tilde C_\beta | \hat a_L^\dagger \widehat{\tilde C_{2\beta'}}_L \hat a_L | \tilde C_\beta \rangle - \langle \tilde C_\beta | \hat a_L^\dagger \hat a_L | \tilde C_\beta \rangle| \le \epsilon.
\end{equation}
Combining \eqref{eq:alim} (as written,  and also with the operator $\widehat{\tilde C_{2\beta'}}_L$ replaced by $\mathbb{1}$) with \eqref{eq:blim}  for small enough $\beta'$ then yields
\begin{eqnarray}
|\langle \tilde C_\beta | a^\dagger \widehat{\tilde C_{2\beta'}}_L a | \tilde C_\beta \rangle
- \langle \tilde  C_\beta |a^\dagger  a| \tilde C_\beta \rangle| &\le&
|\langle \tilde C_\beta | a^\dagger \widehat{\tilde C_{2\beta'}}_L a | \tilde C_\beta \rangle
- \langle \tilde  C_\beta | \hat a_L^\dagger \widehat{\tilde C_{2\beta'}}_L \hat a_L | \tilde C_\beta \rangle| \cr &+& |\langle \tilde  C_\beta | \hat a_L^\dagger \widehat{\tilde C_{2\beta'}}_L \hat a_L | \tilde C_\beta \rangle - \langle \tilde  C_\beta | \hat a_L^\dagger \hat a_L | \tilde C_\beta \rangle|\cr&+& |\langle \tilde  C_\beta | \hat a_L^\dagger  \hat a_L | \tilde C_\beta \rangle
-\langle \tilde C_\beta | a^\dagger   a | \tilde C_\beta \rangle| \cr
\cr &\le& 4 \epsilon \|a\| \, \sqrt{\langle \tilde C_\beta | \tilde C_\beta \rangle} + 2 \epsilon^2 + \epsilon,
\end{eqnarray}
which clearly vanishes as $\epsilon \rightarrow 0$. This shows that $\sup_{\beta'>0} \langle \tilde C_\beta | a^\dagger \widehat{\tilde C_{2\beta'}}_L a | \tilde C_\beta \rangle$ cannot be smaller than $\langle \tilde  C_\beta |a^\dagger  a | \tilde C_\beta \rangle$, and thus that the inequality in \eqref{eq:sups} is saturated. As a result, we have established that for all $a \in {\cal A}^B_L$ (or correspondingly ${\cal A}^B_R$) our trace may be written in the form
\begin{equation}
\label{eq:alttrace6}
\tr(a^\dagger a) = \sup_{\beta, \beta' > 0} \, \langle \tilde C_\beta | a^\dagger \widehat{\tilde C_{2\beta'}}_L a | \tilde C_\beta \rangle.
\end{equation}

To establish cyclicity, we will now show that for any $\beta, \beta'>0$, we have
\begin{equation}
\label{eq:swapbetas}
 \langle \tilde C_\beta | a^\dagger \widehat{\tilde C_{2\beta'}}_L a | \tilde C_\beta \rangle  =  \langle \tilde C_{\beta'} | a \widehat{\tilde C_{2\beta}}_L a^\dagger | \tilde C_{\beta'} \rangle, \quad \forall a \in \A^B_L.
\end{equation}
To show this, we first derive the intermediate result
\begin{equation}
\label{eq:swapbetas2}
 \langle \tilde C_\beta | \hat b_L \widehat{\tilde C_{2\beta'}}_L a | \tilde C_\beta \rangle  =  \langle \tilde C_{\beta'} | a \widehat{\tilde C_{2\beta}}_L \hat b_L | \tilde C_{\beta'} \rangle, \quad \forall a \in \A^B_L,\, \hat b_L \in \hat A_L.
\end{equation}
Our first step is to use the fact that $a$ is the strong operator topology limit of some net of operators $\{\widehat{a_\nu}_L\}$ in $\hat{A}_{L}$, where $\nu$ takes values in some directed index set $\mathscr{J}$ (see again footnote \ref{foot:tops}). This means that the net of states $\{\widehat{a_\nu}_L |\psi\rangle \}$ converges in the Hilbert space norm to $a|\psi \rangle$ for all $|\psi\rangle$, and in particular for the two choices $|\psi\rangle := | \tilde C_\beta \rangle$, $|\psi'\rangle := \widehat{\tilde C_{2\beta}}_L \hat b_L | \tilde C_{\beta'} \rangle$ (defined by the desired $\beta$, $\beta'$, and $\hat b_L$). For any $\epsilon >0$, we may then consider the balls $B_{\epsilon}$, $B_{\epsilon}'$ of radius $\epsilon$ in ${\cal H}_{LR}$ that are respectively centered on the states $a |\psi\rangle$,  $a |\psi'\rangle$. Convergence of the nets  $\{\widehat{a_\nu}_L |\psi \rangle \}$ and $\{\widehat{a_\nu}_L |\psi' \rangle \}$ to  $a |\psi\rangle$ and $a |\psi'\rangle$  means that we can always find a value of $\nu$ such that we have both $\widehat{a_\nu}_L |\psi \rangle \in B_{\epsilon}$ and $\widehat{a_\nu}_L |\psi' \rangle \in B_{\epsilon}'$.  By choosing a sequence $(\epsilon_n)$ in $\mathbb{R}^+$ with $\epsilon_n \rightarrow 0$, we can thus construct a sub{\it sequence} $(\widehat{a_n}_L)$ of the net $\{\widehat{a_\nu}_L \}$ for which we have both $\widehat{a_n}_L |\psi \rangle \rightarrow a |\psi\rangle$ and $\widehat{a_n}_L |\psi' \rangle \rightarrow a |\psi'\rangle$, or more explicitly
\begin{equation}
\label{eq:usefullimits}
\widehat{a_n}_L |\tilde C_\beta \rangle \rightarrow a |\tilde C_\beta\rangle, \quad \text{and}\quad
\widehat{a_n}_L \widehat{\tilde C_{2\beta}}_L \hat b_L | \tilde C_{\beta'} \rangle \rightarrow a \widehat{\tilde C_{2\beta}}_L \hat b_L | \tilde C_{\beta'} \rangle.
\end{equation}
This is a small extension of the standard argument that every metrizable topology is sequential.

The first limit in \eqref{eq:usefullimits} then allows us to write
\begin{eqnarray}
 \langle \tilde C_\beta | \hat b_L \widehat{\tilde C_{2\beta'}}_L a | \tilde C_\beta \rangle &=& \langle \tilde C_\beta| \hat b_L \widehat{\tilde C_{2\beta'}}_L \left( \lim_{n\rightarrow \infty} \widehat{a_n}_L | \tilde C_\beta \rangle \right) \\
\label{eq:Cmanip}
  &=& \lim_{n\rightarrow \infty} \langle \tilde C_\beta| \hat b_L \widehat{\tilde C_{2\beta'}}_L \widehat{a_n}_L | \tilde C_\beta \rangle.
\end{eqnarray}
In passing to the second line we have used the fact that bounded operators and normalizable states define continuous functions on the Hilbert space to take the limit outside the inner product. Similarly, the second limit in \eqref{eq:usefullimits} yields
\begin{eqnarray}
 \langle \tilde C_{\beta'} | a \widehat{\tilde C_{2\beta}}_L \hat b_L | \tilde C_{\beta'} \rangle &=& \langle \tilde C_{\beta'} | \left( \lim_{n\rightarrow \infty} \widehat{a_n}_L \widehat{\tilde C_{2\beta}}_L \hat b_L | \tilde C_{\beta'} \rangle \right) \\
\label{eq:Cmanip2}
  &=& \lim_{n\rightarrow \infty} \langle \tilde C_{\beta'} | \widehat{a_n}_L \widehat{\tilde C_{2\beta}}_L \hat b_L | \tilde C_{\beta'} \rangle.
\end{eqnarray}
Furthermore, \eqref{eq:Cmanip} and \eqref{eq:Cmanip2} are equal since
\begin{equation}
\langle \tilde C_\beta| \hat b_L \widehat{\tilde C_{2\beta'}}_L \widehat{a_n}_L | \tilde C_\beta \rangle
= \tr \left( \tilde C_{\beta} b \tilde C_{2\beta'} a_n \tilde C_{\beta} \right)
= \tr \left( \tilde C_{\beta'} a_n \tilde C_{2\beta} b \tilde C_{\beta'} \right)
= \langle \tilde C_{\beta'} | \widehat{a_n}_L \widehat{\tilde C_{2\beta}}_L \hat b_L | \tilde C_{\beta'} \rangle,
\end{equation}
where the middle step uses cyclicity of the trace \eqref{eq:cyclic} on $\hat A_L$. Thus we have shown the desired intermediate result \eqref{eq:swapbetas2}.

We are now ready to derive \eqref{eq:swapbetas} from \eqref{eq:swapbetas2}. In fact, we  can derive the stronger result
\begin{equation}
\label{eq:swapbetas3}
 \langle \tilde C_\beta | b \widehat{\tilde C_{2\beta'}}_L a | \tilde C_\beta \rangle  =  \langle \tilde C_{\beta'} | a \widehat{\tilde C_{2\beta}}_L b | \tilde C_{\beta'} \rangle, \quad \forall a,b \in \A^B_L,
\end{equation}
from which \eqref{eq:swapbetas} follows immediately by setting $b=a^\dag$. To obtain \eqref{eq:swapbetas3}, we use an argument similar to the one above to find a sequence $(\widehat{b_n}_L)$ in $\hat A_L$ that satisfies both of the conditions
\begin{equation}
\widehat{b_n}_L |\tilde C_{\beta'} \rangle \rightarrow b |\tilde C_{\beta'}\rangle, \quad \text{and}\quad
\widehat{b_n}_L \widehat{\tilde C_{2\beta'}}_L a | \tilde C_\beta \rangle \rightarrow b \widehat{\tilde C_{2\beta'}}_L a | \tilde C_\beta \rangle.
\end{equation}
Then \eqref{eq:swapbetas3} follows by writing
\begin{eqnarray}
\langle \tilde C_\beta | b \widehat{\tilde C_{2\beta'}}_L a | \tilde C_\beta \rangle
&=& \langle \tilde C_\beta | \lim_{n\rightarrow \infty} \left(\widehat{b_n}_L \widehat{\tilde C_{2\beta'}}_L a | \tilde C_\beta \rangle \right) \\
&=& \lim_{n\rightarrow \infty} \langle \tilde C_\beta | \widehat{b_n}_L \widehat{\tilde C_{2\beta'}}_L a | \tilde C_\beta \rangle \\
&=& \lim_{n\rightarrow \infty} \langle \tilde C_{\beta'} | a \widehat{\tilde C_{2\beta}}_L \widehat{b_n}_L | \tilde C_{\beta'} \rangle \\
&=& \langle \tilde C_{\beta'} | a \widehat{\tilde C_{2\beta}}_L \lim_{n\rightarrow \infty} \left(\widehat{b_n}_L | \tilde C_{\beta'} \rangle \right) \\
&=& \langle \tilde C_{\beta'} | a \widehat{\tilde C_{2\beta}}_L b | \tilde C_{\beta'} \rangle,
\end{eqnarray}
where the middle step follows by applying \eqref{eq:swapbetas2} to each $\widehat{b_n}_L$.

Having established \eqref{eq:swapbetas}, we take the supremum over $\beta$ and $\beta'$ on both sides of this relation and use \eqref{eq:alttrace6} to obtain the desired cyclic identity
\begin{equation}
\label{eq:cyclicvN}
\tr \, a^\dagger a = \tr \, a a^\dagger, \ \ \ \forall a \in {\cal A}^B_L.
\end{equation}

We emphasize that our trace will generally give $+\infty$ for some positive elements of ${\cal A}^B_L$. In particular, according to \eqref{eq:alttrace52} the trace of the identity operator $\mathbb{1}$  is
\begin{equation}
\label{eq:traceid}
\tr \, \mathbb{1}=\lim_{\beta \downarrow 0} \langle \tilde C_{\beta} | \tilde C_{\beta} \rangle =\lim_{\beta \downarrow 0} \|C_{\beta}\|^{-2}\, \tr (C_{2\beta}) =\lim_{\beta \downarrow 0} \, \tr (C_{2\beta}),
\end{equation}
where in the last step we used $\lim_{\beta \downarrow 0} \|C_{\beta}\| =1$ from Lemma \ref{lemma:Cnorm} of appendix \ref{app:lemmas}. The right-hand side is the trace of a cylinder of vanishing length, which certainly diverges in familiar semiclassical theories of gravity.

\section{Type I von Neumann Factors, Hilbert Space Structure, and Entropy}
\label{sec:typeI}

As indicated above, the trace operation $\tr$ will turn out to be the key to unlocking the structure of any von Neumann algebra ${\cal A}^B_L$ defined as above by a diagonal Hilbert space sector ${\cal H}_{LR} = {\cal H}_{B \sqcup B}$, as well as to unlocking the structure of ${\cal H}_{LR} = {\cal H}_{B \sqcup B}$ itself.  Our work in section \ref{sec:algebras} established that $\tr$ satisfies the following two properties on ${\cal A}^B_L$:

\begin{enumerate}
\item Linearity: $\tr(a+b)=\tr(a)+\tr(b)$, and $\tr(\lambda a)=\lambda \tr(a)$ for any  positive $a,b \in \mathcal{A}^B_L$ and $\lambda\geq 0$.
\item Cyclicity: $\tr(a a^\dagger)= \tr(a^\dagger a)$.  This in particular implies that the trace is invariant under the action of unitaries in the sense that for $b,U \in {\cal A}^B_L$ with $U^\dagger = U^{-1}$, defining $c =b U^\dagger$ gives $\tr(Ub^\dagger b U^\dagger) = \tr(c^\dagger c) = \tr(cc^\dagger) = \tr (b U^{-1} U b^\dagger)= \tr(bb^\dagger) =  \tr (b^\dagger b)$.
\end{enumerate}

We can also establish three further properties:
\begin{enumerate}
\setcounter{enumi}{2}
\item Faithfulness: for positive $a \in \mathcal{A}^B_L$, we have $\tr(a)=0$ if and only if $a=0$.
\item Semifiniteness: for any non-zero positive operator $a\in \mathcal{A}^B_L$, there exists a non-zero positive operator $b\in \mathcal{A}^B_L$ with $b \le a$ and $\tr (b)<\infty$. The notation $b \le a$ means that $a-b$ is positive.
\item Normality: for a bounded increasing net of positive operators $\rho_\alpha$ in $\mathcal{A}^B_L$ with supremum $\rho=\sup _\alpha \rho_\alpha$, we have $\tr \, \rho = \sup_{\alpha} \tr \, \rho_\alpha$. (See appendix \ref{app:trprop} for details.)
\end{enumerate}
The faithfulness property was shown to hold on $\hat A_L$ in section \ref{subsec:Rep}, but here we wish to show that it holds on the full von Neumann algebra ${\cal A}^B_L$.  We will give a similar proof in section \ref{subsec:exTrIn} after showing that the trace inequality also extends to ${\cal A}^B_L$.

The proofs of properties 4 and 5 are short, but they are somewhat technical.  To avoid distraction from the main results we thus relegate them to appendix \ref{app:trprop}.  Of course, each property above has an analogue for ${\cal A}^B_R$.

As noted in section \ref{sec:algebras}, properties 1 and 2 are the minimal requirements for the function $\tr$ to be called a trace on a von Neumann algebra.  The faithfulness property then gives a sense in which our trace is non-degenerate.  Semifiniteness guarantees that not all non-zero operators have infinite trace, and the normality condition describes a sense in which the trace is continuous.

These latter properties are important since there is no faithful normal semifinite trace on a type III von Neumann factor.  Establishing 3,  4, and 5 above thus tells us that our von Neumann algebra contains only type I and type II factors.  Furthermore, for such factors  there is a unique faithful, normal, semifinite trace up to an overall factor (about which more will be said below); see e.g.\ \cite{Takesaki}.

As noted above, our argument for faithfulness will rely on extending the trace inequality \eqref{eq:TrIn1} to the full von Neumann algebra.  It turns out that this can be accomplished by an extension of the argument of section \ref{subsec:tr1}.  This will be done in section \ref{subsec:exTrIn} below.  Section \ref{subsec:typeI} will then use this result to show that type II factors are excluded (implying  that ${\cal A}^B_L$ contains only type I factors) and to analyze the implications for the structure of ${\cal H}_{B \sqcup B}$.  The entropy defined by $\tr$ is then discussed in section \ref{subsec:SfromA}.

\subsection{The trace inequality on \texorpdfstring{${\cal A}^B_L$}{}}
\label{subsec:exTrIn}

We wish to extend the argument of section \ref{subsec:tr1} to establish the trace inequality on ${\cal A}^B_L$.  It will first be useful to establish the following regularized version of the trace inequality, whose derivation will have much in common with the argument of section \ref{subsec:tr1}.

\begin{lemma}
For any $\beta, \beta' >0$ and any $a,b \in {\cal A}^B_L$, we have
\begin{equation}
\label{eq:regTrIn}
\langle \tilde C_\beta | a^\dagger b \widehat{\tilde C_{2 \beta'}}_L b^\dagger a | \tilde C_\beta \rangle \le
\langle \tilde C_\beta |  a^\dagger a | \tilde C_\beta \rangle \langle \tilde C_{\beta'} |  b^\dagger b | \tilde C_{\beta'} \rangle.
\end{equation}
\end{lemma}

\begin{proof}
Let us choose nets $\{a_\nu \}, \{b_\kappa \} \subset A^B_L$ so that  the nets
$\{\hat{a}_{\nu,L}\}, \{\hat{b}_{\kappa,L}\}$ of representatives on ${\cal H}_{B\sqcup B}$ converge in the strong operator topology to the desired operators $a, b \in {\cal A}^B_L$.
For later use we note that for any $\beta, \beta' > 0$ this implies that the nets of states
$\{\hat{a}_{\nu,L} |\tilde C_\beta \rangle \}$ and $\{\hat{b}_{\kappa,L} |\tilde C_{\beta'} \rangle \}$ converge respectively to $a |\tilde C_\beta \rangle$ and $b |\tilde C_{\beta'} \rangle$,
where here the limits are taken using the standard Hilbert space topology.
As in the proof of \eqref{eq:cyclicvN}, we can then find {\it sequences} $(\hat{a}_{n,L} |\tilde C_\beta \rangle )$ and $(\hat{b}_{m,L} |\tilde C_{\beta'}\rangle)$ that also satisfy
\begin{equation}
\label{eq:abCconv}
\hat{a}_{n,L} |\tilde C_\beta \rangle \rightarrow a |\tilde C_\beta \rangle,  \ \ \ {\rm and} \ \ \
\hat{b}_{m,L} |\tilde C_{\beta'} \rangle \rightarrow b |\tilde C_{\beta'} \rangle.
\end{equation}
In the notation of section \ref{subsec:tr1}, consider again the `4-boundary' Hilbert space ${\cal H}_{B_{L_1}, B_{R_1}, B_{L_2}, B_{R_2}}$ and the associated pre-Hilbert space
$H_{B_{L_1}, B_{R_1}, B_{L_2}, B_{R_2}}$.
Note that these spaces both contain the states
$\left|(\tilde C_\beta)_{L_1,R_1}, (\tilde C_{\beta'})_{L_2,R_2} \right\rangle$ and $\left|(\tilde C_\beta)_{L_2,R_1}, (\tilde C_{\beta'})_{L_1,R_2} \right\rangle$.
Acting with the sequences $(\hat{a}_{n,L})$ and $(\hat{b}_{m,L})$, we define the states
\begin{eqnarray}
\label{eq:PSI12}
|\Psi_1(n,m) \rangle &: =& \hat a_{n,L_1} \hat b_{m,L_2} \left|(\tilde C_\beta)_{L_1,R_1}, (\tilde C_{\beta'})_{L_2,R_2} \right\rangle,\\
|\Psi_2 (n,m)\rangle &: =& \hat a_{n,L_2} \hat b_{m,L_1} \left|(\tilde C_\beta)_{L_2,R_1}, (\tilde C_{\beta'})_{L_1,R_2} \right\rangle,
\end{eqnarray}
where the operators act at the boundaries indicated by the subscripts $L_1, L_2$; these states again lie in both the Hilbert space ${\cal H}_{B_{L_1}, B_{R_1}, B_{L_2}, B_{R_2}}$ and the pre-Hilbert space $H_{B_{L_1}, B_{R_1}, B_{L_2}, B_{R_2}}$.  Note that, as in section \ref{subsec:tr1}, the two states considered here are related by the action of the
`swap' operator ${\cal S}_{L_1, L_2}$ that exchanges the labels $L_1,L_2$ on the relevant two copies of $B$; see figure \ref{fig:4bswap}.
\begin{figure}[ht!]
        \centering
\includegraphics[scale=0.8]{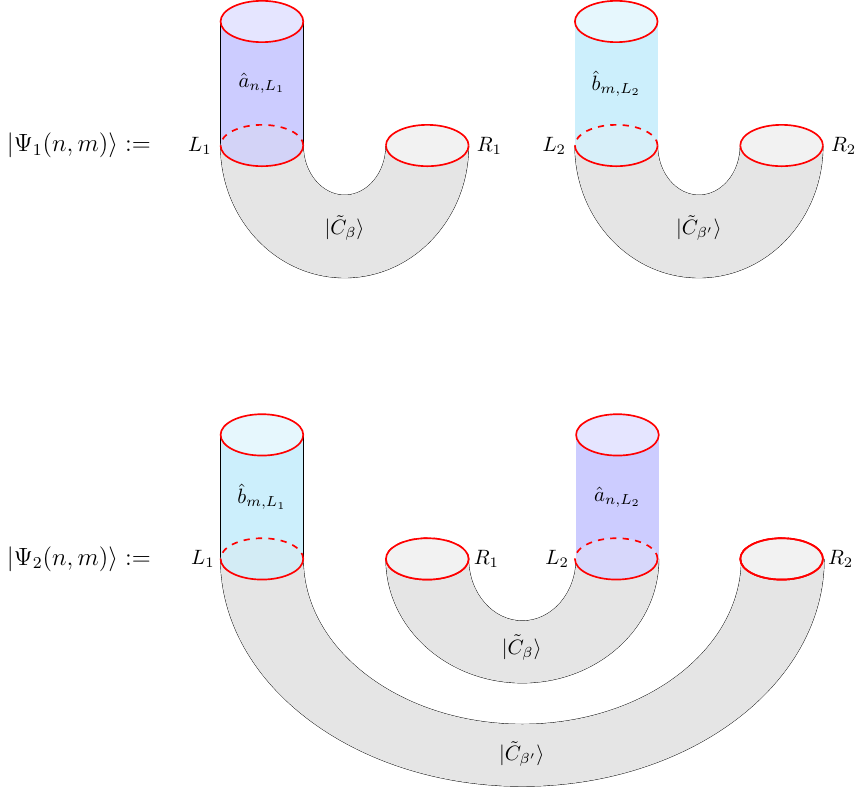}\caption{The states $|\Psi_1(n,m) \rangle$, $|\Psi_2(n,m) \rangle$ defined by \eqref{eq:PSI12}.}
\label{fig:4bswap}
\end{figure}

We will now use \eqref{eq:abCconv} to show that the associated diagonal sequences $\{|\Psi_1(n,n)\rangle \}$, $\{|\Psi_2(n,n)\rangle \}$ are both Cauchy sequences in $H_{B_{L_1}, B_{R_1}, B_{L_2}, B_{R_2}}$, so that their limits define states in ${\cal H}_{B_{L_1}, B_{R_1}, B_{L_2}, B_{R_2}}$ that we may call
\begin{equation}
\label{eq:PSI12lim}
|\Psi_1 \rangle  := \lim_{n\to \infty} |\Psi_1(n,n)\rangle,\quad
|\Psi_2 \rangle  := \lim_{n\to \infty} |\Psi_2(n,n)\rangle.
\end{equation}

To see that these sequences are Cauchy, we first compute the (pre-)inner products
\begin{eqnarray}
\langle \Psi_1(n,m) | \Psi_1(n',m') \rangle &=& \tr(\tilde C_\beta a_n^\star a_{n'} \tilde C_\beta) \tr(\tilde C_{\beta'} b_m^\star b_{m'} \tilde C_{\beta'}) \cr &=& \langle \tilde C_\beta |\hat{a}_{n,L}^\dagger \hat{a}_{n',L} | \tilde C_\beta \rangle
\langle \tilde C_{\beta'} |\hat{b}_{m,L}^\dagger \hat{b}_{m',L} | \tilde C_{\beta'} \rangle. \ \ \
\end{eqnarray}
Due to the convergence of \eqref{eq:abCconv}, for any $\epsilon > 0$ there are integers $n_0, m_0$ such that for all $n, n'>n_0$ and $m, m' > m_0$ we have
\begin{eqnarray}
|\langle \tilde C_\beta |\hat{a}_{n,L}^\dagger \hat{a}_{n',L} | \tilde C_\beta \rangle - \langle \tilde C_\beta |a^\dagger a | \tilde C_\beta \rangle| < \epsilon, \cr
|\langle \tilde C_{\beta'} |\hat{b}_{m,L}^\dagger \hat{b}_{m',L} | \tilde C_{\beta'} \rangle - \langle \tilde C_{\beta'} |b^\dagger b | \tilde C_{\beta'} \rangle| < \epsilon.
\end{eqnarray}
There is thus some $n_1$ such that for all $n,m,n',m' > n_1$ we have
\begin{equation}
\Big|\langle \Psi_1(n,m) | \Psi_1(n',m') \rangle -  \langle \tilde C_\beta |a^\dagger a | \tilde C_\beta \rangle \langle \tilde C_\beta |b^\dagger b | \tilde C_\beta \rangle\Big| < \epsilon.
\end{equation}
The usual computation then shows that we also have
\begin{align}
\label{eq:PSI1Cauchy}
&\Big| |\Psi_1(n,m) \rangle
- |\Psi_1(n',m') \rangle \Big |^2 \cr
=& \langle \Psi_1(n,m) | \Psi_1(n,m) \rangle - 2 {\rm Re} \langle \Psi_1(n,m) | \Psi_1(n',m') \rangle + \langle \Psi_1(n',m') | \Psi_1(n',m') \rangle \le 4 \epsilon.
\end{align}
In particular, the sequence $\{|\Psi_1(n,n)\rangle \}$ is Cauchy.  The argument for $\{|\Psi_2(n,n)\rangle \}$ is identical.
It is also clear from the work above that we have the norms
\begin{equation}
\label{eq:limnorms}
\langle \Psi_1 | \Psi_1 \rangle = \langle \Psi_2 | \Psi_2 \rangle =\langle \tilde C_\beta |a^\dagger a | \tilde C_\beta \rangle \langle \tilde C_{\beta'} |b^\dagger b | \tilde C_{\beta'} \rangle.
\end{equation}

Note that it was not really necessary to take the limits along the diagonal.  The above argument also establishes the relations
\begin{eqnarray}
|\Psi_1 \rangle &=& \lim_{n\rightarrow \infty} \lim_{m\rightarrow \infty}  |\Psi_1(n,m) \rangle
=  \lim_{m\rightarrow \infty}\lim_{n\rightarrow \infty}  |\Psi_1(n,m) \rangle, \cr
|\Psi_2 \rangle &=& \lim_{n\rightarrow \infty} \lim_{m\rightarrow \infty}  |\Psi_2(n,m) \rangle
=  \lim_{m\rightarrow \infty}\lim_{n\rightarrow \infty}  |\Psi_2(n,m) \rangle,
\end{eqnarray}
so that we may take these limits in any order that we like.

Having represented the states $|\Psi_1 \rangle$,
$|\Psi_2 \rangle$ in terms of the above limits,  we may use continuity of the inner product to write $\langle \Psi_1 |\Psi_2 \rangle$ as a limit of inner products
$\langle \Psi_1 (n, m) | \Psi_2(n', m') \rangle$.  Moreover, we can use the above freedom to choose the order of limits to first take $n,n'\rightarrow \infty$ while saving the limit $m,m'\rightarrow \infty$ for later.  Thus we write
\begin{eqnarray}
\label{eq:P1P2ip1}
\langle \Psi_1  | \Psi_2\rangle &:=& \lim_{m, m' \rightarrow \infty} \lim_{n, n' \rightarrow \infty}\langle \Psi_1 (n, m) | \Psi_2(n', m') \rangle \cr
&=& \lim_{m, m' \rightarrow \infty} \lim_{n, n' \rightarrow \infty}\tr(\tilde C_\beta a_n^\star b_{m'} \tilde C_{2\beta'} b_m^\star a_{n'} \tilde C_\beta) \cr
&=& \lim_{m, m' \rightarrow \infty} \lim_{n, n' \rightarrow \infty} \langle \tilde C_\beta |\hat{a}_{n,L}^\dagger \hat{b}_{m',L} \widehat{\tilde C_{2\beta'}}_L \hat{b}_{m,L}^\dagger \hat{a}_{n',L} | \tilde C_\beta \rangle \cr
&=& \lim_{m, m' \rightarrow \infty}  \langle \tilde C_\beta |a^\dagger \hat{b}_{m',L} \widehat{\tilde C_{2\beta'}}_L \hat{b}_{m,L}^\dagger a | \tilde C_\beta \rangle,
\end{eqnarray}
where the second line can be read off from figure \ref{fig:4bswap}.  In addition,  the final step used \eqref{eq:abCconv} and the fact that the operators $\hat{b}_{m',L}$, $\widehat{\tilde C_{2\beta'}}_L,$ and  $\hat{b}_{m,L}^\dagger$ are bounded. We may then use \eqref{eq:swapbetas3} (with $a$, $b$ replaced first by $\hat{b}_{m,L}^\dagger a$, $a^\dagger \hat{b}_{m',L}$ and then by $b^\dagger a$, $a^\dagger b$) to write
\begin{eqnarray}
\label{eq:P1P2ip}
\langle \Psi_1  | \Psi_2\rangle
&=& \lim_{m, m' \rightarrow \infty}  \langle\tilde C_{\beta'} |\hat{b}_{m,L}^\dagger a  \widehat{\tilde C_{2\beta}}_L  a^\dagger \hat{b}_{m',L}| \tilde C_{\beta'} \rangle \cr
&=&  \langle \tilde C_{\beta'} |b^\dagger a  \widehat{\tilde C_{2\beta}}_L  a^\dagger b| \tilde C_{\beta'} \rangle
=  \langle\tilde C_{\beta} |a^\dagger b  \widehat{\tilde C_{2\beta'}}_L  b^\dagger a | \tilde C_{\beta} \rangle.
\end{eqnarray}
Note that \eqref{eq:P1P2ip} is manifestly real and non-negative, as it is the norm squared of $\left|\widehat{\tilde C_{\beta'}}_L  b^\dagger a | \tilde C_{\beta} \right\rangle$.

The desired Lemma now follows by applying the Cauchy-Schwarz inequality to $|\Psi_1\rangle, |\Psi_2\rangle$ and using
\eqref{eq:limnorms} and \eqref{eq:P1P2ip}.
\end{proof}

Having derived \eqref{eq:regTrIn} for any $\beta, \beta' > 0$, we now take supremums of both sides over both $\beta$ and $\beta'$. By \eqref{eq:alttrace6}, taking $\sup_{\beta, \beta'>0}$ on the left-hand side gives $\tr(a^\dagger bb^\dagger a)$.  On the right-hand side we simply have
\begin{eqnarray}
\sup_{\beta, \beta'>0} \left(
\langle \tilde C_\beta |  a^\dagger a | \tilde C_\beta \rangle \langle \tilde C_{\beta'} |  b^\dagger b | \tilde C_{\beta'} \rangle \right) &=&
 \left(\sup_{\beta>0}
\langle \tilde C_\beta |  a^\dagger a | \tilde C_\beta \rangle\right) \left( \sup_{\beta'>0}  \langle \tilde C_{\beta'} |  b^\dagger b | \tilde C_{\beta'} \rangle \right) \cr
&=& \tr(a^\dagger a) \, \tr(b^\dagger b),
\end{eqnarray}
with the convention that if one supremum (or trace) is $0$ while the other supremum (or trace) is $+\infty$, their product is defined as $0$.

Thus for all $a, b \in {\cal A}^B_L$ we obtain
\begin{equation}
\label{eq:TrInvN}
\tr(a^\dagger bb^\dagger a) \le \tr(a^\dagger a) \, \tr(b^\dagger b).
\end{equation}
This is our trace inequality on the von Neumann algebra ${\cal A}^B_L$.
If we were allowed to cyclically permute $a^\dagger bb^\dagger a$ inside the left trace to $aa^\dagger bb^\dagger$, then this would be of the form $\tr(AB) \le \tr(A) \tr(B)$ for two positive operators $A= aa^\dagger$, $B=bb^\dagger$.  Such a cyclic permutation is not in fact allowed because the trace $\tr$ on ${\cal A}^B_L$ is defined only on positive elements of ${\cal A}^B_L$. While the operator $a^\dagger bb^\dagger a$ is positive, the operator $AB$ is known to be positive if and only if $A$ and $B$ commute.  Nevertheless, the trace inequality \eqref{eq:TrInvN} will suffice for our purposes below.

Before continuing, however, we pause to note three useful corollaries of the above argument.  The first is
\begin{corollary}
\label{cor:triprel}
For $a \in {\cal A}^B_L$ with $\tr(a^\dagger a)$ finite, the limit $\displaystyle \lim_{\beta \downarrow 0}\, a|\tilde C_\beta\rangle$ converges in ${\cal H}_{B \sqcup B}$.  Let us call the limit $|a\rangle$.  Then we also have
\begin{equation}
\label{eq:triprel}
\tr(a^\dagger a) = \langle a | a \rangle.
\end{equation}
\end{corollary}
\begin{proof}
Note that for $\beta' > \beta$ we have
\begin{equation}
\langle \tilde C_\beta | a^\dagger a | \tilde C_{\beta'} \rangle = \Braket{ \tilde C_\beta | a^\dagger a \widehat{\tilde C_{\frac{\beta'-\beta}{2}}}_R | \tilde C_{\frac{\beta'+\beta}{2}} } = \Braket{ \tilde   C_{\frac{\beta'+\beta}{2}}| a^\dagger a | \tilde C_{\frac{\beta'+\beta}{2}} },
\end{equation}
since $\widehat{\tilde C_{\frac{\beta'-\beta}{2}}}_R$ commutes with any operator in the left algebra  ${\cal A}^B_L$.
The right-hand side approaches a finite limit $\tr(a^\dagger a)$ as $\beta, \beta' \to 0$, and thus we may use steps much like those above to find that any sequence $\beta_n \rightarrow 0$ defines a Cauchy sequence $a|\tilde C_{\beta_n}\rangle$ in ${\cal H}_{B\sqcup B}$, and that the limit is the same for all such sequences.  Calling the limit $|a\rangle$ and using \eqref{eq:alttrace52} then immediately gives \eqref{eq:triprel}.
\end{proof}

We also have
\begin{corollary}
\label{cor:faithfulvN}
The trace $\tr$ defined by \eqref{eq:alttrace5} is faithful on the von Neumann algebras ${\cal A}^B_L$ and ${\cal A}^B_R$.
\end{corollary}
\begin{proof}
Together with Lemma \ref{lemma:Cnorm} from appendix \ref{app:lemmas}, the continuity axiom (Axiom \ref{ax:continuity}) implies that for any rimmed surface $a \in \underline{Y}^d_{LR}$ the states $\hat{a}_L |\tilde C_\beta \rangle = |a\tilde C_\beta\rangle$ converge in the Hilbert space norm to  $|a\rangle$ as $\beta \rightarrow 0$.  Since any $b \in {\cal A}^B_L$ is bounded, we also find $b^\dagger\hat{a}_L |\tilde C_\beta \rangle \rightarrow b^\dagger|a\rangle$.  Thus
\begin{equation}
\label{eq:cor2eq}
\langle a | b b^\dagger |a \rangle = \lim_{\beta \downarrow 0} \langle \tilde C_\beta |\hat{a}_L^\dagger bb^\dagger \hat{a}_L|\tilde C_\beta \rangle =  \tr(\hat a_L^\dagger b b^\dagger \hat a_L) \le \tr(\hat a_L^\dagger \hat a_L) \, \tr(b^\dagger b),
\end{equation}
where in the last step we have used the trace inequality \eqref{eq:TrInvN}. It follows that if $\tr(b b^\dagger)$ vanishes for any $b$, then \eqref{eq:cor2eq} requires $b^\dagger |a \rangle$ to vanish for all $a \in \underline{Y}^d_{LR}$.  But $b^\dagger$ is bounded, and the states $|a\rangle$ define a dense subspace of the Hilbert space, so we must have $b b^\dagger=0$.  This establishes faithfulness on ${\cal A}^B_L$, and the argument on ${\cal A}^B_L$ is identical.
\end{proof}

Finally, we have
\begin{corollary}
\label{cor:tenprod}
Given a Hilbert space sector ${\cal H}_{B'}$, the $n$-fold tensor product ${\cal H}_{B'}^{\otimes n} : =  \otimes_{i=1}^n {\cal H}_{B'} = {\cal H}_{B'} \otimes \dots \otimes {\cal H}_{B'}$ (with $n$ factors on the right-hand side) is naturally identified with a subspace of the Hilbert space ${\cal H}_{\sqcup_{i=1}^n B'} : = {\cal H}_{B' \sqcup \dots \sqcup B'}$  associated with $n$ copies of $B'$.
\end{corollary}
\begin{proof}
We first prove this for $n=2$. We wish to show that any $|a\rangle \otimes |b\rangle \in {\cal H}_{B'} \otimes {\cal H}_{B'}$ is naturally mapped to a state in ${\cal H}_{B' \sqcup B'}$. Since $|a\rangle, |b\rangle \in {\cal H}_{B'}$, we can find sequences $|a_m\rangle, |b_m\rangle$ in $H_{B'}$ that converge to $|a\rangle, |b\rangle$, respectively. Using steps much like the ones used above in showing $|\Psi_1(m,m)\rangle$ to be Cauchy, we find that $|a_m \sqcup b_m\rangle$ is a Cauchy sequence in $H_{B' \sqcup B'}$. Moreover, its limit in ${\cal H}_{B' \sqcup B'}$ is independent of the choices for the sequences $|a_m\rangle, |b_m\rangle$ so long as they converge to $|a\rangle, |b\rangle$. Thus we may call this limit $|a\sqcup b\rangle$, and so we have defined a natural map from ${\cal H}_{B'} \otimes {\cal H}_{B'}$ to ${\cal H}_{B' \sqcup B'}$ by mapping $|a\rangle \otimes |b\rangle$ to $|a\sqcup b\rangle$. Moreover, this map is linear and preserves the inner product. Therefore, it provides a natural isomorphism between ${\cal H}_{B'} \otimes {\cal H}_{B'}$ and a subspace of ${\cal H}_{B' \sqcup B'}$. This argument clearly generalizes to all $n\in {\mathbb Z}^+$, thus establishing this corollary.
\end{proof}

This corollary allows us to embed ${\cal H}_{B'}^{\otimes n}$ into ${\cal H}_{\sqcup_{i=1}^n B'}$ as a subspace. Thus, for any operator acting on any one of the $n$ tensor factors of ${\cal H}_{B'}$, we can now also allow it to act on this subspace of ${\cal H}_{\sqcup_{i=1}^n B'}$. For the case of $n=2$ and $B' = B \sqcup B$, we can use this fact to write $|\Psi_1\rangle$ from \eqref{eq:PSI12lim} in the form:
\begin{equation}
|\Psi_1 \rangle  = a_{L_1} b_{L_2} \left|(\tilde C_\beta)_{L_1,R_1}, (\tilde C_{\beta'})_{L_2,R_2} \right\rangle,
\end{equation}
where we have used \eqref{eq:abCconv} and \eqref{eq:PSI12}. In the limit $\beta,\beta'\to 1$, $|\Psi_1 \rangle$ converges to $|a_{L_1,R_1}, b_{L_2,R_2}\rangle$ (called $|a\sqcup b\rangle$ in the previous paragraph) and the inner product \eqref{eq:P1P2ip} becomes
\begin{equation}\label{eq:P1P2ip2}
\langle a_{L_1,R_1}, b_{L_2,R_2}|{\cal S}_{L_1, L_2}|a_{L_1,R_1}, b_{L_2,R_2}\rangle = \langle b| a a^\dag |b\rangle = \langle a| b b^\dagger |a\rangle.
\end{equation}

Note, however, that Corollary \ref{cor:tenprod} states only that  ${\cal H}_{B'}^{\otimes n} = \otimes_{i=1}^n {\cal H}_{B'}$ gives a subspace of ${\cal H}_{\sqcup_{i=1}^n B'}$, allowing that
${\cal H}_{\sqcup_{i=1}^n B'}$ may well be strictly larger than $\otimes_{i=1}^n {\cal H}_{B'}$.  This is in particular true for the topological model of \cite{Marolf:2020xie} without end-of-the-world branes, as well as for models with end-of-the-world branes studied in \cite{Marolf:2020xie} when considered in baby universe superselection sectors where the partition function is larger than the number of flavors of such branes.  As noted in \cite{Marolf:2020xie}, this discussion is directly analogous to the considerations of \cite{Harlow:2015lma} (see also \cite{Harlow:2018tqv}), so the issue may be called the `Harlow factorization question'. When such extra states exist, and if one wishes to insist that there be a dual formulation as a standard non-gravitating quantum field theory (for which locality would strictly require
${\cal H}_{\sqcup_{i=1}^n B'} = \otimes_{i=1}^n {\cal H}_{B'}$), one might wish to  call this phenomenon the `Harlow factorization problem'.

\subsection{\texorpdfstring{$\mathcal{A}_{L}^B$}{} and \texorpdfstring{$\mathcal{A}_{R}^B$}{} contain only type I factors}
\label{subsec:typeI}

We can now say much more about the structure of the von Neumann algebras $\mathcal{A}_{L}^B$ and $\mathcal{A}_{R}^B$ defined by a diagonal Hilbert space ${\cal H}_{B \sqcup B}$.  This is the part of the paper where we have developed enough control over our algebras, and in particular over our trace $\tr$, to reach into the mathematics literature and make use of powerful results (even if they are nevertheless elementary by the standards of theorems about von Neumann algebras).

It turns out that much of the study of a von Neumann algebra ${\cal A}$  can be reduced to the study of projections $P \in {\cal A}$.  Here as usual a projection is defined as an operator that satisfies $P = P^\dagger$ and $P^2 =P$.  It will thus be useful to better understand the implications of our results for such $P$.

Let us in particular apply our von Neumann algebra trace inequality \eqref{eq:TrInvN} to the case $a=b=P\in {\cal A}^B_L$. Since $a^\dagger b b^\dagger a = P^4 = P$, we quickly obtain
\begin{equation}
\label{eq:Projbound}
    \tr P \leq (\tr P)^2.
\end{equation}
Since $\tr$ is faithful, unless $P$ is the trivial projection $P=0$ we must have
\begin{equation}
\label{eq:trPbound}
\tr P \geq 1.
\end{equation}
Thus the trace of any non-zero projection in $\mathcal{A}^B_L$ is bounded below by 1.

Now, any von Neumann algebra ${\cal A}$ can be decomposed as the direct sum/integral of so-called von Neumann {\it factors} ${\cal A}_\mu$ which are just von Neumann algebras with trivial centers.  Furthermore, any faithful normal semifinite trace on ${\cal A}$ induces a faithful normal semifinite trace on every factor ${\cal A}_\mu$.  Since it is known that any von Neumann factor is of type I, II, or III, and since there is no such trace on any type III factor, all factors of our ${\cal A}^B_L$ must be of type I or type II; see e.g.\ \cite{Takesaki}. (As usual, we should in principle allow for exceptions on sets of measure zero. However, Lemma \ref{lemma:discrete} below will show that ${\cal A}^B_L$ is a discrete direct sum of factors so that no interesting such exceptions can arise.)

However, \eqref{eq:trPbound} quickly leads to a much stronger result.
The crucial point is that for any faithful normal semifinite trace on a non-trivial type II von Neumann factor, there is decreasing family of non-zero projections $P_\lambda$ with $\lambda \in {\mathbb R}^+$ such that $\tr (P_\lambda) \rightarrow 0$ as $\lambda \rightarrow 0$; see e.g.\ proposition 8.5.5 of \cite{KR:1997}.    But for small enough $\lambda$ such $P_\lambda$ clearly violate \eqref{eq:trPbound}.  We thus see that ${\cal A}^B_L$ must contain only type I factors.  Such von Neumann algebras are said to be of type I.  This is the first key result of this paper.

Another important result from the literature is the so-called commutation theorem for semifinite traces; see e.g.\ theorem 2.22 of \cite{Takesaki}.  This theorem states that a von Neumann algebra $\mathcal{A}$ with a semifinite trace $\tr$ is the commutant of its opposite algebra $\mathcal{A}^{op}$ ($\mathcal{A}$ with reversed multiplication rule) when acting on the Hilbert space $\mathcal{H}=\left\{a \in \mathcal{A}: \tr (a^{\dagger} a) <\infty\right\}$.  Here  the two algebras act by left and right multiplication $\hat b_L |a\rangle = |ba\rangle$ and $\hat b_R |a \rangle = |ab \rangle$, and the above notation means that a dense subspace of $\mathcal{H}$ is defined by operators with finite $\tr (a^{\dagger} a)$, and that in this subspace we have $\langle a | a\rangle = \tr (a^\dagger a)$ as in \eqref{eq:triprel}.    This is precisely the structure of any diagonal Hilbert space ${\cal H}_{B \sqcup B}$, on which the algebras $\A_L^B$ and $\A_R^B$ act as opposites. It thus follows that $\A_L^B$ and $\A_R^B$ are commutants on ${\cal H}_{B \sqcup B}$. Alternately, without using \eqref{eq:triprel} in Corollary \ref{cor:triprel}, one can check that our algebras satisfy the conditions for the commutation theorems in \cite{Rieffel:1976} which again imply that $\A_L^B$ and $\A_R^B$ are commutants.   In either case, saying that
$\A_L^B$ and $\A_R^B$ are commutants on ${\cal H}_{B \sqcup B}$ means that $\A_R^B$ consists precisely of those bounded operators on ${\cal H}_{B \sqcup B}$ that commute with all operators in $\A_L^B$, and also that $\A_L^B$ consists precisely of those bounded operators on ${\cal H}_{B \sqcup B}$ that commute with all operators in $\A_R^B$.

The above observations now tell us much about the structure of a diagonal Hilbert space sector ${\cal H}_{B \sqcup B}$.
In analyzing this structure, it is useful to consider the center ${\cal Z}^B_L$  of $\A_L^B$, which is defined to the  subalgebra of operators in $\A_L^B$ that commute with all operators in $\A_L^B$; i.e., ${\cal Z}^B_L := \{a  |\, a\in \A_L^B,\  ab=ba \,\,\,\, \forall b \in \A_L^B \}$.  In particular, any $z \in {\cal Z}^B_L$ commutes with its adjoint $z^\dagger$.  This means that central operators are normal and can be diagonalized on ${\cal H}_{B \sqcup B}$.  In fact, since all elements of
${\cal Z}^B_L$ commute with each other, we can simultaneously diagonalize all operators in ${\cal Z}^B_L$ on the Hilbert space
${\cal H}_{B \sqcup B}$.

Interestingly, we can  use  the bound \eqref{eq:trPbound} to show that any central operator $z \in {\cal Z}^B_L$ has a purely discrete spectrum. For clarity, we state this as the following lemma:
\begin{lemma}\label{lemma:discrete}
The spectrum of any  $z \in {\cal Z}^B_L$ is purely discrete in the sense that  ${\cal H}_{B \sqcup B}$ is the closure of the linear span of all normalizable eigenstates of $z$.
\end{lemma}
\begin{proof}
Without loss of generality, let us take $z$ to be self-adjoint.   To establish the desired result, let us first write ${\cal H}_{B \sqcup B} ={\cal H}_{B \sqcup B}^D \oplus  {\cal H}_{B \sqcup B}^{\perp D}$ where ${\cal H}_{B \sqcup B}^D$ with superscript $D$ for ``discrete spectrum'') is the closure of the linear span of all normalizable eigenstates of $z$.  Note that the projection $P_{z_0}$ onto normalizable states with eigenvalue $z_0$ also defines a central element of our von Neumann algebra ${\cal A}^B_L$.  Summing over all such $z_0$ then shows that the projection $P_D$ onto ${\cal H}_{B \sqcup B}^D$ again lies in ${\cal Z}^B_L$, so that the complementary projection $P_D^\perp  = \mathbb{1} - P_D$  onto ${\cal H}_{B \sqcup B}^{\perp D}$ must lie in ${\cal Z}^B_L$ as well.

Now assume that $P_D^\perp$ is not the zero operator. The semifinite property of our trace then implies that there is some non-zero positive operator $a \in \A^B_L$ with finite trace such that $P_D^\perp -a$ is positive, and thus in particular for which $a$ annihilates all states in ${\cal H}_{B \sqcup B}^D$.

Consider now the projection $P_{a > \epsilon}$ onto the part of the spectrum of $a$ with eigenvalues greater than $\epsilon$ for some $\epsilon > 0$ and note that $P_{a > \epsilon} \in {\cal A}^B_L$.    Since $a$ has finite trace and $a-\epsilon P_{a>\epsilon}$ is positive, the normality property of our trace then requires that $\tr (P_{a > \epsilon})$ is again finite.
Furthermore, since $a$ is not the zero operator, the spectral theorem (say, in the form of theorem 5.2.2 of \cite{KR1}) implies that $P_{a > \epsilon}$ must be non-vanishing for some $\epsilon > 0 $.

It will also be useful to construct the operator $z P_{a > \epsilon}$.  Since $z$ commutes with $P_{a > \epsilon}$, the operator $z P_{a > \epsilon}$ is self-adjoint and so can be diagonalized in the sense of the spectral theorem (see again theorem 5.2.2 of \cite{KR1}). Furthermore, $z P_{a > \epsilon}$ can have no normalizable eigenvector $|\psi\rangle$ with non-zero eigenvalue $z_0$ since then $P_{a > \epsilon} |\psi\rangle$ would be a (necessarily non-vanishing) normalizable eigenvector of $z$ in ${\cal H}_{B \sqcup B}^{\perp D}$.  Let us now define $\lambda_{max} : = \|z P_{a > \epsilon}\|$ to be the operator norm of $z P_{a > \epsilon}$. Note that $\lambda_{max}$ cannot be zero since then $z P_{a > \epsilon}$ would vanish and any $|\psi\rangle \in \H_{B \sqcup B}$ not annihilated by $P_{a > \epsilon}$ would define a normalizable eigenvector $P_{a > \epsilon} |\psi\rangle$ of $z$ in ${\cal H}_{B \sqcup B}^{\perp D}$ with eigenvalue $0$. Moreover, at least one of $\lambda_{max}$ or $-\lambda_{max}$ is a spectral value of $z P_{a > \epsilon}$ (see e.g.\ proposition 3.2.15 of \cite{KR1}). Without loss of generality, we assume that $\lambda_{max}$ is a spectral value of $z P_{a > \epsilon}$ (if not, simply replace $z P_{a > \epsilon}$ by $-z P_{a > \epsilon}$ below).  Then since $z P_{a > \epsilon}$ has no normalizable eigenvector of eigenvalue $\lambda_{max}$, the above spectral theorem implies that for any $\lambda_0 < \lambda_{max}$ and any positive integer $n \in \mathbb{Z}^+$ there are real numbers $\lambda_1, \lambda_2, \dots, \lambda_n$ with  $\lambda_0 < \lambda_1 < \lambda_2 < \dots < \lambda_n < \lambda_{max}$ for which the projections $P_{[\lambda_{i-1}, \lambda_{i}]}$  onto the spectral intervals $[\lambda_{i-1}, \lambda_{i}]$ of $z P_{a > \epsilon}$ are non-vanishing for $i=1,\dots n$.  If we choose $\lambda_0 > 0$, we have $P_{[\lambda_0, \lambda_{max}]} \le P_{a > \epsilon}$ (as any state annihilated by $P_{a > \epsilon}$ must be annihilated by $z P_{a > \epsilon}$, and thus also by $P_{[\lambda_0, \lambda_{max}]}$). Since $\tr(P_{a > \epsilon})$ is finite, normality of our trace again requires $\tr(P_{[\lambda_0, \lambda_{max}]})$ to be finite, which yields the bound
\begin{equation}
\sum_{i=1}^n \tr(P_{[\lambda_{i-1}, \lambda_{i}]}) \le \tr(P_{[\lambda_0, \lambda_{max}]}).
\end{equation}
But since all of these traces are positive (and, in particular, non-zero since the projections are non-trivial), for $n > \tr(P_{[\lambda_0, \lambda_{max}]})$ some
$P_{[\lambda_{i-1}, \lambda_{i}]}$ must have trace less than $1$.  This contradicts the bound \eqref{eq:trPbound}, so that $P_D^\perp$ must in fact vanish.  Thus $z$ has purely discrete spectrum in the sense that ${\cal H}_{B \sqcup B}$ is the closure of the linear span of all normalizable eigenstates of $z$.
\end{proof}

The analogous statement will again hold when we simultaneously diagonalize all central operators in ${\cal Z}^B_L$.
Let us denote the simultaneous eigenspaces by ${\cal H}^\mu_{B \sqcup B}$ for $\mu$ in some index set ${\cal I}$. In particular, for each $\mu \in {\cal I}$ there is a set of complex numbers $z_\mu$ such that the states $|\psi\rangle$ in  ${\cal H}^\mu_{B \sqcup B}$ are precisely the set of states for which $z|\psi\rangle = z_\mu |\psi\rangle$ for any $z \in {\cal Z}^B_L$. The Hilbert space ${\cal H}_{B \sqcup B}$ then decomposes as a direct {\it sum} (not a more general integral) over such eigenspaces:
\begin{equation}
\label{eq:HBBdecomp}
{\cal H}_{B \sqcup B} = \bigoplus_{\mu \in {\cal I}} \ {\cal H}^\mu_{B \sqcup B}.
\end{equation}
  There is of course a corresponding resolution of the identity on ${\cal H}_{B \sqcup B}$ in terms of orthogonal projections $P_\mu$ onto ${\cal H}^\mu_{B \sqcup B}$:
\begin{equation}
\label{eq:muresolve}
\mathbb{1}_{B \sqcup B} = \bigoplus_{\mu \in {\cal I}} \  P_\mu,
\end{equation}
where $P_\mu P_\nu = P_\mu\delta_{\mu \nu}$ and  $P_\mu^\dagger = P_\mu$.

The fact that ${\cal A}^B_L$ is weakly closed means that the projections $P_\mu$ lie in ${\cal A}^B_L$, and thus in fact also lie in the center ${\cal Z}^B_L$.  As a result,
the decomposition \eqref{eq:HBBdecomp} also has an analogue at the level of the von Neumann algebra ${\cal A}^B_L$.  To see this, simply note that for each $\mu \in {\cal I}$  we can define a subalgebra ${\cal A}^B_{L, \mu}$ of operators of the form $P_\mu a$ for $a \in {\cal A}^B_L$.
We may thus use the resolution of the identity \eqref{eq:muresolve} to write
\begin{equation}
\label{eq:Algdecomp}
{\cal A}^B_L = \bigoplus_{\mu \in {\cal I}}\  {\cal A}^B_{L,\mu}.
\end{equation}
It is also clear that each ${\cal A}^B_{L,\mu}$ annihilates any
${\cal H}^\nu_{B \sqcup B}$ with $\mu \neq \nu$, and that the subalgebras ${\cal A}^B_{L,\mu}$ are von Neumann algebras in their own right (acting on ${\cal H}^\mu_{B \sqcup B}$).

Furthermore, consider any operator $z_\mu$ in the center of ${\cal A}^B_{L,\mu}$.  Then since $P_\mu$ projects the full von Neumann algebra onto ${\cal A}^B_{L,\mu}$, the operator $P_\mu z_\mu$ will commute with all $a$ in the full von Neumann algebra ${\cal A}^B_{L}$ and thus lies in the original center.  But we have already diagonalized all operators in the center of ${\cal A}^B_{L}$. So, on the subspace
${\cal H}^{\mu}_{B \sqcup B}$, the operator $P_\mu z_\mu$ must act as a multiple of the identity $\mathbb{1}_\mu$. But on ${\cal H}^{\mu}_{B \sqcup B}$  the operator $P_\mu$ is already proportional to $\mathbb{1}_\mu$ (with non-zero coefficient), so this must be true of our $z_\mu$ as well.  It follows that each ${\cal A}^B_{L,\mu}$ is a von Neumann algebra with trivial center.

A von Neumann algebra with trivial center is known as a von Neumann factor.  Such factors can be classified as being of type I, II, or III, and our arguments above showed that each ${\cal A}^B_{L,\mu}$ must be of type I.

There is, of course, also a corresponding decomposition of ${\cal A}^B_R$.  In fact, since
${\cal A}^B_L$ and ${\cal A}^B_R$ are commutants of each other, their central subalgebras must define the {\it same} set of operators on ${\cal H}_{B \sqcup B}$; i.e., any operator on ${\cal H}_{B \sqcup B}$ that lies in the center of
${\cal A}^B_L$ must also lie in the center of ${\cal A}^B_R$.  The decomposition of  ${\cal A}^B_R$ into  ${\cal A}^B_{R,\mu}$ thus uses precisely the same index set ${\cal I}$, and it is associated with the identical decomposition of the Hilbert space \eqref{eq:HBBdecomp}.  Furthermore, the subalgebras ${\cal A}^B_{L,\mu}$, ${\cal A}^B_{R,\mu}$ are both type I von Neumann factors that are commutants of each other on the corresponding Hilbert space ${\cal H}_{B \sqcup B}^\mu$.

Now, it is also known that any type I von Neumann factor is isomorphic to the algebra ${\cal B}({\cal H})$ of all bounded operators on some Hilbert space ${\cal H}$.  Since this is true of both ${\cal A}^B_{L,\mu}$ and ${\cal A}^B_{R,\mu}$, the fact that these two algebras are commutants on ${\cal H}_{B \sqcup B}^\mu$ can be used to show that this Hilbert space admits a factorization
\begin{equation}
\label{eq:mufactor}
{\cal H}_{B \sqcup B}^\mu = {\cal H}_{B \sqcup B,L}^\mu \otimes {\cal H}_{B \sqcup B,R}^\mu,
\end{equation}
such that the action of every $a \in {\cal A}^B_{L,\mu}$ on \eqref{eq:mufactor} is of the form $a_L \otimes \mathbb{1}_{\mu,R}$, where $\mathbb{1}_{\mu,R}$ is the identity on ${\cal H}_{B \sqcup B,R}^\mu$.   Furthermore,
{\it any} bounded operator on \eqref{eq:mufactor} of the form $a_L \otimes \mathbb{1}_{\mu,R}$ is a member of ${\cal A}^B_{L,\mu}$. The operators
$a \in {\cal A}^B_{R,\mu}$ on \eqref{eq:mufactor} are analogously the set of operators of the form $\mathbb{1}_{\mu,L} \otimes  a_R$.
Since we can extend \eqref{eq:starop} to show that ${}^\dagger$ defines an anti-linear isomorphism between ${\cal A}^B_{L,\mu}$ and ${\cal A}^B_{R,\mu}$, it follows that the Hilbert space factor ${\cal H}_{B \sqcup B,L}^\mu$ is similarly isomorphic to ${\cal H}_{B \sqcup B,R}^\mu$.  Nevertheless, we maintain the labels $L,R$ for clarity below.

This is precisely the structure advertised in the introduction.  In particular, by restricting it to ${\cal H}_{B\sqcup B}$, any density matrix on the quantum gravity Hilbert space clearly defines a  density matrix $\rho_{B \sqcup B}$ on the sector ${\cal H}_{B\sqcup B}$.  For convenience, let us suppose that $\rho_{B \sqcup B}$ is normalized in the sense that it gives an expectation value of $1$ for the identity on ${\cal H}_{B\sqcup B}$.  This is equivalent to saying that the standard Hilbert space trace of $\rho_{B \sqcup B}$ (defined by summing diagonal matrix elements of  $\rho_{B \sqcup B}$  over an orthonormal basis of the Hilbert space ${\cal H}_{B\sqcup B}$) yields $1$.  This $\rho_{B \sqcup B}$ then defines density matrices $\rho^\mu$ for which
\begin{equation}
p_\mu \rho^\mu : = P_\mu \rho_{B \sqcup B} P_\mu,
\end{equation}
where the $p_\mu$ are probabilities given by the expectation values of the operators $P_\mu$ in the state $\rho_{B \sqcup B}$ and the $\rho^\mu$ are normalized density matrices on ${\cal H}^\mu_{B\sqcup B}$.

The key point is that each such $\rho^\mu$ now induces normalized density matrices $\rho^\mu_L, \rho^\mu_R$ on the Hilbert space factors ${\cal H}_{B \sqcup B,L}^\mu , {\cal H}_{B \sqcup B,R}^\mu$.  If one thinks of density matrices as positive linear functionals on the algebra of observables then $\rho^\mu_L, \rho^\mu_R$ are the restrictions of $\rho^\mu$ to the left and right von Neumann algebras ${\cal A}^B_{L,\mu}$, ${\cal A}^B_{R,\mu}$.  But one may equivalently think of
$\rho^\mu_L$ as the trace of $\rho^\mu$ over ${\cal H}_{B \sqcup B,R}^\mu$, and $\rho^\mu_R$ is similarly the trace of $\rho^\mu$ over ${\cal H}_{B \sqcup B,L}^\mu$.

This structure will allow us to discuss entropies in what one might call ``standard physics terms'' in section \ref{subsec:SfromA} below.  However, it will simplify our discussion of entropies if we first make a brief digression (section \ref{subsec:normoftrace}) in order to more carefully analyze the normalization of the trace $\tr$.

\subsection{The normalization of \texorpdfstring{$\tr$}{}}
\label{subsec:normoftrace}

We saw in section \ref{subsec:typeI} that the von Neumann algebra ${\cal A}^B_L$ defined by a diagonal Hilbert space sector ${\cal H}_{B \sqcup B}$ can be written in terms of type I von Neumann factors ${\cal A}^B_{L,\mu}$.  Since the operators in ${\cal A}^B_{L,\mu}$ are just operators in ${\cal A}^B_L$ of the form $P_\mu a$, the trace $\tr$ on ${\cal A}^B_{L}$ is also defined on positive operators in ${\cal A}^B_{L,\mu}$.

However,  as noted above, our ${\cal A}^B_{L,\mu}$ can be thought of as the algebra of all bounded operators on the Hilbert space factor ${\cal H}_{B \sqcup B,L}^\mu$.  Faithful, normal, semifinite traces on such algebras are known to be unique up to an overall normalization constant. In particular, there must be real numbers ${\cal C}_\mu > 0$ such that for all positive $a \in {\cal A}_{L, \mu}^B$ we have
\begin{equation}
\label{eq:Ltr}
{\cal C}_\mu \tr (a) = \Tr_{\mu}(a) : = \sum_{i} {}_L\langle i | a | i\rangle_L,
\end{equation}
where $|i \rangle_L$ is a (discrete) orthonormal basis\footnote{Such a basis exists for {\it any} Hilbert space; see e.g.\ \cite{RS}.  Thus, even at this point, we have had no need to assume that any of our Hilbert spaces are separable, though in any physical context it would be natural to assume that this is the case, especially in each subsector ${\cal H}_{B \sqcup B}^\mu$.} for   ${\cal H}_{B \sqcup B,L}^\mu$, and where the right-hand side defines the operation $\Tr_{\mu}$ on our von Neumann algebra.
Similarly,
 for all positive $a \in {\cal A}_{R, \mu}^B$ we have
\begin{equation}
\label{eq:Rtr}
{\cal C}_\mu \tr (a) = \Tr_{\mu}(a) : = \sum_{i} {}_R\langle i | a | i\rangle_R,
\end{equation}
where $|i \rangle_R$ is a (discrete) orthonormal basis for   ${\cal H}_{B \sqcup B,R}^\mu$.  In particular, since the left and right algebras and Hilbert spaces are isomorphic, and since the trace on positive operators is invariant under this isomorphism (i.e., the trace on the right algebra acts in just the same way as the trace on the left algebra), the constants ${\cal C}_\mu$ in \eqref{eq:Ltr} are identical to the constants ${\cal C}_\mu$ in \eqref{eq:Rtr}.

The trace inequality \eqref{eq:TrInvN} can be seen to constrain the values of the constants ${\cal C}_\mu$.   In particular, let $P$ be a one-dimensional projection onto a state in  ${\cal H}_{B \sqcup B,L}^\mu$.  Then $\Tr_{\mu}(P)=1$, so
$\tr(P) = 1/{\cal C}_\mu$.  But we saw in section \ref{subsec:typeI} that any non-zero projection must have $\tr(P) \ge 1$.  Thus we see that the constants ${\cal C}_\mu$ satisfy
\begin{equation}
\label{eq:Cmubound}
{\cal C}_\mu \le 1, \ \ \ \forall \mu.
\end{equation}

However, there are also further constraints on the constants ${\cal C}_\mu$.
To see this, recall that section \ref{subsec:exTrIn} derived the trace inequality implying \eqref{eq:Cmubound} by considering a standard consequence (the Cauchy-Schwarz inequality) of positivity of the inner product on the Hilbert space sector ${\cal H}_{B \sqcup B \sqcup B \sqcup B}$ associated with 4 copies of the boundary $B$.  It is thus natural to ask if further information about the $C_\mu$ can be obtained by considering Hilbert spaces associated with even more copies of $B$.

This turns out to be the case.  To proceed, note that by Corollaries \ref{cor:triprel} and \ref{cor:tenprod} from section \ref{subsec:exTrIn}, for any allowed $B$ and any operators $a_1, \dots, a_n \in {\cal A}^B_L$ with finite $\tr(a_i^\dag a_i)$ there is a state that we may call $|a_1, \dots, a_n\rangle$ in the Hilbert space  ${\cal H}_{\sqcup_{i=1}^{2n} B}$ associated with $n$ copies of the boundary $B \sqcup B$.  Furthermore, the Cauchy-Schwarz inequality used to derive the trace inequality \eqref{eq:TrInvN} follows in the standard way from positivity of the norm squared of the state $|a_{L_1, R_1}, b_{L_2, R_2} \rangle - |a_{L_2, R_1}, b_{L_1, R_2} \rangle$ in ${\cal H}_{\sqcup_{i=1}^{2n} B}$ with $n=2$.  One would thus like to investigate the positivity of analogous totally anti-symmetric combinations of states for general $n$.

For simplicity, we do so here only for the simple case when all operators $a_1, \dots, a_n$ agree and where they are equal to a finite-dimensional projection $P$.   In particular, we will establish the following lemma which strengthens the bound \eqref{eq:trPbound}:
\begin{lemma}
\label{lemma:Pquant}
For any non-zero finite-dimensional projection $P \in {\cal A}_L^B$, the trace $\tr(P)$ is a positive integer.
\end{lemma}
\begin{proof}
To begin,  recall that our trace is normal and semifinite.  These properties imply that for any non-zero positive operator (and thus any non-zero projection $P$) there is some non-zero positive operator $Q$ with $P-Q$ positive such that $Q$ has finite non-zero trace.  In particular, for any one-dimensional projection $P$ on a Hilbert space, any $Q$ with $P-Q$ positive must annihilate states orthogonal to the image of $P$.  As a result, $Q = \alpha P$ for some real number $\alpha$ with $1 \ge \alpha > 0$.  As a result, $\tr(P) = \alpha^{-1} \tr (Q)$ must be finite, showing that any one-dimensional projection does indeed have a finite trace.  The same is then necessarily true for any finite-dimensional projection $P$ by linearity.

It is then useful to define the notation
\begin{equation}
|P^{\otimes n}\rangle : = |P_{L_1 R_1}, \dots ,P_{L_n R_n} \rangle.
\end{equation}
For any permutation $\pi$ on $n$ labels we may also define the (left) $n$-boundary swap operator ${\cal S}^L_\pi$ which acts on ${\cal H}_{\sqcup_{i=1}^{2n} B}$ by permuting the labels of the $n$ left boundaries as dictated by $\pi$. We then wish to consider the norm of the state
\begin{equation}
\label{eq:Pwedge}
|P^{\wedge n}\rangle := \sum_{\pi} (-1)^\pi {\cal S}^L_\pi |P^{\otimes n}\rangle,
\end{equation}
where the sum is over all $n$-object permutations and $(-1)^\pi$ is $+1$ ($-1$) for even (odd) permutations $\pi$.

To compute this norm, it will be useful to write \eqref{eq:Pwedge} in the form
\begin{equation}
\label{eq:Pwedge2}
|P^{\wedge n}\rangle = \sum_{i=1}^n  (-1)^{\delta_{i,1}-1}{\cal S}_{L_1,L_i} \left|P_{L_1,R_1}, P^{\wedge (n-1)}_{L_2,R_2,\dots, L_n, R_n} \right\rangle
\end{equation}
where $\delta_{i,1}$ is a Kronecker $\delta$ symbol and,  in analogy with section \ref{subsec:exTrIn}, we have defined the operators ${\cal S}_{L_1,L_i}$ that swap the labels on the 1st and $i$th left boundaries (so that, in particular, ${\cal S}_{L_1,L_1}$ is the identity). The norm-squared of \eqref{eq:Pwedge2} is

\begin{eqnarray}
\label{eq:Pwedgecomp}
\langle P^{\wedge n}| P^{\wedge n}\rangle &=& \sum_{i,j=1}^n  (-1)^{\delta_{i,1}+\delta_{j,1}} \Braket{P_{L_1,R_1}, P^{\wedge (n-1)}_{L_2,R_2,\dots, L_n, R_n} | {\cal S}_{L_1,L_i}{\cal S}_{L_1,L_j} |P_{L_1,R_1}, P^{\wedge (n-1)}_{L_2,R_2,\dots, L_n, R_n} } \cr
&=& n \langle P | P \rangle \langle P^{\wedge (n-1)} | P^{\wedge (n-1)} \rangle - n (n-1) \langle  P^{\wedge (n-1)} | P^{\wedge (n-1)} \rangle \cr
 &=& n \big( \langle P | P \rangle - (n-1) \big) \langle  P^{\wedge (n-1)} | P^{\wedge (n-1)} \rangle,
\end{eqnarray}
where we pass from the first to the second line by using that, when $i$, $j$, and $1$ are all distinct, we may write
\begin{eqnarray}
&& \Braket{P_{L_1,R_1}, P^{\wedge (n-1)}_{L_2,R_2,\dots, L_n, R_n} | {\cal S}_{L_1,L_i}{\cal S}_{L_1,L_j} |P_{L_1,R_1}, P^{\wedge (n-1)}_{L_2,R_2,\dots, L_n, R_n} } \cr
&=& \Braket{P_{L_1,R_1}, P^{\wedge (n-1)}_{L_2,R_2,\dots, L_n, R_n} | {\cal S}_{L_i,L_j}{\cal S}_{L_1,L_i} |P_{L_1,R_1}, P^{\wedge (n-1)}_{L_2,R_2,\dots, L_n, R_n} } \cr
&=& -\Braket{P_{L_1,R_1}, P^{\wedge (n-1)}_{L_2,R_2,\dots, L_n, R_n} | {\cal S}_{L_1,L_i} |P_{L_1,R_1}, P^{\wedge (n-1)}_{L_2,R_2,\dots, L_n, R_n} } \cr
&=& -\Braket{P^{\wedge (n-1)}_{L_2,R_2,\dots, L_n, R_n} | P_{L_i} P_{L_i}^\dag |P^{\wedge (n-1)}_{L_2,R_2,\dots, L_n, R_n} } \cr
&=& -\Braket{P^{\wedge (n-1)}_{L_2,R_2,\dots, L_n, R_n} |P^{\wedge (n-1)}_{L_2,R_2,\dots, L_n, R_n} } = -\langle  P^{\wedge (n-1)} | P^{\wedge (n-1)} \rangle.
\end{eqnarray}
Here in going to the 3rd line, we used that $\left|P^{\wedge (n-1)}_{L_2,R_2,\dots, L_n, R_n}\right\rangle$ is odd under ${\cal S}_{L_i,L_j}$ when $i\ne j$. The 4th line then follows from a derivation similar to how we derived \eqref{eq:P1P2ip2}; see again figure \ref{fig:4bswap}. In particular, the swap operator ${\cal S}_{L_1,L_i}$ on the 3rd line effectively pulls the $P_{L_1,R_1}$ in the bra and ket into the middle of the 4th line, where they now act at $L_i$. In going to the 5th line, we used that $\left|P^{\wedge (n-1)}_{L_2,R_2,\dots, L_n, R_n}\right\rangle$ is invariant under any $P_{L_i}^\dag =P_{L_i}$ since $P^2=P$.

Now \eqref{eq:Pwedgecomp} must be non-negative. But $\langle P^{\wedge (n-1)} | P^{\wedge (n-1)} \rangle$ is also non-negative. Therefore, for any $n\in {\mathbb Z}^+$ with $n>1$, unless $ | P^{\wedge (n-1)} \rangle = 0$, we must have
\begin{equation}
\label{eq:endP}
n-1 \le \langle P | P \rangle = \tr(P^2) = \tr(P).
\end{equation}
But finiteness of $\tr(P)$ means that \eqref{eq:endP} must fail for some $n$.

The above argument then requires the state
$ |P^{\wedge (n-1)} \rangle$ to vanish for such $n$.
On the other hand, \eqref{eq:Pwedgecomp} also calculates the norm of $ |P^{\wedge n} \rangle $ recursively in terms of the norms of the states $ |P^{\wedge m} \rangle$ with $m < n$.  We thus see that this norm can vanish only if $\tr(P) = \tr(P^2) = \langle P | P \rangle$ is a non-negative integer, so that the trace of any non-zero finite-dimensional projection $P$ must be some positive integer $n$.
\end{proof}

We may now apply Lemma \ref{lemma:Pquant} to a one-dimensional projection $P$ onto a state in some ${\cal H}_{B \sqcup B,L}^\mu$.  Writing $\tr(P)=n$, we then have
\begin{equation}
1 = \Tr_{\mu}(P)= {\cal C}_\mu \tr(P)= n{\cal C}_\mu,
\end{equation}
which requires each constant ${\cal C}_\mu$ to be of the form
\begin{equation}
\label{eq:quantC}
{\cal C}_\mu^{-1} = n_\mu  \in \mathbb{Z}^+.
\end{equation}

The quantization condition \eqref{eq:quantC} allows us to give a particularly nice physical and mathematical description of our trace $\tr$. For $n \in \mathbb{Z}^+$, let us introduce the $n$-dimensional Hilbert spaces ${\cal H}_n$.  We may then define the extended Hilbert space factors
\begin{eqnarray}
\widetilde {\cal H}_{B \sqcup B,L}^\mu &: =& {\cal H}_{B \sqcup B,L}^\mu \otimes {\cal H}_{n_\mu}, \cr
\widetilde {\cal H}_{B \sqcup B,R}^\mu &: =& {\cal H}_{B \sqcup B,R}^\mu \otimes {\cal H}_{n_\mu},
\end{eqnarray}
and the modified summed Hilbert space
\begin{equation}
\label{eq:BBtilde}
\widetilde {\cal H}_{B \sqcup B} :=\bigoplus_{\mu \in {\cal I}}  \, \left(\widetilde {\cal H}_{B \sqcup B,L}^\mu \otimes \widetilde {\cal H}_{B \sqcup B,R}^\mu \right).
\end{equation}
Let us also define the trace $\widetilde{\Tr}_\mu$ on the algebra ${\cal B}\left(\widetilde {\cal H}_{B \sqcup B,L}^\mu\right)$ of bounded operators on $\widetilde {\cal H}_{B \sqcup B,L}^\mu$ by the standard Hilbert space sum
\begin{equation}
\widetilde {\Tr}_{\mu} (\tilde a)  := \sum_{\tilde i} {}_L \langle\tilde i | \tilde a | \tilde i \rangle{}_L
\end{equation}
in terms of an orthonormal basis $|\tilde i\rangle_L$ of $\widetilde {\cal H}_{B \sqcup B,L}^\mu$ for any $\tilde a \in {\cal B}\left(\widetilde {\cal H}_{B \sqcup B,L}^\mu\right)$.  Then any $a \in {\cal A}^B_{L, \mu}$ defines an operator
\begin{equation}
\tilde a : = \left(a \otimes \mathbb{1}_{n_\mu} \right)
\end{equation}
on $\widetilde {\cal H}_{B \sqcup B,L}^\mu= {\cal H}_{B \sqcup B,L}^\mu \otimes {\cal H}_{n_\mu}$ and satisfies the relation
\begin{equation}
\label{eq:tildeTrreln}
\tr (a) = \widetilde{\Tr}_\mu(\tilde a) =  \sum_{\tilde i}{}_L \langle\tilde i | \tilde a | \tilde i \rangle{}_L.
\end{equation}
Furthermore, the operators $\tilde a$ again define a faithful representation of ${\cal A}^B_{L, \mu}$ on $\widetilde {\cal H}_{B \sqcup B,L}^\mu$.  We have thus found a representation of our trace $\tr$ in terms of a standard Hilbert space sum over diagonal matrix elements\footnote{We have done this here separately for each $\mu$, but we are free to combine the operations ${\Tr}_\mu$ by summing over $\mu$.  This will be discussed further in section \ref{subsec:SfromA}.}.    Physically, one might say that our trace $\tr$ gives the Hilbert space trace in a context where there are `hidden sectors' ${\cal H}_{n_\mu}$ on which operators in ${\cal A}^B_{L, \mu}$ act as the identity.

We now make a final further comment by again using Corollary \ref{cor:tenprod}.  In particular, let us consider  the tensor product ${\cal H}_B \otimes {\cal H}_B$ of two Hilbert spaces that are each associated with a single copy of $B$.  Corollary \ref{cor:tenprod} then guarantees that ${\cal H}_B \otimes {\cal H}_B$ is a
subspace of the Hilbert space ${\cal H}_{B \sqcup B}$ associated with two copies of $B$.  Furthermore, it is clear that any $a \in {\cal A}^B_L$ must act on this space as some $a_L \otimes \mathbb{1}_{{\cal H}_B}$, and analogously for elements of ${\cal A}^B_R$.  As a result, this ${\cal H}_B \otimes {\cal H}_B$ must be one of the terms ${\cal H}_{B \sqcup B}^\mu = {\cal H}_{B \sqcup B,L}^\mu \otimes {\cal H}_{B \sqcup B,R}^\mu$ in the decomposition \eqref{eq:HBBdecomp}.  Let us denote this term by {$\mu=\otimes$.  Then in particular we have
\begin{equation}
{\cal H}_{B \sqcup B,L}^{\otimes} = {\cal H}_B,  \ \ \ {\rm and} \ \ \  {\cal H}_{B \sqcup B,R}^{\otimes} = {\cal H}_B.
\end{equation}

Assuming that ${\cal H}_B$ is not empty, this observation will also tell us that the trace normalization factor ${\cal C}_{\otimes} = 1/n_{\otimes}$ associated with this subspace must be $1$.  To see this, consider any $a \in \U{Y}^d_{B}$ (associated with just one copy of $B$) that is not in the null space ${\cal N}_{B}$.  This $a$ of course defines a non-zero state $|a\rangle \in {\cal H}_B$.  For convenience, let us normalize $a$ so that $\langle a | a \rangle =1$.  From this $a$ we can of course construct $a \sqcup a^* \in \U{Y}^d_{B} \otimes \U{Y}^d_{B} \subset \U{Y}^d_{B\sqcup B}$, which then defines an operator $\hat P_{a,L} := \widehat{(a \sqcup a^*})_L \in \hat A^{B \sqcup B}_L$.  We use the notation $\hat P_{a,L}$ because a short computation (see figure \ref{fig:Pacomp}) yields
\begin{equation}
\hat P_{a,L}^\dagger = \hat P_{a,L},  \ \ \ {\rm{and}} \ \ \ \hat P_{a,L}^2 = \hat P_{a,L},
\end{equation}
showing that it is a projection. We also find
\begin{equation}
\hat P_{a,L} \left( |c\rangle \otimes |d\rangle \right) =  \left( \langle a |c\rangle\right) \left( | a\rangle \otimes |d\rangle \right), \quad \forall |c\rangle \otimes |d\rangle \in {\cal H}_B \otimes {\cal H}_B.
\end{equation}
Furthermore, it is clear that for any $b \in \U{Y}^d_{B\sqcup B}$ (with two $B$-boundaries), the state $\hat P_{a,L} |b\rangle = |a\rangle \otimes b^t| a^*\rangle$ lies in ${\cal H}^\otimes_{B \sqcup B} = {\cal H}_B \otimes {\cal H}_B$; see again figure \ref{fig:Pacomp}.  Self-adjointness then requires that $\hat P_{a,L}$ annihilate the orthogonal complement of ${\cal H}^{\otimes}_{B \sqcup B}$.  In particular, we see that any such operator is of the form
\begin{equation}
\hat P_{a,L} =  \left(|a\rangle \langle a| \otimes \mathbb{1}_{{\cal H}_B} \right)  P_{\otimes},
\end{equation}
where $P_{\otimes}$ is the projection onto the product sector ${\cal H}^\otimes_{B \sqcup B} = {\cal H}_B \otimes {\cal H}_B$ in \eqref{eq:HBBdecomp}.
\begin{figure}[t]
        \centering
\includegraphics[scale=0.78]{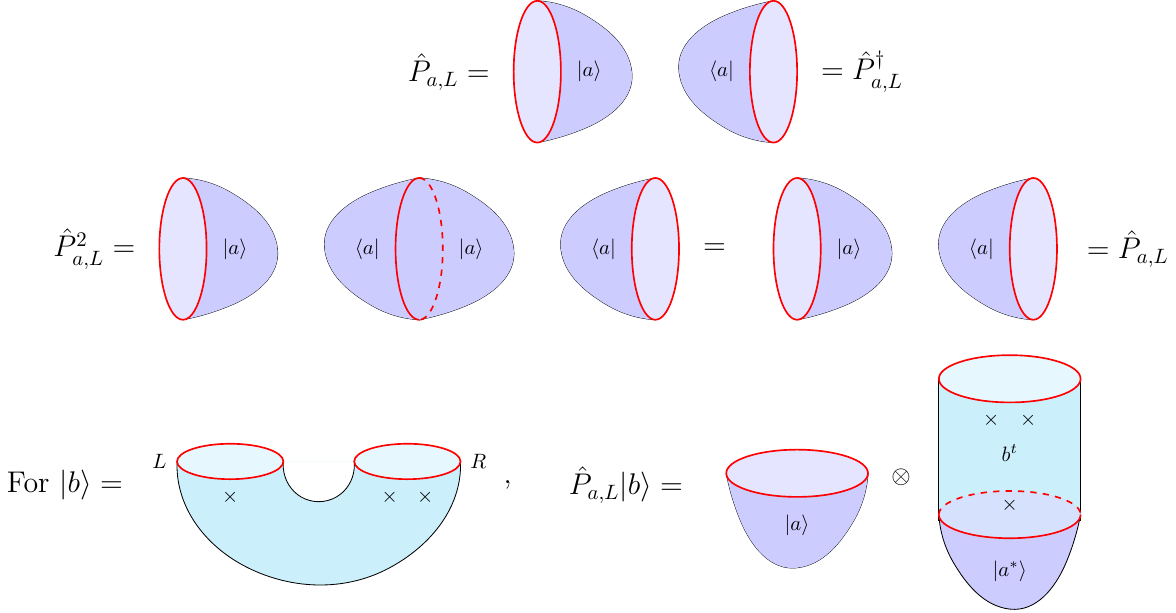}\caption{ Consider any $a \in \U{Y}^d_{B}$ (associated with the `one-boundary' state $|a\rangle$) that is normalized so that $\langle a | a \rangle=1$.
 The operator $\hat P_{a,L}$ is a projection since it is manifestly self-adjoint (top line) and squares to itself (second line).  Furthermore, for any $|b\rangle \in H_{B \sqcup B}$ (not necessarily in $H_B \otimes H_B$), we have $\hat P_{a,L} |b\rangle = |a\rangle \otimes b^t| a^*\rangle$ (bottom line).   Thus $\hat P_{a,L} |b\rangle$ clearly lies in ${\cal H}_B \otimes {\cal H}_B$.}
\label{fig:Pacomp}
\end{figure}
As a result, the trace operation $\Tr_{\otimes}$ defined by ${\cal H}_{B \sqcup B,L}^{\otimes} = {\cal H}_B$
is well-defined on $\hat P_{a,L}$ and we find
\begin{equation}
\Tr_{\otimes} (\hat P_{a,L} ) = \langle a | a \rangle = \tr(\hat P_{a, L}),
\end{equation}
where the last step is again apparent from figure \ref{fig:Pacomp}.   We must thus also have ${\cal C}_{\otimes}=1 = n_\otimes$ as claimed above.

\subsection{Hidden sectors and entropy from algebras}
\label{subsec:SfromA}

As noted above, the extended Hilbert space $\widetilde {\cal H}_{B \sqcup B}$ is mathematically useful.  We now turn to the question of whether this space may be physically useful as well.

To do so, let us suppose for the moment that the $B\sqcup B$ sector of the physical Hilbert space of our quantum gravity theory were actually of the form \eqref{eq:BBtilde}, and in particular that it contained factors ${\cal H}_{n_\mu}$ in some $\widetilde {\cal H}_{B \sqcup B}^\mu$ such that all of the operators defined by our path integral act trivially on these
${\cal H}_{n_\mu}$.   Note that we suppose this to be true not only for operators in our algebras ${\cal A}^B_L$ and ${\cal A}^B_R$, but also for any operator on ${\cal H}_{B \sqcup B}$ defined by any surface in $Y^d_{B \sqcup B \sqcup B \sqcup B}$  (i.e., defined by any surface that has four $B$-boundaries), and in fact even for operators defined by surfaces in $Y^d_{B \sqcup B \sqcup \tilde B}$ for general $\tilde B$ that map ${\cal H}_{B \sqcup B}$ to other sectors of the Hilbert space.   So long as these are the only observables to which we have access, we will never know that such `hidden sectors' actually exist. Furthermore, none of our observables would be able to change the parts of the state associated with such hidden sectors no matter how hard we might try.

Indeed, let us suppose that for each $\mu \in {\cal I}$ (i.e., for each term ${\cal H}_{B \sqcup B}^\mu$ in the decomposition \eqref{eq:HBBdecomp}) there is a preferred (normalized) maximally entangled state $|\chi_\mu\rangle \in {\cal H}_{n_\mu} \otimes {\cal H}_{n_\mu}$.  We can use such states to map any state $|\psi\rangle \in {\cal H}_{B \sqcup B}$ isometrically to the state
\begin{equation}
|\tilde \psi\rangle : = \bigoplus_{\mu \in {\cal I}} \   \left( |\chi_\mu\rangle \otimes P_\mu | \psi\rangle \right) \in \widetilde  {\cal H}_{B \sqcup B}.
\end{equation}
This is the relation that arises when our original ${\cal H}_{B \sqcup B}$ defines a code subspace of
$\widetilde{\cal H}_{B \sqcup B}$ for a quantum error correcting code\footnote{\label{foot:sum}More precisely, it is a sum of such codes, since even our `physical Hilbert space' $\widetilde{\cal H}_{B \sqcup B}$ does not necessarily factorize, but is only required to be a direct sum of such factors.} with two-sided recovery associated with the algebras ${\cal A}^B_L$ and ${\cal A}^B_R$ \cite{Harlow:2016vwg}.  In particular, if we call the above isometric embedding $\chi: {\cal H}_{B \sqcup B} \rightarrow \widetilde{\cal H}_{B \sqcup B} $, then $\chi$ can be used to translate operators on ${\cal H}_{B \sqcup B}$ into operators that act on the image of $\chi$ in
$\widetilde  {\cal H}_{B \sqcup B}$.

The insertion of the maximally entangled state $|\chi_\mu\rangle$ will clearly lead to differences in quantitative measures of left-right entanglement as defined by the Hilbert spaces ${\cal H}_{B \sqcup B}$ and $\widetilde {\cal H}_{B \sqcup B}$.  The two descriptions of the Hilbert space thus also clearly lead to different notions of entropy.  This is in particular familiar from the discussion of \cite{Harlow:2016vwg}.  Our setting here is slightly more general than that of \cite{Harlow:2016vwg} (see again footnote \ref{foot:sum}), so we will postpone writing detailed formulas until we define further notation and terminology.

In the theory of von Neumann algebras one can introduce a notion of entropy whenever one has a faithful normal semifinite trace $\tr$ on (positive elements of) a von Neumann algebra ${\cal A}$.  For positive $a \in {\cal A}$ with $\tr(a) =1$, one may attempt to define
\begin{equation}
\label{eq:VNentropy}
S_{vN}(a) : = \tr(-a\ln a).
\end{equation}
Since $a$ is a positive bounded (and thus self-adjoint) operator on some Hilbert space, the operator $a \ln a$ can be defined using the spectral representation of $a$, and since it is bounded, it must also lie in ${\cal A}$.  In a truly general von Neumann algebra the operator $-a\ln a$ need not be positive, which means that \eqref{eq:VNentropy} is not obviously well-defined for all such positive $a$.  But \eqref{eq:TrInvN} and \eqref{eq:triprel} imply that the operator norm $\|a\|$ of $a$ is bounded by $\sqrt{\tr(a^\dag a)} = \sqrt{\tr(a^2)}$ which is in turn bounded by  $\sqrt{(\tr\, a)^2}=1$, so in our context $-a\ln a$ is  positive and \eqref{eq:VNentropy} is well-defined so long as we allow it to take the value $+\infty$.

On the other hand, in physics we typically wish to discuss entropies defined by states, which in the present context we take to mean (normalized) pure states $|\psi\rangle$ in some Hilbert space ${\cal H}$.  The connection to the above entropy for positive elements of a von Neumann algebra ${\cal A}$ acting on ${\cal H}$ is of course through the concept of a density matrix which, in the general context, is more properly called a density operator.  The point is that any physics encoded in $|\psi\rangle$ that can be extracted using observables in ${\cal A}$ is described by the expectation values $\langle \psi | a | \psi \rangle$ for $a \in {\cal A}$.  In particular, since ${\cal A}$ is closed under the product operation, this includes all possible correlation functions (which are described by the case $a = a_1 a_2 \dots a_n$).
In fact, since any $a$ can be written in terms of its self-adjoint and anti-self-adjoint parts, and since self-adjoint operators can be written in a standard spectral representation in terms of projections onto their spectral intervals, it suffices to know $\langle \psi | a | \psi \rangle$ for all projections $a \in {\cal A}$.  We will be slightly more general than this below, but it will still be useful to henceforth restrict discussion of such expectation values to positive operators $a \in {\cal A}$.

Suppose now that we are given a (faithful normal semifinite) trace $\tr$ on ${\cal A}$. For the purposes of computing such expectation values, we can replace $|\psi\rangle$ with a density operator $\rho_\psi \in {\cal A}$ if we can find a positive $\rho_\psi$ such that for all positive $a \in {\cal A}$ we have
\begin{equation}
\langle \psi | a | \psi \rangle = \tr\left( \rho_\psi^{1/2}a   \rho_\psi^{1/2} \right),
\end{equation}
where the positive square root $\rho_\psi^{1/2}$ of the positive operator $\rho_\psi$ can as usual be defined using a spectral decomposition.
In familiar physics contexts the cyclic property of the trace would be used to write the right-hand side as $\tr(\rho_\psi a)$, but our trace does not allow this since it is defined only on positive operators and $\rho_\psi a$ need not be positive.  When such a $\rho_\psi$ exists, we can use \eqref{eq:VNentropy} to define an entropy on ${\cal A}$ for the state $|\psi\rangle \in {\cal H}$.

Let us now apply this discussion to states $|\psi\rangle$ in one of our diagonal Hilbert space sectors ${\cal H}_{B \sqcup B}$.  Then we in fact have {\it three} potentially useful notions of traces on e.g.\ ${\cal A}^B_L$.  The first is the trace $\tr$ defined by \eqref{eq:alttrace5}.  The second is defined by the collection of Hilbert space traces $\Tr_\mu$. Here we make use  of the fact that the decomposition of \eqref{eq:Algdecomp} of ${\cal A}^B_L$ in terms of factors allows us to analogously decompose any $a \in {\cal A}^B_L$ as $a = \oplus_{\mu \in {\cal I}} \  a_\mu$.  We may thus define a trace on ${\cal A}^B_L$ by simply summing the traces $\Tr_\mu$ over $\mu$, writing
\begin{equation}
\label{eq:BigTr}
\Tr (a) : = \sum_{\mu \in {\cal I}} \ \Tr_\mu(a_\mu).
\end{equation}

We could also have chosen to insert some additional positive coefficient $f(\mu)$ that changes the weights assigned to each $\mu$ in \eqref{eq:BigTr}, but we have explicitly chosen {\it not} to do so in defining $\Tr$.  However, our third trace $\widetilde{\Tr}$ is equivalent to such a reweighting of \eqref{eq:BigTr} and is defined by the expression
\begin{equation}
\label{eq:BigTrtilde}
\widetilde{\Tr} (a) : =  \sum_{\mu \in {\cal I}} \ \widetilde{\Tr}_\mu(a_\mu) =  \sum_{\mu \in {\cal I}} \ n_\mu \Tr_\mu(a_\mu),
\end{equation}
where the $n_\mu$ are the positive integers defined in section \ref{subsec:normoftrace}.
As noted in that section, for any $a_\mu$ we in fact have $\tr(a_\mu) = \widetilde {\Tr}_\mu(a_\mu)$, which by the linearity of $\tr$ means that our first trace  $\tr$ is identical on ${\cal A}^B_L$ to our third trace $\widetilde{\Tr}$.  Nevertheless, it will be useful to continue to use both symbols  $\tr$ and $\widetilde{\Tr}$  below to allow us to emphasize different conceptual points of view.

Now, given a normalized state $|\psi \rangle \in {\cal H}_{B \sqcup B}$, we can use \eqref{eq:HBBdecomp} to write $|\psi\rangle =  \sum_{\mu \in {\cal I}} \ |\psi_\mu \rangle$ with $|\psi_\mu \rangle \in {\cal H}_{B \sqcup B}^\mu$.  Then for any $a \in {\cal A}^B_L$ we have
\begin{equation}
\label{eq:aexmusum}
\langle \psi |a | \psi \rangle = \sum_{\mu \in {\cal I}} \  \langle \psi_\mu |a | \psi_\mu \rangle.
\end{equation}
Let us now ask whether we can construct a density operator $\rho_\psi$ such that the expectation values \eqref{eq:aexmusum} are reproduced using the $\Tr$ operation through
\begin{equation}
\label{eq:Trrho}
\langle \psi |a | \psi \rangle = \Tr(\rho_\psi^{1/2}a \rho_\psi^{1/2}).
\end{equation}
Working in any given $\mu$-sector, the factorization of the Hilbert space ${\cal H}_{B \sqcup B}^\mu$ into ${\cal H}_{B \sqcup B,L}^\mu$ and ${\cal H}_{B \sqcup B,R}^\mu$, and the fact that ${\cal A}^B_{L,\mu}$ is precisely the algebra of bounded operators on ${\cal H}_{B \sqcup B,L}^\mu$, mean that we can do this via the usual computation that considers the operator
$| \psi_\mu \rangle\langle \psi_\mu|$ on ${\cal H}_{B \sqcup B}^\mu$  and then traces over ${\cal H}_{B \sqcup B,R}^\mu$; i.e., we can define $\rho^\mu_\psi$ as an operator on ${\cal H}_{B \sqcup B,L}^\mu$ by giving a formula for its matrix elements between states $|\alpha\rangle_L$ and $|\beta \rangle_L$ in  ${\cal H}_{B \sqcup B,L}^\mu$. To do so, it will be useful to introduce an orthonormal basis $|i\rangle_R$ for  ${\cal H}_{B \sqcup B,R}^\mu$, and to use the notation
\begin{equation}
|\alpha,i\rangle{}_{LR}: = |\alpha \rangle_L \otimes |i\rangle_R
\end{equation}
for the tensor product of any $|\alpha \rangle_L \in {\cal H}^\mu_{B\sqcup B, L}$ and the basis state $|i\rangle_R$.  We will also introduce $p_\mu = \langle \psi_\mu | \psi_\mu \rangle$. Working in those $\mu$ sectors where $p_\mu \ne 0$, we take the matrix elements of $\rho_\psi^\mu$ between the above states to be
\begin{equation}
\label{eq:rhopsimatrixelements}
{}_{L}\langle \alpha | \rho_\psi^\mu |\beta \rangle_L := (p_\mu)^{-1} \sum_i {}_{LR}\langle \alpha,i| \psi_\mu \rangle  \langle \psi_\mu |\beta, i\rangle{}_{LR}.
\end{equation}
Normalizability of $|\psi_\mu\rangle$ implies this sum to converge, and $\rho_\psi^\mu$ is positive since its expectation value in any state is given by setting $\alpha=\beta$ in \eqref{eq:rhopsimatrixelements}, in which case the right-hand side is a sum of non-negative terms.  We also clearly have $\Tr_\mu (\rho_\psi^\mu)=1$, so $\rho_\psi^\mu$ is bounded.    Defining
\begin{equation}
\rho_\psi = \bigoplus_{\mu \in {\cal I}} \  \left(p_\mu \, \rho_\psi^\mu\right)
\end{equation}
then yields \eqref{eq:Trrho} in the usual way as desired.

The above construction of $\rho_\psi \in {\cal A}^B_L$ using the Hilbert space trace $\Tr$ may seem natural.  But one can of course repeat precisely the same construction using the state $|\tilde \psi \rangle = \chi |\psi\rangle$ in the enlarged Hilbert space $\widetilde {\cal H}_{B \sqcup B}$ that includes the hidden sectors ${\cal H}_{n_\mu}$.  In this case we use the trace $\widetilde \Tr$ to write the final result in the form
\begin{equation}
\label{eq:tildeTrrho}
\langle \psi |a | \psi \rangle = \widetilde\Tr(\tilde \rho_\psi^{1/2}a \tilde \rho_\psi^{1/2}),
\end{equation}
where $\tilde \rho_\psi =  \oplus_{\mu \in {\cal I}} \  \left(p_\mu \, \tilde \rho_\psi^\mu\right)$ for  $\tilde \rho_\psi^\mu$ defined by
\begin{equation}
\label{eq:tilderhopsimatrixelements}
{}_{L}\langle \tilde \alpha | \tilde \rho_\psi^\mu |\tilde \beta \rangle_L := (p_\mu)^{-1} \sum_{\tilde i} {}_{LR}\langle \tilde \alpha,\tilde i| \tilde \psi_\mu \rangle  \langle \tilde \psi_\mu |\tilde \beta, \tilde i\rangle{}_{LR},
\end{equation}
where $|\tilde \alpha\rangle_L, |\tilde \beta\rangle_L$ are states in $\widetilde{\cal H}_{B\sqcup B,L}^\mu$, the states $|\tilde i \rangle_R$ form an orthonormal basis for
$\widetilde{\cal H}_{B\sqcup B,R}^\mu$, and $|\tilde \psi_\mu \rangle$ is a state in $\widetilde{\cal H}_{B\sqcup B}^\mu$ defined by the decomposition $|\tilde \psi\rangle =  \sum_{\mu \in {\cal I}} \ |\tilde \psi_\mu \rangle$.  In analogy with our previous notation, we have defined $|\tilde \alpha, \tilde i\rangle{}_{LR} : = |\tilde \alpha\rangle_L \otimes |\tilde i \rangle_R$.  Furthermore, we emphasize that the probabilities
\begin{equation}
p_\mu = \langle \tilde \psi_\mu | \tilde \psi _\mu \rangle = \langle \psi_\mu | \psi_\mu \rangle
\end{equation}
are identical to those used in \eqref{eq:rhopsimatrixelements} due to the fact that our map $\chi: {\cal H}_{B \sqcup B} \rightarrow \widetilde{\cal H}_{B \sqcup B}$ is an isometry.

Comparing \eqref{eq:Trrho} and \eqref{eq:tildeTrrho}, and recalling that the traces $\Tr$ and $\widetilde{\Tr}$ do {\it not} agree, we find that despite -- and one might even say, because of -- agreement between the left-hand sides of
\eqref{eq:Trrho} and \eqref{eq:tildeTrrho}, the density operators $\rho_\psi$ and $\tilde \rho_\psi$ will generally represent distinct elements of ${\cal A}^B_L$.  As a result, using $\rho_\psi$ and $\Tr$ to define the entropy of $|\psi\rangle$ via \eqref{eq:VNentropy} generally leads to different results than using $\tilde \rho_\psi$ and $\widetilde{\Tr}$.  We emphasize here that while the definition of $\tilde \rho_\psi$  used the state $|\tilde \psi \rangle$ as an intermediate step, this
$|\tilde \psi\rangle$  was constructed from $|\psi\rangle$ using the isometric embedding $\chi$, so
$\tilde \rho_\psi$ is still uniquely determined by the original state $|\psi \rangle$.

Of course, since $\widetilde{\Tr}$ and $\tr$ are identical functions on ${\cal A}^B_L$, we may choose to write \eqref{eq:tildeTrrho} in the form
\begin{equation}
\label{eq:tildeTrrho2}
\langle \psi |a | \psi \rangle = \tr(\tilde \rho_\psi^{1/2}a \tilde \rho_\psi^{1/2}).
\end{equation}
We may correspondingly use \eqref{eq:VNentropy} with the path integral trace $\tr$ to define a notion of von Neumann entropy $S_{vN}^L$ associated with the left algebra ${\cal A}^B_L$ for general (normalized) states $|\psi \rangle \in {\cal H}_{B \sqcup B}$.  We may also perform the standard computation to relate the total entropy to the average of entropies in each $\mu$-sector:
\begin{equation}
\label{eq:trS}
S_{vN}^L(\psi) = \tr(-\tilde \rho_\psi \ln \tilde \rho_\psi) = \sum_{\mu \in {\cal I}} \ p_\mu  \, \tr(-\tilde \rho^\mu_\psi \ln \tilde \rho^\mu_\psi) - \sum_{\mu \in {\cal I}} \ p_\mu \ln p_\mu.
\end{equation}
Here the superscript ${}^L$ emphasizes that while the state $|\psi\rangle$ is pure, we are considering a notion of entropy associated only with the left algebra ${\cal A}^B_L$.

The important points of our discussion above are that there is a $\tilde \rho_\psi$ that for $a\in {\cal A}^B_L$ correctly computes expectation values in the original state $|\psi\rangle$, and that \eqref{eq:trS} can be represented as a standard entropy defined by `tracing out the right tensor factor' of each $\mu$-sector of the Hilbert space $\widetilde {\cal H}_{B \sqcup B}$.  Again, the Hilbert space $\widetilde {\cal H}_{B \sqcup B}$ is not a tensor product of Hilbert spaces, though it could be written as a direct sum of such spaces so that the full density matrix $\tilde \rho_\psi$ is defined by a direct sum over the $\tilde \rho^\mu_\psi$.  As a result, there is a natural `entropy of mixing' contribution given by the last term on the right-hand side of \eqref{eq:trS}.

Having defined the entropy $S_{vN}^L$ for general states $|\psi \rangle$, we would now like to discuss this entropy in a theory that admits a bulk semiclassical limit described by a familiar theory of gravity (say, Einstein-Hilbert or Jackiw-Teitelboim plus perturbative corrections).  In general, such a limit will give a good description only of appropriately semiclassical states $|\psi\rangle$.  With this as motivation, let us thus consider a normalized $|\psi\rangle$ defined by the path integral with boundary conditions given by some single smooth source-manifold-with-boundary $\psi$ in the sense that it is an element of $Y^d_{B \sqcup B}$ multiplied by a normalization constant to ensure $\langle\psi |\psi\rangle=1$.  We can of course easily extend the analysis to the case of finite linear combinations described by $\U{Y}^d_{B \sqcup B}$ so long as the number of terms and the coefficients remain fixed in the desired semi-classical limit, but we expect general states $|\psi\rangle$ to be more difficult to study using semiclassical techniques.

When $|\psi\rangle$ is defined by such a boundary-source-manifold $\psi$, we can use \eqref{eq:tildeTrrho2} to show that $\tilde \rho_\psi$ is the operator in ${\cal A}^B_L$ defined by the boundary-source-manifold $\psi\psi^\star$. In $\hat{A}_{L}$ we might call this operator $\widehat{(\psi \psi^\star)}_L$, though we will refer to it below as simply $\psi\psi^\dagger$ since we wish to regard it as a member of
${\cal A}^B_L$.
 \begin{figure}[t]
        \centering
\includegraphics[width=0.35\linewidth]{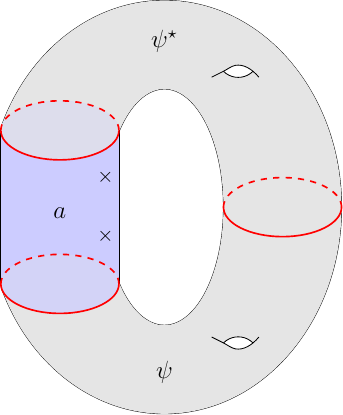}\caption{When $a$ and $\psi$ are both simple two-boundary surfaces, the path integral shown computes both $\tr (\psi^\star a \psi)$  and $\langle \psi |a |\psi\rangle$.  For such $a,\psi$, this shows that $\tilde \rho_\psi = \widehat{(\psi \psi^\star)}_L$ satisfies \eqref{eq:tildeTrrho2} since the former also agrees with $\tr(\tilde\rho_\psi^{1/2} a \tilde\rho_{\psi}^{1/2})$.}
\label{fig:rhopsi}
\end{figure}
This identification will not be a surprise to most readers, as when $a \in Y^d_{B \sqcup B}$ the argument in figure \ref{fig:rhopsi} establishes \eqref{eq:tildeTrrho2}.  However, for completeness we define $\tilde \rho_\psi := \psi \psi^\dagger$ and present the following more general argument that holds for any positive $a \in {\cal A}^B_L$ (which we write in the form $a=bb^\dagger$):
\begin{eqnarray}
\label{eq:showtilderhopsi}
\tr \left( \tilde \rho_\psi^{1/2}b   b^\dagger \tilde \rho_\psi^{1/2} \right) &=& \tr \left( b^\dagger \tilde \rho_\psi b \right) \cr
&=& \tr \left( b^\dagger \widehat{\psi}_L \widehat{\psi^\star}_L b \right) \cr
&=& \tr \left( \widehat{\psi^\star}_L b b^\dagger \widehat{\psi}_L \right) \cr
&=& \lim_{\beta \downarrow 0} \langle \tilde C_{\beta} |\widehat{\psi^\star}_L b b^\dagger \widehat{\psi}_L| \tilde C_{\beta}  \rangle \cr
&=& \lim_{\beta \downarrow 0} \langle \psi \tilde C_{\beta} | b b^\dagger|\psi \tilde C_{\beta}  \rangle \cr
&=&\langle \psi | b b^\dagger|\psi \rangle.
\end{eqnarray}
In this argument, the 1st and 3rd steps use cyclicity of the trace \eqref{eq:cyclicvN} on the von Neumann algebra, the 4th step uses \eqref{eq:alttrace52}, and the final step uses the facts that the states $|\psi \tilde C_{\beta}  \rangle = \widehat{\tilde C_{\beta}}_R |\psi \rangle$ converge to $|\psi \rangle$ (according to Lemma \ref{lemma:Cnorm} and Corollary \ref{cor:Cto1} of appendix \ref{app:lemmas}) and that the operator $b b^\dagger$ is bounded.  Comparing the beginning and end of \eqref{eq:showtilderhopsi} shows that $\tilde \rho_\psi$ satisfies \eqref{eq:tildeTrrho2} as desired, and thus that it is the correct density operator for the state $|\psi\rangle$.

The important consequence of this observation is that for all $n \in \mathbb{Z}^+$ we may also write
\begin{equation}
\label{eq:gravrep}
\tr(\tilde \rho_\psi^n) = \zeta\left(M(\left[\psi \psi^\star\right]^n)\right),
\end{equation}
where $(\psi \psi^\star)^n$ is just the product (say, using the left product $\cdot_L$) of $n$ copies of $\psi \psi^\star$.  In this form we see that the traces on the left-hand side of \eqref{eq:gravrep} are computed by applying what may be called the `gravitational replica trick' to $\psi \psi^\star$; see figure \ref{fig:gravrep}.
\begin{figure}[ht!]
        \centering
\includegraphics[width=0.4\linewidth]{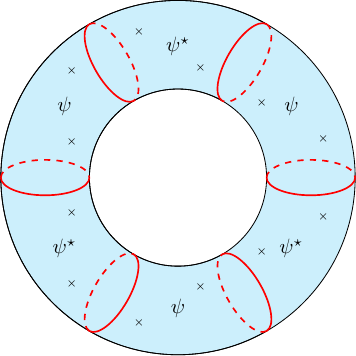}\caption{ The manifolds $M(\left[\psi \psi^\star\right]^n)$ that define path integrals computing $\tr(\tilde \rho_\psi^n)$ are constructed by cyclically gluing together $n$ alternating copies of $\psi$ and $\psi^\star$. The case shown here has $n=3$.}
\label{fig:gravrep}
\end{figure}

As a result, if such path integrals admit an appropriate semiclassical limit described by saddles of either Einstein-Hilbert or Jackiw-Teitelboim gravity with perturbative corrections, then we may argue as in \cite{Lewkowycz:2013nqa} that in this limit the von Neumann entropy \eqref{eq:trS} is given by the Ryu-Takayanagi entropy $A[\gamma]/4G$, where $\gamma$ is the minimal surface homologous to the left $B$ in the bulk saddle that dominates the path integral defined by $M(\psi \psi^\star)$.  This limiting result then receives perturbative corrections from quantum effects in the bulk as described in \cite{Faulkner:2013ana,Engelhardt:2014gca,Dong:2017xht}, and from higher derivatives terms in the classical action as described in \cite{Dong:2013qoa,Camps:2013zua,Miao:2014nxa}}. This representation of the Ryu-Takayanagi entropy as the semiclassical limit of an entropy defined by standard Hilbert space operations on $\widetilde{H}_{B \sqcup B}$ is the second key result of this work.

\section{Examples}
\label{sec:ex}

Let us now discuss examples of theories that satisfy the axioms of section \ref{subsec:ax} and which illustrate the above results.
A large class of examples is provided by boundary CFTs via the AdS/CFT correspondence; in this context, the CFT partition function provides a $\zeta(M)$ that satisfies all our axioms. However, $\zeta(M)$ defined in this way is generally not written directly in the language of bulk gravitational variables. Thus below we will discuss examples in which $\zeta(M)$ is defined directly in the bulk.
We will focus on models that are known to exist, which means in practice that they  must be extremely simple.  We will thus consider only two-dimensional bulk systems\footnote{It is worth noting that the toy models studied in \cite{Benini:2022hzx} can be viewed as three-dimensional bulk examples where $\zeta(M)$ is defined by a Chern-Simons theory (suitably gauged) and satisfies our axioms.}, defined either by the topological model of \cite{Marolf:2020xie} (see section \ref{subsec:top} below) or by an appropriate completion of Jackiw-Teitelboim (JT) gravity (see section \ref{subsec: JT}),   perhaps coupled to end-of-the-world (EOW) branes or matter defined by some quantum field theory (or by some proxy for such a QFT; see section \ref{subsec:JTm}). This discussion will be brief, but it illustrates how the decomposition \eqref{eq:HBBdecomp} and the associated hidden-sector dimensions $n_\mu$ can be non-trivial, as well as what the implications might be for understanding semiclassical entropy computations.  For such models it will be useful to note that any codimension-2 boundary $B$ is a zero-dimensional manifold, which means that it is a discrete collection of points.  We will focus on the case where this collection is finite so that the number of points is some $m\in \mathbb{Z}^+$ which we will call `the number of boundaries'.

An important point, however, is that taking a strict semiclassical limit of such models generally leads to algebras that are {\it not} of type I, or at least that contain continuous spectra for central operators.  This indicates that there is some sense in which one can take a class of models that satisfy our axioms for finite values of their couplings and, by taking an appropriate limit, one can nevertheless arrive at models which violate our axioms.  The nature of such limits and the manner in which they violate our axioms will be discussed in section \ref{subsec:typeII}.

\subsection{2d topological gravity}
\label{subsec:top}

Let us begin with the topological gravity model of \cite{Marolf:2020xie}, first without EOW-branes.  Here the allowed closed source manifolds are disjoint unions of circles, and the topological nature of the model means that the path integral depends on the number of circles but is completely independent of their lengths. The associated source-manifolds-with-boundary are thus unions of line segments.  In particular, they always have an even number of boundary points so that the $m$-boundary sectors
with $m$ odd are all empty.  This is in particular true for the one-boundary Hilbert space ($m=1$).
Thus, if we use $B$ to denote a single point, then ${\cal H}_B = \emptyset$ and, as a consequence, the product sector ${\cal H}^{\otimes}_{B\sqcup B}
= {\cal H}_B \otimes {\cal H}_B$ is empty as well.

It is important to recall that the framework described above applies separately to each baby universe superselection sector of such models.  Recall that for any such $\alpha$-
sector, there is exactly one state in the two-boundary sector
${\cal H}_{B\sqcup B}$. This is the state defined by taking the path integral boundary conditions to be given by a line segment.   This state plays a role similar to the cylinders $C_\beta$ discussed above for higher dimensions, so we will denote the line-segment state of this topological model as simply $|C\rangle$.    Since the model is topological, changing the length of the line segment does not affect the state.  And since there are no matter fields in the model, all line segments of a given length are diffeomorphic to each other.  However the norm $\langle C | C\rangle$ depends on the choice of $\alpha$-sector, and in fact turns out to fully characterize any $\alpha$-sector in this simple model \cite{Marolf:2020xie}.

Since the bulk theory has a one-dimensional two-boundary Hilbert space ${\cal H}_{B\sqcup B}$, all operators on this space are proportional to the identity operator.
This is in particular true of both the left and right von Neumann algebras ${\cal A}_L^B$ and ${\cal A}_R^B$.  The Hilbert space ${\cal H}_{B\sqcup B}\cong \mathbb{C}$ then factorizes in a trivial way according to $\mathbb{C} = \mathbb{C} \otimes \mathbb{C},$ where ${\cal A}^B_L$ and ${\cal A}^B_R$ are indeed both isomorphic to the (rather trivial) algebra of all bounded operators on $\mathbb{C}$.

Despite the uniqueness of $|C\rangle$ in ${\cal H}_{B\sqcup B}\cong \mathbb{C}$, the entropy of $|C\rangle$ defined by \eqref{eq:trS} is non-zero.
In the language of section \ref{subsec:normoftrace}, we may note that $C^2=C =  C^\dagger$ (since all of these are line segments and the length of the segment is irrelevant) so that $C$ is a projection. By the argument of that section we must then have $\tr(C) = n$ for some $n \in \mathbb{Z}^+$, but $n$ is not generally equal to unity\footnote{The requirement $\tr(C) \in \mathbb{Z}^+$ was not noticed in \cite{Marolf:2020xie} and provides further constraints on the allowed couplings beyond those listed in that work. These constraints will be discussed further in forthcoming work.}.  It is thus the rescaled cylinder $c = C/n$ that has unit trace, though we see that $c^2 = C^2/n^2 = c/n$ is not a projection.  In particular,
since $\ln c = \ln C - \ln n$ and $C \ln C =0$ (as is always the case for a projection), the normalized state $n^{-1/2}|C\rangle$ has left density matrix $\tilde \rho = C/n = c$ for which the entropy is
$S_{vN} = \tr(-c \ln c) = \ln n$.  This entropy can be reproduced by embedding our one-dimensional Hilbert space in ${\cal H}_n \otimes {\cal H}_n$ by mapping $n^{-1/2}|C\rangle$ to ${n}^{-1/2}|C\rangle \otimes |\chi \rangle$ for some normalized maximally entangled state $|\chi \rangle$.  In other words, this model provides an explicit example where the hidden sectors of section \ref{subsec:SfromA} are required to give a strict Hilbert space interpretation of what might here be called the Ryu-Takayanagi entropy.

One way to make the model more interesting is to add some number $k$ of flavors of end-of-the-world branes.  This case was also discussed in \cite{Marolf:2020xie}. End-of-the-world branes lead to a non-trivial one-boundary sector ${\cal H}_B$, where the dimension of ${\cal H}_B$ is $\min(k,n)$ with $n=\tr(C)$ (this $\tr(C)$ was called $d$ in \cite{Marolf:2020xie}). Now consider the two-boundary sector ${\cal H}_{B\sqcup B}$.  For $k\ge n$ {\it  all} two-boundary states lie in the tensor product sector:  ${\cal H}_{B\sqcup B} = {\cal H}_B \otimes {\cal H}_B$.  But for $k < n$ there is precisely one new state in ${\cal H}_{B\sqcup B}$ that does not lie in the tensor product sector.  It is given by the part of $|C\rangle$ that is orthogonal to all two-boundary states defined by having two end-of-the-world branes. The decomposition \eqref{eq:HBBdecomp} thus takes the two-term form
\begin{equation}
 {\cal H}_{B\sqcup B} = \left({\cal H}_B \otimes {\cal H}_B\right) \oplus {\cal H}^\perp_{B \sqcup B},
\end{equation}
with ${\cal H}^\perp_{B \sqcup B} \cong \mathbb{C}$. Of course, we can again make the trivial tensor product decomposition ${\cal H}^\perp_{B \sqcup B} \cong \mathbb{C} = \mathbb{C} \otimes \mathbb{C}$ for this sector of the Hilbert space.  The story of entropy in this sector is then similar to what occurs in the topological model without EOW branes discussed above, though the hidden sectors are now of dimension $n-k$.

\subsection{Pure asymptotically-AdS JT gravity}
\label{subsec: JT}

Another example to consider is the UV-completion of pure JT gravity.
Again the
allowed closed source manifolds are disjoint unions of circles, though now the path integral does depend on their lengths.
There are no one-boundary states in this model, and because JT gravity has no local degrees of freedom there is a basis of 2-boundary states of the form $|E,E\rangle$, so that the left and right energies $E$ necessarily  agree in all states.

Furthermore, all operators in our algebras are again defined by line segments, though now we find a distinct operator for every possible length of the boundary.  For a segment of length $\beta$ we may call this operator $e^{-\beta H}$.  As a result, the left von Neumann algebra ${\cal A}^B_L$ is now the abelian algebra defined by bounded functions of the Hamiltonian $H$.  This means that the factors in \eqref{eq:HBBdecomp} are labeled by the allowed energies $E$, and the algebra contains a separate factor for each value of $E$. We can thus use the eigenvalues $E$ to label sectors ${\cal H}^\mu_{B \sqcup B}$ in our general decomposition \eqref{eq:HBBdecomp}.   We may thus replace all labels $\mu$ with $\mu=E$ below.   Each of the associated factors is again just the trivial algebra $\mathbb{C}$, which is also known as the type $I_1$ factor.  Each ${\cal H}^E_{B \sqcup B}$ is the correspondingly-trivial one-dimensional Hilbert space of states proportional to $|E, E\rangle$.  At finite values of the couplings, we have shown that our axioms require the set of $\mu$ labels to be discrete, which means that in any given baby universe $\alpha$-sector the spectrum of $H$ must be discrete as well.

However, this model has a semiclassical limit in which we can compute Ryu-Takayanagi-like entropies.  Such entropies can be studied in microcanonical thermofield-double-like states defined by some window of energy eigenvalues $[E_1,E_2]$.    Such computations are not semiclassical in the regime where the window contains only one eigenvalue $E$, but they can be semiclassical when this window is relatively large.  As a result, semiclassical methods are not sufficient to compute the integers $n_E$ associated with a given value of $E$.  But if we assume these integers to be small (or, at least, not exponentially large), we find that the dominant contribution to a large RT entropy must come from the final entropy-of-mixing term in \eqref{eq:trS}.  A large RT entropy then indicates that the probabilities $p_\mu$ are exponentially small, and thus that there is an exponential density of energy states even if hidden sectors are not included in the Hilbert space.

One might also ask what the construction of ensembles dual to JT gravity tells us about potential bulk hidden sectors of the form described in section \ref{subsec:SfromA}.  Recall that the present work has focused  on a single member of such an ensemble rather than the ensemble as a whole, and that the detailed properties of the model can differ from one member of the ensemble to another.  We have in fact already seen this dependence on the member of the ensemble in discussing the topological model of \cite{Marolf:2020xie} in which the dimension $n$ of the hidden sectors was given by $\tr\, C$, so that such sectors are trivial (and thus unnecessary) in members of the ensemble with $\tr\, C=1$.

To discuss the issue for JT, we should in fact understand that the construction of the dual ensemble may require extra information that is not obviously present in the original bulk theory; see again the discussions in \cite{Stanford:2019vob,Johnson:2019eik,Johnson:2020exp,Johnson:2020mwi,Johnson:2021owr,Johnson:2021zuo,Johnson:2022wsr,Johnson:2023ofr,
Turiaci:2023jfa}.  It is thus better to simply assume that we are given a dual matrix ensemble and to then attempt to interpret this ensemble in bulk terms.

In fact, for reasons explained above, in the present context we should assume that we are simply given a {\it single dual matrix}.  Now, as described previously, from the bulk point of view there is only one observable that can be studied in JT gravity.  This observable is the energy, which is the bulk dual of the given matrix.  But the quantities that can actually be studied in the bulk are the eigenvalues of the matrix so, as was already noted, each such eigenvalue thus uniquely determines a (one-dimensional) Hilbert space sector ${\cal H}^E_{B \sqcup B}$ (where each $B$ is a single point), or equivalently a unique normalized bulk state $|E,E\rangle$.

On the other hand, a basis for the Hilbert space of the dual matrix theory is given by the full set of eigen{\it vectors} of the matrix.  For a generic diagonalizable matrix, the eigenvalues can be used to label its eigenvectors so that there is no harm in calling the matrix eigenvectors $|E, E\rangle$ as well (since here the second $E$ in $|E,E\rangle$ is a redundant label that is not even in principle allowed to differ from the first).  However, when the matrix has degenerate eigenvalues such a labelling cannot be complete.  Instead, we must introduce an additional degeneracy index (say, $I_E$) for any degenerate eigenvalue, in which case the dual matrix eigenvectors might be denoted $|E, E, I_E\rangle$.  Note that on the JT gravity side of the duality there is simply no way to resolve the distinction between states labelled by distinct values of the index $I_E$.

In order to agree with the bulk Gibbons-Hawking entropy computation, such degeneracy indices $I_E$ can range only over some finite number of values $n_E$ for each $E$. The addition of the index $I_E$ is thus precisely equivalent to saying that the space of matrix eigenvectors with eigenvalue $E$ is given by the tensor product $\tilde {\cal H}^E_{B \sqcup B}:=  {\cal H}^E_{B \sqcup B} \otimes {\cal H}_{n_E}$ of a hidden sector ${\cal H}_{n_E}$ with the bulk sector  ${\cal H}^E_{B \sqcup B}$ constructed by the path integral.  It is reasonable to expect that non-trivial hidden sectors often arise through such accidental degeneracies.  However, as is true in the current case, when this occurs such degeneracies should arise only in  a measure zero subset of the dual ensemble and may thus be safely neglected.

\subsection{JT gravity with ``matter''}
\label{subsec:JTm}

Finally, it will be useful to discuss JT gravity coupled to quantum matter fields, or at least a simple toy model thereof.  Since simple path integrals for JT coupled to quantum fields do not define finite partition functions \cite{Saad:2019lba}, we  will proceed here by simply writing down by fiat algebras that we deem plausible for UV-completions of such models.

In particular, let us consider a putative UV-complete model in which there are again no one-boundary states, but where the two boundary states are now spanned by $|E,E'\rangle$ for some discrete set of values for $(E, E')$; note that discreteness is guaranteed by the finiteness of {$\tr (C_\beta) = \widetilde \Tr(C_\beta)$. In particular, $E$ is allowed to differ from $E'$.  Let us also assume that the left von Neumann algebra ${\cal A}^B_{L}$ includes all bounded functions of the left Hamiltonian, and that it contains operators that can change any $|E_1,E'\rangle$ to any $|E_2, E'\rangle$.  We will refer to such operators as matter operators.  In this case, we can introduce Hilbert spaces ${\cal H}_{L}, {\cal H}_{R}$ which respectively have bases $\{|E\rangle\}$, $\{|E'\rangle\}$.  These are not the one-boundary Hilbert space $ {\cal H}_{B}$, as ${\cal H}_B$ has already been declared to be empty.  Nevertheless, we see explicitly that ${\cal H}_{B \sqcup B} =  {\cal H}_{L} \otimes {\cal H}_{R}$.  We also see that ${\cal A}^B_{L}$ can change any state $|\alpha\rangle \otimes |\beta\rangle$ into any other tensor product state $|\alpha'\rangle \otimes |\beta\rangle$ with some other $|\alpha'\rangle$ and the same $|\beta\rangle$.  As a result, ${\cal A}^{B}_L$ acts on the entire Hilbert space as
the algebra of bounded operators ${\cal B}({\cal H}_{L})$  on ${\cal H}_{L}$.

Despite the explicit factorization of  ${\cal H}_{B \sqcup B}$, this Hilbert space is not what we called the product sector ${\cal H}_B \otimes {\cal H}_B$, since the product sector is trivial in this model.  As a result, the trace $\tr(P)$ of a one-dimensional projection $P$ can in principle be any positive integer $n \in \mathbb{Z}^+$, though here it must in fact be the same positive integer $n$ for all such projections $P$.  This integer would need to be computed in any given model, and $n >1 $ would suggest that the model be augmented by the addition of hidden sectors.

For any such $n$ this model yields a single type I factor.  As a result, entropy in \eqref{eq:trS} comes entirely from the {\it first} term and does not involve any entropy of mixing.  The pure JT case, where (up to contributions from hidden sectors) the entropy was entirely given by the entropy-of-mixing term in \eqref{eq:trS}, might thus seem sharply different.  But this distinction is not really so large in the sense that (assuming the energy eigenvalues to be non-degenerate and the number of such eigenvalues involved to be large in comparison with the above integer $n$) in the theory with `matter' we can replace $\tr$ by $\widetilde \Tr$ and then compute this trace in the energy basis to find an entropy-of-mixing-like formula $S_{vN}^L(\psi) \approx - \sum_E p_E \ln p_E$, where $p_E$ is the probability of finding the system to have left-energy $E$.   We thus see that while the two terms in \eqref{eq:trS} are distinct in a given model, small changes in the model can move a given physical contribution from one term to another.  This should not really be a surprise as the spectrum of $\mu$-sectors is generally defined by the spectra of operators in the center ${\cal Z}$ of (say) the left algebra, and one might think that -- much as in our discussion of degenerate eigenvalues for pure JT gravity -- the existence of any non-trivial central operators at all requires a fine tuning to set commutators to zero.

\subsection{Axiom violations in semiclassical limits}
\label{subsec:typeII}
As noted in the introduction to the current section, semiclassical gravity
generally leads to algebras that are not of type I, or at least that contain continuous spectra
for central operators.  Let us suppose that such models arise from limits of UV-complete models at finite couplings, and let us also suppose that such finite-coupling  UV-complete models are to satisfy our axioms.  Then there must be a sense in which one can take a class
of models that satisfy our axioms for finite values of their couplings and, by taking an
appropriate limit, one can nevertheless arrive at models which violate our axioms.   It is useful to describe how such violations arise in the context of definite simple models.

Let us therefore consider again either pure JT gravity or our imagined UV-complete theory of JT gravity coupled to quantum fields. JT gravity contains a parameter $S_0$ that controls the semiclassical Gibbons-Hawking entropy of the ground state, and which weights contributions to the path integral by $e^{S_0 \chi}$ where $\chi$ is the Euler characteristic of the spacetime. Taking the limit $S_0\rightarrow \infty$ thus suppresses contributions from higher topologies.  In this limit, in the case with quantum fields, Penington and Witten argued that the left von Neumann algebra on ${\cal H}_{B\sqcup B}$ is of type II \cite{Penington:2023dql}.  For pure JT gravity  the von Neumann algebra is an abelian algebra whose factors are necessarily of type I, though in the semiclassical limit the spectrum of the central operators becomes continuous\footnote{A model of pure 3d gravity (with restricted bulk topology) having a similar semiclassical algebraic structure was recently studied in \cite{Chua:2023ios}.} (see again \cite{Penington:2023dql}).  Either of these results would be forbidden by our analysis, so in both cases the limit must violate at least one of our axioms.

There are in fact at least 4 different ways that one might attempt to discuss the large $S_0$ limit of a JT-gravity theory.  We now discuss each of them in turn, though the discussion of each will be quite short.

The first approach is simply to keep $S_0$ finite, but to take it to be larger than other quantities of interest.  In this case, the results of \cite{Penington:2023dql} would tell us only that the algebra when coupled to matter is {\it approximately} of type II, or that the spectra for pure JT are {\it approximately} continuous in the sense that any spacing between energy levels is much smaller than other parameters of interest.  This, of course, does not require any actual violation of our axioms.  Instead, it suggests only that there is some `near violation' whose form will become clear by discussing the other approaches below.

The second approach is to note that discussions of semiclassical physics tend to focus on disk amplitudes, and to recall from the above discussion that the disk amplitude is weighted by $e^{S_0 \chi(disk)} = e^{S_0}$.  To keep the disk amplitude finite as $S_0 \rightarrow \infty$, one might thus rescale the gravitational path integral $\zeta$ by defining $\tilde \zeta_2 = e^{-S_0} \zeta$.   This rescaling will preserve most of our axioms, though it will necessarily violate the factorization axiom.  In particular, since $\zeta$ satisfies factorization we have
\begin{equation}
\label{eq:zetatildefact}
\tilde \zeta_2 (M_1 \sqcup M_2) = e^{-S_0} \zeta (M_1 \sqcup M_2) =  e^{-S_0} \zeta (M_1) \zeta(M_2) = e^{+S_0}  \tilde \zeta_2 (M_1)  \tilde \zeta_2 (M_2).
\end{equation}
Furthermore, this violation becomes arbitrarily strong in the limit $S_0 \rightarrow \infty$.  It is thus no surprise that the $S_0 \rightarrow \infty$ limit does not have the properties described in this work.

A further concern in the second approach just described is that only the single-disk amplitudes are finite in the limit $S_0 \rightarrow \infty$.  In particular, as one can see from \eqref{eq:zetatildefact}, amplitudes that involve larger numbers of disks will still diverge.  The third approach is designed to remedy this problem and to also maintain factorization.  To do so, we rescale the path integral by  $e^{-mS_0}$ where $m$ is the number of circles that define the boundary conditions; i.e., we define
$\tilde \zeta_3(M) = e^{-mS_0} \zeta(M)$ where $m$ is the number of connected components of the closed (and compact) boundary source-manifold $M$ and then consider the limit of $\tilde \zeta_3$ as $S_0 \rightarrow +\infty$.

While this third approach is more satisfactory with regard to both finiteness and factorization, it can run afoul of reflection positivity.  In particular, let us recall that the proof of the trace inequality \eqref{eq:TrInvN} involved positive-definiteness of the inner product on the four-boundary Hilbert space.  Let us then further recall that the relevant computation turned out to involve both path integrals with what in the current context would be $m=1$ boundary circle, as well as path integrals with $m=2$ boundary circles.  The computation is thus sensitive to the fact that, in defining $\tilde \zeta_3$, we have changed the relative weights between these two terms by a factor of $e^{-S_0}$.

In fact, for any computation that involves only disks, defining $\tilde \zeta_3(M) = e^{-mS_0} \zeta(M)$ is equivalent to simply using the original path integral $\zeta$ with $S_0=0$.  On the other hand, performing this rescaling and taking $S_0 \rightarrow \infty$ suppresses the contributions of all non-disk topologies, even though these would have made extremely important contributions to the original path integral $\zeta$ if we had in fact set $S_0 =0$.  In this way we see that reflection-positivity can easily fail for the rescaled path integral $\tilde \zeta_3$ even if it holds for the original path integral $\zeta$ at all finite values of $S_0$.  And, indeed, the argument leading to \eqref{eq:TrInvN} shows that the failure of the trace inequality for the rescaled path integral is directly equivalent to such reflection-positivity violations.

This then brings us to the fourth and final (and perhaps the most sensible) approach to discussing the large $S_0$ limit of JT gravity.  In the Hilbert space sector with $2k$ boundaries, if we really wish to keep all amplitudes finite without sacrificing reflection-positivity, then we should simply rescale the path integral by $e^{-kS_0}$; i.e., we might define $\tilde \zeta_{4, 2k} = e^{-kS_0} \zeta$. We emphasize here that the rescaling depends on the choice of codimension-2 boundary that defines the Hilbert space sector, and not directly on the number of spacetime boundaries that define the path integral.  We also emphasize that, as we saw explicitly in figure \ref{fig:4bndy}, computations in a given Hilbert space sector can involve path integrals with varying numbers of codimension-1 boundaries.  As a result,  as indicated by the notation $\tilde \zeta_{4, 2k}$, after performing this rescaling we no longer have a single path integral that defines the entire theory.  Instead, we have effectively separated the sectors of the Hilbert space associated with different numbers of boundaries and declared each to be its own separate theory.  We see that taking $S_0\rightarrow \infty$ will mean that the inner product in the $2k$-boundary sector is determined entirely by amplitudes with $k$ disks.  In particular, while the one-disk amplitudes will make finite contributions to the inner product in the two-boundary sector, their contributions to the inner product in any higher-boundary sector of the Hilbert space has been set to zero.  While this is a natural semiclassical treatment of the system, it clearly violates our first (and from some perspectives most trivial) axiom which simply states that all computations are controlled by the same path integral.  It should thus again be no surprise that type II behavior and/or continuous spectra for central operators can arise in the limit $S_0 \rightarrow \infty$.

We expect similar comments to apply to the $G \rightarrow 0$ limit of higher dimensional UV-complete theories, though in that case there is no clear analogue of the 2nd approach.

\section{Discussion}
\label{sec:disc}

For the convenience of the reader, the results of the somewhat lengthy preceding sections will now be briefly summarized in section \ref{subsec:summary} below.  We will then provide some further remarks concerning these results in section \ref{subsec:remarks}. Finally, we will conclude with a short discussion of open issues and future directions in section \ref{subsec:OIFD}.

\subsection{Summary}
\label{subsec:summary}

The work above considered the possibility that quantum theories of gravity admit UV-completions associated with objects that can be called `Euclidean path integrals'.   We took such objects to satisfy 5 simple axioms that we call finiteness, reality, reflection positivity, continuity, and factorization. The first of these axioms states that the path integral defines a map $\zeta$ to $\mathbb{C}$ from the space of smooth closed $d$-dimensional boundary-source-manifolds for some $d$. Here the bulk theory is thought of as being of some dimension $D > d$. The reality, reflection positivity, and factorization axioms were of the standard form.  In particular, the factorization axiom required $\zeta(M_1 \sqcup M_2) = \zeta(M_1) \zeta(M_2)$ for closed source-manifolds $M_1, M_2$. Although general gravitational path integrals may naively violate this assumption, when the 
path integral defines a positive-definite inner product general arguments suggest\footnote{The technical caveats were reviewed in section \ref{subsec:ax}.} that the theory will decompose into
so-called baby universe superselection sectors in which factorization is satisfied
\cite{Marolf:2020xie}, in which case our arguments apply sector-by-sector.
The remaining axiom of continuity was extremely weak and required only continuity in $\epsilon$ in contexts where it was possible to insert a `cylinder' $C_\epsilon$ of the form $B \times [0,\epsilon]$ into the boundary source manifold while maintaining smoothness of the path integral boundary conditions.
The details of the axioms were described in section \ref{subsec:ax}.

Section \ref{sec:PIHS} used these axioms to construct Hilbert space sectors ${\cal H}_B$ associated with closed boundary manifolds $B$ of dimension $d-1$.  In particular, the reflection positivity axiom implies that the inner product is positive-definite on all such ${\cal H}_B$ for any $B$ (whether or not $B$ is connected).  When $B$ is of the form $B=B_1 \sqcup B_2$, where $B_1,B_2$ are also closed and compact, section \ref{sec:algebras} then defined a surface algebra of operators $A^{B_1}_L$ acting at $B_1$ and a second algebra of operators $A^{B_2}_R$ acting at $B_2$.  The path integral also defined a trace operation on these algebras. Importantly, our axioms imply both algebras to be represented by {\it bounded} operators when acting on ${\cal H}_{B_1 \sqcup B_2}$.  This was shown to be a consequence of positivity of the inner product on the higher-boundary Hilbert spaces ${\cal H}_{B_1 \sqcup B_1 \sqcup B_1 \sqcup B_1}$ and ${\cal H}_{B_2 \sqcup B_2 \sqcup B_2 \sqcup B_2}$, which in particular implied the trace inequality \eqref{eq:TrIn1} recently discussed in \cite{Colafranceschi:2023txs}.

Since these representations involved only bounded operators, they could be completed to von Neumann algebras
${\cal A}^{B_1,B_2}_L$, ${\cal A}^{B_1,B_2}_R$.  Although the original algebra $A^{B_1}_L$ was independent of $B_2$, its von Neumann algebra completion ${\cal A}^{B_1,B_2}_L$ does generally depend on the $B_2$ that defines the ${\cal H}_{B_1 \sqcup B_2}$ used to construct the completion.  We analyzed only the diagonal case $B_1 = B_2 =B$, leaving the more general case for future work. In the diagonal case one can denote the algebras more simply as ${\cal A}^B_L$, ${\cal A}^B_R$. The above trace then admits an extension to a trace $\tr$ on the full von Neumann algebras ${\cal A}^{B}_L$ and ${\cal A}^{B}_R$ as shown at the end of section \ref{sec:algebras}.

Critically, section \ref{sec:typeI} showed that this extended trace also satisfies a trace inequality of the form \eqref{eq:TrInvN}.  Together with the results in appendix \ref{app:trprop}, this implies the extended trace to be faithful, normal, and semifinite. Using the trace inequality again then  immediately implied our (diagonal) von Neumann algebras to be of type I, meaning that they are direct sums of type I factors ${\cal A}^B_{L, \mu}$ (or ${\cal A}^B_{R, \mu}$).  The diagonal Hilbert space sectors ${\cal H}_{B \sqcup B}$ also decompose as direct sums
\begin{equation}
\label{eq:HBB2}
{\cal H}_{B \sqcup B} = \bigoplus_{\mu \in {\cal I}} \ {\cal H}^\mu_{B \sqcup B}
\end{equation}
over Hilbert spaces ${\cal H}^\mu_{B \sqcup B}$, each of which is necessarily a tensor product
 ${\cal H}_{B \sqcup B,L}^\mu \otimes {\cal H}_{B \sqcup B,R}^\mu$ on which ${\cal A}^B_{L, \mu}$  (${\cal A}^B_{R, \mu}$) acts non-trivially only on the left (right) factor. Deriving \eqref{eq:HBB2} also made additional use of the trace inequality in showing the index set ${\cal I}$ to be discrete; i.e., in showing that \eqref{eq:HBB2} contains only a direct sum and not a more general direct integral.
The decomposition \eqref{eq:HBB2} provides at least one sharp sense in which we can show that quantum gravity Hilbert spaces (at least those associated with a given value of $\mu$) factorize into products of Hilbert spaces associated with natural subsystems; see e.g.\ \cite{Giddings:2015lla,Donnelly2018,Giddings:2020yes,Giddings:2021khn} for discussion of related issues.

Using positivity of the inner product on the Hilbert space sectors associated with $2n$ copies of $B$, section \ref{subsec:normoftrace} generalized the argument that led to the trace inequality \eqref{eq:TrInvN} to show that the path-integral-trace $(\tr)$ of any non-zero finite-dimensional projection must be a positive integer.  As a result, for some  $n_\mu \in \mathbb{Z}^+$  it agrees on  ${\cal A}^B_{L, \mu}$ with the standard Hilbert space trace defined by summing diagonal matrix elements over an orthonormal basis in the {\it extended} Hilbert space $\widetilde {\cal H}_{B \sqcup B,L}^\mu = {\cal H}_{B \sqcup B,L}^\mu \otimes {\cal H}_{n_\mu}$, where ${\cal H}_{n_\mu}$ is a `hidden sector' Hilbert space of dimension $n_\mu$.

 As a result of the type I structure, the trace $\tr$ can be used to define a notion of `left entropy' (or entropy with respect to the left algebra ${\cal A}^B_L$) on {\it pure states} $|\psi\rangle \in  {\cal H}_{B \sqcup B}$.  Furthermore, due to the relation between $\tr$ and the Hilbert space traces on both $\widetilde {\cal H}_{B \sqcup B,L}^\mu$ and the corresponding right extended factor $\widetilde {\cal H}_{B \sqcup B,R}^\mu$, this entropy can be interpreted in terms of an entropy of mixing term together with the familiar entropies $S_{vN}^{\mu, L} : = \widetilde \Tr_\mu ( -\tilde \rho_\psi^\mu \, \ln \, \tilde \rho_\psi^\mu)$ defined by considering the projections $|\psi_\mu\rangle$ of $|\psi\rangle$ to ${\cal H}^\mu_{B \sqcup B}$, isometrically embedding a normalized version of
$|\psi_\mu\rangle$ in the extended Hilbert space $\widetilde {\cal H}^\mu_{B \sqcup B} = \widetilde {\cal H}_{B \sqcup B,L}^\mu \otimes \widetilde {\cal H}_{B \sqcup B,R}^\mu$, tracing out the right factor in the usual way to define the density matrix $\tilde \rho_\psi^\mu$, and then summing expectation values of $-\tilde \rho_\psi^\mu \, \ln \, \tilde \rho_\psi^\mu$ over an orthonormal basis of $\widetilde {\cal H}_{B \sqcup B,L}^\mu$.  The final result for the left entropy of $|\psi\rangle$ then takes the form
\begin{equation}
\label{eq:vNsummary}
S_{vN}^L(\psi) = \sum_{\mu \in {\cal I}} \ p_\mu \, S_{vN}^{\mu, L} - \sum_{\mu \in {\cal I}} \ p_\mu \, \ln \, p_\mu,
\end{equation}
where $p_\mu := \langle \psi_\mu | \psi_\mu\rangle$ is the norm of $|\psi_\mu\rangle$ in  ${\cal H}^\mu_{B \sqcup B}$.

We then observed at the end of section \ref{sec:typeI} that, if our theory admits an appropriate limit described by semiclassical bulk Einstein-Hilbert or Jackiw-Teitelboim gravity, the corresponding limit of \eqref{eq:vNsummary} is given by the Ryu-Takayanagi entropy of the left $B$ as defined by the corresponding bulk saddle.  Quantum and higher derivative corrections are of course also incorporated in the usual way.

This then provides what one might call a Hilbert space interpretation of the Ryu-Takayanagi formula. We emphasize that it uses the extended Hilbert space
\begin{equation}
\label{eq:BBtilde2}
\widetilde {\cal H}_{B \sqcup B} := \bigoplus_{\mu \in {\cal I}}  \ \left[ \left( {\cal H}_{B \sqcup B,L}^\mu \otimes {\cal H}_{n_\mu}  \right)\otimes \left( {\cal H}_{B \sqcup B,R}^\mu \otimes {\cal H}_{n_\mu} \right) \right].
\end{equation}
The factors ${\cal H}_{n_\mu}$ are naturally called `hidden sectors' since -- aside from their connection to the trace $\tr$ and the associated entropy -- they are invisible to the algebras of observables defined by the original path integral.

We emphasize that nowhere in this work did we require the existence of a local dual field theory.  Of course, the axioms of section \ref{subsec:ax} will be true when such a dual theory exists, but the existence of a dual formulation would also entail much more structure.  In particular, none of our axioms require any form of locality for a hypothetical dual formulation (beyond the rather weak constraints implied by the factorization axiom).  Indeed, with the possible exception of the factorization requirement, our axioms are generally accepted as properties that should be satisfied by any bulk Euclidean gravitational path integral.  And while the factorization axiom is less obvious from the bulk point of view, as reviewed in section \ref{subsec:ax} this axiom in fact follows from the others modulo certain mathematical subtleties.  While such subtleties remain to be further investigated in any given model, it is nevertheless plausible that our axioms hold for any Euclidean UV-completion of a theory of quantum gravity, whether it be called string field theory, spin-foam loop quantum gravity, or by some other name.    What we find interesting about the above construction is just how much structure can be obtained with the simple and limited Axioms \ref{ax:finite}-\ref{ax:factorize}.

The final section (section \ref{sec:ex}) above discussed both topological and JT gravity examples in order to illustrate general features of our construction.  In particular, they provided contexts in which the decomposition \eqref{eq:HBB2} is non-trivial in the sense  that we required more than one value of $\mu$.  These examples also featured cases with non-trivial hidden sectors, and the JT example illustrated the idea that such sectors can arise due to accidental degeneracies in the bulk description.  We also discussed the fact that defining a semiclassical limit of JT gravity by taking a strict limit $S_0 \rightarrow \infty$ leads to violations of our axioms, so that it is no surprise that \cite{Penington:2023dql} finds the semiclassical theory to have a type II algebra and/or continuous spectra for central operators.

\subsection{Remarks}
\label{subsec:remarks}

Before concluding, we wish to make a few further remarks.  The first of these concerns constraints on the decomposition \eqref{eq:HBB2} that can be deduced by requiring that the theory admit a familiar semiclassical limit.  In particular, at least when the boundary $B$ is a sphere or a torus, and when the system is coupled to
an external bath that can absorb radiation,
 standard semiclassical physics tells us that large black holes can evaporate at least until they are microscopically small (when quantum gravity effects then fail to be under strict control). This is the case even if the original black hole is a two-sided Kruskal extension of e.g.\ a large AdS-Schwarzschild black hole with its famous Einstein-Rosen bridge connecting two distinct asymptotic regions.  By the usual arguments, this should be proportional to the semiclassical limit of the state that we have called $|C_\beta \rangle$ (or, equivalently, proportional to $|\tilde C_\beta\rangle$) for inverse temperatures $\beta$ less than the critical value $\beta_{HP}$ set by the Hawking-Page transition \cite{HawkingPage}.

We now couple one side of our system -- say the left side $B$ -- to a non-gravitational bath. Thus, time evolution will mix the gravity and bath degrees of freedom and, in this sense, will then change the state of the bulk.
However, so long as this coupling is describable by a real-time version of our path integral,
the time evolution operator should be a unitary operator built from elements of ${\cal A}^{B}_L$ and operators acting on the bath.
Then for any subspace ${\cal H}^\mu_{B \sqcup B}$, since the projection $P_\mu$ onto this subspace is in the center and commutes exactly with all elements of ${\cal A}^{B}_L$, it will also commute with the time evolution operator even in the context of coupling to a bath.  Furthermore,  since the time evolution is unitary, it will preserve the probability $p_\mu = \langle \psi | P_\mu | \psi \rangle$ associated with any subspace ${\cal H}^\mu_{B \sqcup B}$, where $|\psi\rangle$ is the normalized state on the entire gravity-with-bath system.
The decay of a large black hole to a small entropy object then tells us something about the parameters that describe whatever $\mu$-sectors were present in the original state.

To understand such constraints, we should first consider what values of $\mu$ will have non-zero probabilities $p_\mu$. Since the probabilities are preserved in time, this can be determined from the initial state $N^{-1}|\tilde C_\beta \rangle$ where $N = \sqrt{\langle \tilde C_\beta | \tilde C_\beta \rangle}$.   The question thus reduces to asking when $P_\mu |\tilde  C_\beta\rangle$ is a state of non-zero norm.  But faithfulness of the trace (Corollary \ref{cor:faithfulvN}) means that $\tr(P_\mu)\neq 0$, and since
$P_\mu^2 = P_\mu$, the limit $\lim_{\beta\downarrow 0} \langle \tilde  C_\beta |P_\mu|\tilde C_\beta\rangle$ is $\tr(P_\mu)$ according to \eqref{eq:alttrace52}.  Thus $P_\mu |\tilde  C_\beta \rangle$ must be non-zero for small enough $\beta$.  It follows that our black hole evaporation scenario will contain at least some information about all possible subspaces ${\cal H}^\mu_{B \sqcup B}$.

For simplicity, let us focus on the case where $B$ is a sphere.  In that context, semiclassical physics suggests that the evaporation continues until the area of the black hole is of order the Planck scale, so that (using the connection to the Ryu-Takayanagi formula described at the end of section \ref{subsec:SfromA}) the entropy $S_{vN}^L$ on the left $B$ satisfies $S_{vN}^L \sim A/4G =  O(1)$ at this point in the evaporation.  This tells us that \eqref{eq:trS} can take values as small as $O(1)$ in states with probabilities $p_\mu$ determined by the initial gravitational state $N^{-1}|\tilde C_\beta\rangle$ with small $\beta$.   Since both of the terms in \eqref{eq:trS} are positive, this must also be true of each term separately.  As a result, values of $\mu$ for which $p_\mu$ is exponentially small in $1/G$ can contribute at most a total probability of order $G$ to the state.  Similarly, since for each $\mu$ we must have $\tr (-\tilde\rho_\psi^\mu \, \ln \tilde\rho_\psi^\mu) \ge \ln n_\mu$, our $n_\mu$ can be exponentially large in $1/G$ only in a part of the state that contributes at most a total probability of order $G$.  In this sense one might say that `typical' values of $\mu$ must be associated with values of $n_\mu$ that are not exponentially large.  In such sectors the exact value of $n_\mu$ would then contribute to only a small part of the Ryu-Takayanagi entropy in standard situations where the entropy is $O(1/G)$.
It would thus be interesting to better understand whether similar constraints arise for other choices of the boundary $B$, or whether in some cases one finds instead that black hole areas are bounded below by a constant greater than zero\footnote{\label{Mac}For example, one might be concerned about the fact that the black holes with hyperbolic boundaries studied in \cite{Emparan:1998he,Birmingham:1998nr,Emparan:1999gf} have areas that are bounded away from zero.  However, it was shown in \cite{Dias:2010ma} that such hyperbolic black holes are unstable against scalar field condensation for low enough temperature, and when the mass of the scalar field is between the $d$-dimensional and 2-dimensional Breitenlohner-Freedman (BF) bound. The more dominant hairy black hole solution has vanishing horizon area at zero temperature. Similar results has been found for nearly extremal Reissner-Nordstrom black holes \cite{Hartnoll:2008kx} and for nearly extremal rotating black holes \cite{Dias:2010ma}. The resulting bounds on $p_\mu$ and $n_\mu$ are thus similar to the case of spherical $B$ discussed above, and the results of \cite{Horowitz:2022mly,Horowitz:2022leb} suggest that general cases will again behave similarly.}.

Our second remark involves comparisons between our axioms and those of various forms of quantum field theories.  In particular, it has sometimes been said that our axioms resemble Atiyah's postulates \cite{Atiyah_1988} for topological quantum field theory (TQFT); see also \cite{Witten:1988ze}.  Of course, there are also similarities to Segal's axioms \cite{Segal_1988} for conformal field theories and to the Osterwalder-Schrader axioms \cite{Osterwalder:1973dx} of general Euclidean field theories.  We view these similarities as natural for any formulation associated with Euclidean path integrals.  It may seem that our axioms have more in common with TQFT than the other contexts above simply because, first, there are fewer TQFT axioms due to the theories having less structure and, second, the fact that both here and in TQFT we explicitly allow path integrals defined by arbitrary manifolds $N$ with the same boundary $\partial N$ to define states in the same Hilbert space. This is, of course, also the case in non-topological quantum field theories, though it is less commonly emphasized in that context since one can also use local sources to generate general states.

However, it is perhaps more interesting to highlight differences between our axioms and those of TQFT. A critical such point is that TQFT requires the Hilbert space to factorize over disjoint unions (Harlow factorization).  This TQFT axiom encodes an element of locality that, as we have seen, distinctly fails to hold in our context.   We show, however, that (at least in the diagonal context) a weakened version of this property is still a consequence of our axioms in the sense that the Hilbert space $\mathcal H_{B \sqcup B}$ decomposes into a direct sum of terms $\mathcal H^\mu_{B\sqcup B}$ for  which the algebra of bounded operators ${\mathcal B}(\mathcal H^\mu_{B\sqcup B})$ is the tensor product of the algebras $\mathcal A^B_{L,\mu}$ and $\mathcal A^B_{R,\mu}$ associated with the two boundaries.

Another important distinction is that standard TQFT requires the associated Hilbert spaces to have finite dimension.  In contrast, we allow infinite dimensional Hilbert spaces but we are still able to show that our operators are bounded and that our construction yields type I von Neumann algebras\footnote{Generalization of TQFT with infinite dimensional Hilbert spaces have also appeared in~\cite{Rovelli:1993kc,Baez:1997zt,Barrett:1997gw,Freidel:1998pt,Freidel_2005,Oeckl:2003vu}, which attempt to use inspiration from TQFT to define the full gravitational path integral.}.  It would appear that one could also use the basic argument of section \ref{sec:typeI} to generalize the TQFT axioms to allow infinite-dimensional Hilbert spaces, and that one should again find bounded operators and type I algebras.  However, we have not explored this in detail.  We also refer the reader to \cite{Banerjee:2022pmw} for comments on how {\it topological} quantum gravity models provide a different generalization of TQFTs.

The third remark concerns the quantization condition derived in section \ref{subsec:normoftrace} for the trace $\tr(P)$ of any non-zero finite-dimensional projection $P$. In general, we might say that the unit-trace operators $\rho := P/\tr(P)$ define various notions of microcanonical ensemble (not necessarily specified by an energy), so that our quantization condition $\tr(P)=n \in \mathbb{Z}^+$ requires the quantity $\tr (-\rho \ln \rho)$ to take the value $\ln(n)$.  Note that the quantities $\tr (-\rho \ln \rho)$ defined by our path integral are then microcanonical and non-perturbative analogues of the semiclassical (canonical ensemble) partition functions studied by Gibbons and Hawking in their classic Euclidean path integral study of black hole entropy \cite{Gibbons:1976ue}.  It is thus natural to refer to the above result as `quantization of the Gibbons-Hawking density of states.'

This quantization allowed us to construct a Hilbert space $\widetilde {\cal H}_{B \sqcup B}$ (via the inclusion of finite-dimensional hidden sectors) on which the entropy
$\tr (-\rho \ln \rho) = \ln(n)$ directly measured the rank $n$ of the projection $P$ on appropriate left-factors of the Hilbert space $\widetilde {\cal H}_{B \sqcup B}$.  Since the inclusion of hidden sectors can only add states to the theory, we find that $\exp\left(\tr (-\rho \ln \rho)\right) = n$ also bounds the rank of $P$ in the context without hidden sectors; i.e., on the left-factors of the Hilbert space ${\cal H}_{B \sqcup B}$.  A similar result was previously derived in section 4.1 of \cite{Marolf:2020xie} using methods that also involved examining higher-boundary Hilbert spaces.   Here we have extended this result by deriving the above quantization condition and thus showing that the appropriate addition of hidden sectors will saturate the bound of \cite{Marolf:2020xie}.

As described in section \ref{subsec:normoftrace}, our quantization is a direct result of positivity of the inner product on the gravitational Hilbert space.  This should not be a surprise, as classic textbook classifications of e.g.\ unitary representations of the angular momentum algebra in fact take a similar form.  In that context one often begins with what will turn out to be a highest weight state.  One then acts repeatedly with lowering operators.  Assuming all of the states generated by this process to be non-zero would then lead to the construction of a state with negative norm, so at some point this process must terminate.  The condition that this occurs then enforces quantization conditions on the spectra of the relevant operators, and in particular on eigenvalues associated with the original state.  This is very much in parallel with our argument, which could be phrased in terms of first supposing some value for $\tr(P)$ and then, starting with the state $|P\rangle$, constructing more complicated states on Hilbert spaces associated with more and more boundaries. Eventually, at some point determined by $\tr(P)$, one finds that such states must have negative norm unless they are trivial.  The required triviality then imposes $\tr(P) = n \in {\mathbb Z}^+$.

One of the interesting lessons from this work is thus that, in quantum gravity, important such constraints are imposed by considering contexts with large numbers of boundaries.  In particular, in a semiclassical limit the operators $P$ that describe interesting microcanonical ensembles would be expected to have entropies of order $1/G$, and thus values of $\tr(P)$ that are exponentially large in $1/G$.  Our work suggests that such values of $\tr(P)$, even if not integers, would nonetheless appear to be consistent unless one performs computations that involve a similarly exponentially large number of boundaries.  This would clearly be a monumental task.  On the mathematical side, this in particular supports suggestions recently enunciated by Witten \cite{WittenTalk} that while some notion of analytic continuation of integer $n$ results to non-integer cases should violate positivity, such violations of positivity would be invisible in any notion of an asymptotic expansion of such results at large $n$.  See also \cite{Colafranceschi:2023txs} for another recently-discussed sense in which constraints from positivity become invisible in the limit of a large density of states.

\subsection{Future directions}
\label{subsec:OIFD}

It is traditional to close with a discussion of open issues and future directions and, indeed, it seems that there is still much to explore. One such direction would be to understand the analogues of the arguments given above for Lorentzian path integrals, perhaps allowing special codimension-2 singularities as described in \cite{Marolf:2022ybi}.  Since our Axioms \ref{ax:finite}-\ref{ax:factorize} all appear to admit ready extensions to complex sources (which, from the Euclidean perspective, would include real Lorentzian boundary conditions), the main issue is likely to be how to properly state the sense in which the boundary conditions are required to be ``of Lorentz signature'' and in particular how to handle any singularities that arise.

Another limitation of the analysis above was that it studied only diagonal Hilbert space sectors ${\cal H}_{B \sqcup B}$ for which $B$ is a compact closed manifold (without boundary).  It would clearly be of interest to understand the non-diagonal case in more detail.  We will return to this issue in  future work \cite{Marolf:2024adj}.

A more interesting generalization might be to drop the requirement that $B$ be compact and closed, and to instead investigate the case where some $B$ and its complement $\bar B$ meet at some $\partial B = \partial \bar B$.  If there is a dual CFT, then the usual field theory
arguments lead us to expect that any von Neumann algebra associated with $B$ should be of type III.  However, it is important to understand how to derive this result from the bulk gravitational path integral.  Furthermore, despite the fact that the algebra is expected to be of type III, we would like to show that {\it states} on the algebra have a  reasonable notion of renormalized von Neumann entropy.  We have only just  scratched the surface with respect to this issue here, and there is much more to understand.

Another interesting path to explore would be to investigate whether small enlargements of the set of axioms might lead to significant enlargements in the class of results that can be derived.  It would be particularly interesting to understand if there are simple axioms (say, regarding spacetime wormholes) that would allow us to take a general non-factorizing path integral and to write it as a direct sum/integral over `baby universe superselection sectors' as described in section \ref{subsec:ax}.  It would similarly be interesting to find simple axioms which imply {\it Harlow} factorization, meaning that the direct sum \eqref{eq:HBB2} over $\mu$ would reduce to a single term.  This term would then necessarily be of the form ${\cal H}^B_L \otimes {\cal H}^B_R$ (except in cases where these one-boundary Hilbert spaces are trivial).    Developing additional examples would also be useful in this regard, so that the effect of new axioms can be more readily understood.

\section*{Acknowledgements}

DM thanks Geoff Penington for interactions that initially motivated this project, Douglas Stanford for a clarifying conversation about discrete vs.\ continuous spectra for central operators, Daiming Zhang for exchanges emphasizing the importance of using nets instead of just sequences to describe the closures of algebras, Maciej Kolanowski for suggesting some of the points in footnote \ref{Mac}, and the Perimeter Institute Quantum Fields and Strings group for interesting questions. XD thanks Chris Akers and Juan Maldacena for interesting discussions. ZW thanks Tom Faulkner and Elliott Gesteau for valuable discussions. EC thanks Alexey Milekhin for many useful conversations. EC also thanks Matteo Bruno, Laurent Freidel, Daniel Harlow, Fabio Mele and Daniele Oriti for pointing out the parallels with the TQFT framework, and for valuable comments on the subject. We also thank Xiaoyi Liu and Maciej Kolanowski for conversations that led to key parts of this work, and we thank Aron Wall for discussions related to time-reversal.  DM is grateful to the Perimeter Institute for its hospitality during important stages of the project. EC's participation in this project was made possible by a DeBenedictis Postdoctoral Fellowship and through the support of the ID\# 62312 grant from the John Templeton Foundation, as part of the \href{https://www.templeton.org/grant/the-quantum-information-structure-of-spacetime-qiss-second-phase}{``Quantum Information Structure of Spacetime'' Project (QISS)}. The opinions expressed in this project/publication are those of the authors and do not necessarily reflect the views of the John Templeton Foundation.
The work of XD was supported in part by the U.S. Department of Energy, Office of Science, Office of High Energy Physics, under Award Number DE-SC0011702, and by funds from the University of California.
The work of DM and ZW was supported by NSF grant PHY-2107939, and by funds from the University of California.

\appendix

\section{Properties of the unnormalized cylinder operator \texorpdfstring{$C_\beta$}{}}
\label{app:lemmas}

We now provide the proof of the following lemma:

\begin{lemma}
\label{lemma:Cnorm}
The operator norm $\|C_\beta \|$ of the (unnormalized) cylinder operator satisfies $\|C_\beta \| \rightarrow 1$ as $\beta \rightarrow 0$.
\end{lemma}
\begin{proof}
To show that this is the case, recall that since $C_\beta \in Y^d_{B \sqcup B}$ it defines a bounded operator.  Furthermore, since $C_\beta^\star = (C_\beta^*)^t = C_\beta$, we have $\left(\widehat{C_\beta}_L\right)^\dagger = \widehat{C_\beta^\star}_L = \widehat{C_\beta}_L$, so that $\widehat{C_\beta}_L$ is self-adjoint and can be diagonalized.  In addition, since $\widehat{C_\beta}_L = \widehat{C_{\beta/2}}_L^\dagger \widehat{C_{\beta/2}}_L$, the eigenvalues of $\widehat{C_\beta}_L$ are non-negative.

Now consider the family of operators $\widehat{C_{\beta/n}}_L$ for $n \in \mathbb{Z}^+$ and some fixed $\beta>0$.  The norm $\|C_{\beta/n}\|$ is the supremum of the set of eigenvalues of $\widehat{C_{\beta/n}}_L$.  But the operators $\widehat{C_{\beta/n}}_L$ have a common set of eigenstates $|\lambda\rangle$ with eigenvalues $\lambda^{1/n}$ for some bounded set of non-negative real numbers $\lambda$.  In particular, we have
\begin{equation}
\label{eq:A1}
\|C_{\beta/n} \| = \sup_\lambda \, \lambda^{1/n} = (\sup_\lambda \, \lambda)^{1/n} = \| C_\beta\|^{1/n}.
\end{equation}
Thus
\begin{equation}
\label{eq:limCnorm}
\lim_{n \rightarrow \infty} \|C_{\beta/n} \| =1.
\end{equation}

This establishes that we can find sequences of $C_\beta$  with $\beta \rightarrow 0$ for which $\|C_\beta\|\rightarrow 1$.  However, it remains to show that this convergence is sufficiently uniform that $\|C_\beta\|$ converges for an arbitrary sequence of $C_\beta$ with $\beta \rightarrow 0$.

  Suppose first that $\|C_\beta\| \le 1$ for all $\beta>0$.  Then since
$C_{\beta_1+\beta_2} = C_{\beta_1} C_{\beta_2}$ for $\beta_1, \beta_2 >0$, we have
\begin{equation}
\|C_{\beta_1 + \beta_2}\| \le \|C_{\beta_1}\|\,\, \|C_{\beta_1}\| \le \|C_{\beta_1}\|.
\end{equation}
In particular, in this case, for all $\beta>0$ we find that $\|C_\beta\|$ is monotonically increasing as $\beta$ decreases.  Thus \eqref{eq:limCnorm} for any fixed $\beta$ implies $\lim_{\beta \rightarrow 0} \|C_\beta\| =1$.

The remaining case occurs when $\|C_{\beta_0}\| > 1$ for some $\beta_0>0$. We will now show both that this requires $\|C_\beta \| > 1$ for all small enough $\beta$, and that for small enough $\beta$ the norms $\|C_{\beta}\|$ are bounded above by a quantity that tends to 1.  This will then establish $\lim_{\beta \rightarrow 0} \|C_\beta\| =1$ for this final case.

The condition $\|C_{\beta_0}\| > 1$ means that there is some state $|\psi\rangle$ for which $\langle \psi |\widehat{C_{2\beta_0}}_L |\psi \rangle > \langle \psi | \psi \rangle$.  Now, recall that states of the form $|a \rangle$ for $a\in \U{Y}^d_{B\sqcup B}$ are dense in ${\cal H}_{B \sqcup B}$, so that any state $|\psi \rangle$ can be approximated by such $|a\rangle$.  Since $\widehat{C_{2\beta_0}}_L$ is bounded, the expectation value of $\widehat{C_{2\beta_0}}_L$ is a continuous function of $|\psi\rangle$.  Thus there must also be some $a \in \U{Y}^d_{B \sqcup B}$ for which $\langle a |\widehat{C_{2\beta_0}}_L | a \rangle = \lambda \langle a | a \rangle \ne 0$ with $\lambda > 1$.

Let us now consider some small $\beta>0$ and write $\beta_0 = \lf \frac{\beta_0}{\beta} \rf \beta + \Delta$.
Here we use the notation $\lf \beta_0/\beta \rf$ to denote the greatest integer less than or equal to $\beta_0/\beta$.  Thus $0 \le \Delta < \beta$.  As a result, the continuity axiom (Axiom \ref{ax:continuity}) requires
\begin{equation}
\label{eq:Delta}
\langle a |\widehat{C_{2\Delta \Phantom}}_L | a \rangle \rightarrow \langle a | a\rangle
\end{equation}
as $\beta \rightarrow 0$ via any sequence for which $\Delta \neq 0$ (so that the left-hand side is well-defined).   Since for $\Delta \neq 0$ we have
\begin{equation}
\lambda \langle a | a \rangle = \langle a |\widehat{C_{2\beta_0}}_L | a \rangle = \langle a |\widehat{C_{\Delta \Phantom}}_L \widehat{C_{2n\beta}}_L \widehat{C_{\Delta \Phantom}}_L| a \rangle \le \|C_{\beta}\|^{2n} \langle a | \widehat{C_{2\Delta \Phantom}}_L| a \rangle,
\end{equation}
where $n=\lf \beta_0/\beta \rf$, and since Eq.~\eqref{eq:Delta} tells us that for any $\epsilon > 0$ we can find a $\beta_1$ such that for all $\beta < \beta_1$ we have  $\langle a |\widehat{C_{2\Delta \Phantom}}_L | a \rangle \le (1+\epsilon) \langle a |a\rangle$, we find
\begin{equation}
\label{eq:Cnormstep2}
\lambda \langle a | a \rangle \le  \|C_{\beta}\|^{2n} (1+\epsilon) \langle a | a\rangle
\end{equation}
for values of $\beta$ at which $\Delta \neq 0$.  But when $\Delta =0$ the result \eqref{eq:Cnormstep2} follows immediately from \eqref{eq:A1}.  Thus \eqref{eq:Cnormstep2} in fact holds for all $\beta < \beta_1$.

Let us now choose $1+\epsilon < \lambda$.  Then from \eqref{eq:Cnormstep2} we see that there must be a $\beta_1$ such that  $\|C_{\beta}\| \ge \left(\frac{\lambda}{1+\epsilon}\right)^{1/2n} > 1$ for all $\beta < \beta_1$.

On the other hand, we can also show that $\|C_{\beta}\|$ is bounded above by a quantity that tends to $1$ as $\beta\rightarrow 0$.  To do so recall that \eqref{eq:TrIn1} and \eqref{eq:anormbound} imply
\begin{equation}
\label{eq:Cbound}
\|C_\beta \| \le \sqrt{\tr\,(C_\beta^2)} \le \tr \, C_\beta,
\end{equation}
 and that $\tr \, C_\beta$ is a continuous function of $\beta$ for $\beta >0$.  As a result, on any fixed interval $[\beta_2, 2\beta_2]$ we find
$\tr \, C_\beta \le {\cal C}$ for some constant ${\cal C} >0$.  Thus \eqref{eq:Cbound} requires $\|C_\beta \| \le {\cal C}$ on this interval as well.

Furthermore, as discussed above, for any $\beta$ we have $\|C_{\beta/n}\| = \|C_{\beta}\|^{1/n}$.  Since any $\beta >0$ with $\beta < \beta_2$ can be written as $\beta'/n$ for some $\beta' \in [\beta_2, 2\beta_2]$, this gives us an upper bound ${\cal C}^{1/n}$ on such $\|C_\beta\|,$ where $n = \lf 2\beta_2/\beta \rf$.  Combining this with the above observation yields
\begin{equation}
1 < \|C_\beta\| \le {\cal C}^{1/n},
\end{equation}
for all $\beta< \min(\beta_1,\beta_2)$. And since
 ${\cal C}^{1/n} \rightarrow 1$ as $\beta \rightarrow 0$, we must have $\|C_\beta\| \rightarrow 1$ as claimed.
\end{proof}

The above argument also leads to the corollaries below.

\begin{corollary}
\label{cor:Cto1}
As $\beta \rightarrow 0$, the operators $\widehat{C_\beta}_L$ converge in the strong operator topology to the identity $\mathbb{1}$ on any ${\cal H}_{B \sqcup B}$.
\end{corollary}
\begin{proof} We wish to show $\widehat{C_\beta}_L |a \rangle \rightarrow |a\rangle$ for all $|a\rangle \in {\cal H}_{B \sqcup B}$. Due to Lemma \ref{lemma:Cnorm}, this is equivalent to $ \widehat{\tilde C_\beta}_L |a \rangle \rightarrow |a\rangle$.

  To see this, recall that any state $|a\rangle$ is the $n\rightarrow \infty$  limit of states $|a_n\rangle$ for $a_n \in \U{Y}^d_{B \sqcup B}$.  We may thus define $|\epsilon_n\rangle = |a\rangle - |a_n \rangle$ and $|\tilde \epsilon_{\beta, n}\rangle  = \widehat{ \tilde C_\beta}_L |a_n\rangle - |a_n \rangle$ to write
\begin{eqnarray}
\lim_{\beta \rightarrow 0} \widehat{\tilde C_\beta}_L | a\rangle &=& \lim_{\beta \rightarrow 0} \left( \widehat{\tilde  C_\beta}_L |a _n \rangle + \widehat{\tilde C_\beta}_L |\epsilon_n \rangle \right) \cr &=&  |a_n \rangle + \lim_{\beta \rightarrow 0} \left( |\tilde \epsilon_{\beta,n} \rangle + \widehat{ \tilde C_\beta}_L  | \epsilon_n \rangle    \right)  \cr
&=& |a \rangle -  |\epsilon_n \rangle + \lim_{\beta \rightarrow 0} \left( |\tilde \epsilon_{\beta,n} \rangle + \widehat{\tilde C_\beta}_L  |\epsilon_n \rangle \right)
= |a \rangle + \lim_{\beta \rightarrow 0} \left( (\widehat{\tilde C_\beta}_L-1)  |\epsilon_n \rangle \right).
\end{eqnarray}
Here the last step used the fact that the norm of $|\tilde \epsilon_{\beta,n} \rangle$ vanishes for each $n$ in the limit $\beta \rightarrow 0$, due to Lemma \ref{lemma:Cnorm} and the continuity axiom. Since the operator norm of $(\widehat{\tilde C_\beta}_L-1)$ is bounded by $2$ for all $\beta$, the norm of the remaining error term
$\lim_{\beta \rightarrow 0} \left( (\widehat{\tilde C_\beta}_L-1)  |\epsilon_n \rangle \right)$ can be bounded by an arbitrarily small constant at large enough $n$.  We are then free to take the limit $n\rightarrow \infty$ to establish Corollary \ref{cor:Cto1}.
\end{proof}

\begin{corollary}
\label{cor:Cnorms}
The norms $\|C_\beta\|$ satisfy $\|C_\beta\| = (\|C_{\beta=1}\|)^\beta$.
\end{corollary}
\begin{proof}
Note that for $m, n \in \mathbb{Z}^+$ we may generalize \eqref{eq:A1} to write
\begin{equation}
\label{eq:A1gen}
\|C_{m\beta'/n} \| = \sup_\lambda \, \lambda^{m/n} = (\sup_\lambda \, \lambda)^{m/n} = \| C_{\beta'}\|^{m/n},
\end{equation}
where $\lambda$ are the elements of the spectrum of $C_{\beta'}$.  Let us now set $\beta'=1$. For any rational $\beta$, \eqref{eq:A1gen} immediately gives the desired result. For any irrational $\beta$, take $m=\lf n\beta \rf$. Then $\frac{m}{n} < \beta  < \frac{m+1}{n}$, so
\begin{equation}
\label{eq:CnormupperB}
\|C_\beta\| \le \left(\|C_{m/n}\|\right) \left(\|C_{\beta - m/n}\|\right),
\end{equation}
while
\begin{equation}
\label{eq:CnormlowerB}
\left(\|C_{(m+1)/n - \beta}\|\right) \left(\|C_\beta\|\right) \ge \|C_{(m+1)/n}\|.
\end{equation}
But as $n \rightarrow \infty$, both $\beta-m/n$ and ${(m+1)/n - \beta}$ vanish.  Thus by Lemma \ref{lemma:Cnorm} we have  $\|C_{\beta - m/n}\| \rightarrow 1$ and
$\|C_{(m+1)/n -\beta}\| \rightarrow 1$.  Taking limits in \eqref{eq:CnormupperB} and \eqref{eq:CnormlowerB} thus yields
\begin{equation}
\label{eq:Cnormbounds1}
\lim_{n\to\infty} \|C_{(m+1)/n}\| \le \|C_\beta\| \le \lim_{n\to\infty} \|C_{m/n}\|.
\end{equation}
Using \eqref{eq:A1gen} we may rewrite this in the form
\begin{equation}
\label{eq:Cnormbounds2}
\lim_{n\to\infty} \left(\|C_{1}\|\right)^{\frac{m+1}{n}} \le \|C_\beta\| \le \lim_{n\to\infty} \left(\|C_{1}\|\right)^{\frac{m}{n}}.
\end{equation}
But since $\frac{m}{n}$ and $\frac{m+1}{n}$ both approach $\beta$, both the upper and lower bounds are $\left(\|C_1\|\right)^\beta$.  This establishes the desired result.
\end{proof}

\section{The trace is normal and semifinite}
\label{app:trprop}

This appendix establishes that the traces \eqref{eq:alttrace5} on the von Neumann algebras ${\cal A}^B_L$, ${\cal A}^B_R$ are both normal and semifinite.  We call these properties Lemmas \ref{lemma:normal} and \ref{lemma:semifinite} below.  Recall that normality and semifiniteness were defined in properties 4 and 5 at the beginning of section \ref{sec:typeI}.

\begin{lemma}\label{lemma:normal}
The trace $\tr$ defined by \eqref{eq:alttrace5} is normal on both ${\cal A}^B_L$ and ${\cal A}^B_R$.
\end{lemma}
We will give the proof for ${\cal A}^B_L$.  The argument for ${\cal A}^B_R$ is directly analogous.

\begin{proof}
Consider a bounded increasing net of positive operators $a_\nu \in {\cal A}^B_L$ for $\nu$ in some directed index set $\mathscr{J}$. Here `increasing' means that $a_{\nu}\le a_{\nu'}$ whenever $\nu \le \nu'$. For each $a_\nu$ we have
the definition
\begin{equation}
\label{eq:trappC}
\tr \, a_\nu :=\sup_{\beta > 0} \langle \tilde C_\beta|a_\nu|\tilde C_\beta\rangle.
\end{equation}
Furthermore, for an increasing net of positive operators, the expectation value in any state $|\psi \rangle$ is also an increasing net.  In particular, $a_{\nu} \le a_{\nu'}$ implies $\langle \psi|a_{\nu}|\psi\rangle \le \langle \psi|a_{\nu'}|\psi\rangle$ for all $|\psi\rangle$. Thus
\begin{equation}
\langle \psi | \big( \sup_{\nu} a_\nu \big) |\psi \rangle \ge \sup_{\nu} \langle \psi | a_\nu|\psi \rangle = \lim_{\nu \in \mathscr{J}} \langle \psi | a_\nu|\psi \rangle .
\end{equation}
In fact, proposition 4.64 of \cite{Douglas} shows that the above is actually an equality:
\begin{equation}
\langle \psi | \big( \sup_{\nu} a_\nu \big) |\psi \rangle = \sup_{\nu} \langle \psi | a_\nu|\psi \rangle = \lim_{\nu \in \mathscr{J}} \langle \psi | a_\nu|\psi \rangle .
\end{equation}

Combining these results gives
\begin{eqnarray}
\label{eq:normality}
\tr \, \Big( \sup_\nu a_\nu \Big) &=&\sup_{\beta>0} \Big( \langle \tilde C_\beta|\big(\sup_\nu a_\nu \big)|\tilde C_\beta\rangle \Big)= \sup_{\beta>0} \, \sup_\nu \, \langle \tilde C_\beta|a_\nu|\tilde C_\beta\rangle \cr &=&  \sup_{\beta>0,\, \nu}  \,  \langle \tilde C_\beta|a_\nu|\tilde C_\beta\rangle= \sup_\nu \, \sup_{\beta>0} \,  \langle \tilde C_\beta|a_\nu|\tilde C_\beta\rangle= \sup_\nu \left(\tr \, a_\nu \right).
\end{eqnarray}
The key point in \eqref{eq:normality} is that taking the supremum over $\nu$ always commutes with taking the supremum over $\beta$ since taking both supremums (in either order) is equivalent to taking the supremum over all pairs $(\nu,\beta)$.  The result \eqref{eq:normality} is the desired normality property.
\end{proof}

\begin{lemma}\label{lemma:semifinite}
The trace $\tr$ defined by \eqref{eq:alttrace5} is semifinite on both ${\cal A}^B_L$ and ${\cal A}^B_R$.
\end{lemma}
We will give the proof for ${\cal A}^B_L$.  The argument for ${\cal A}^B_R$ is directly analogous.

\begin{proof}
    We need only show that every non-zero positive $a \in {\cal A}^B_L$ satisfies $b \le a$ for some non-zero positive $b\in {\cal A}_L^B$ with finite trace, where the notation $b \le a$ means that $a-b$ is positive.

    Let us begin by recalling that the normalized cylinder operator $\widehat{\tilde C_{2\beta}}_L$ was defined to have operator norm $1$ (though it is not generally the identity).  Thus $\mathbb{1} - \widehat{\tilde C_{2\beta}}_L$ is positive.  It then follows that $\gamma^\dagger (\mathbb{1} - \widehat{\tilde C_{2\beta}}_L)\gamma$ is also positive for any bounded operator $\gamma$, since the expectation value in any state $|\psi\rangle$ will satisfy
    \begin{equation}\label{eq:diffpos}
    \langle \psi | \gamma^\dagger (\mathbb{1} - \widehat{\tilde C_{2\beta}}_L)\gamma | \psi \rangle \ge 0 .
    \end{equation}
The positivity of $\gamma^\dagger (\mathbb{1} - \widehat{\tilde C_{2\beta}}_L)\gamma$  is then equivalent to the statement
\begin{equation}
\gamma^\dagger \widehat{\tilde C_{2\beta}}_L \gamma \le \gamma^\dagger\gamma.
\end{equation}

Next recall that, since $a$ is positive, it is in fact of the form $\gamma^\dagger \gamma$ for $\gamma \in {\cal A}_L^B$.   The above result then implies that our trace is semifinite if we can show that
    $b := \gamma^\dagger \widehat{\tilde C_{2\beta}}_L \gamma$ has finite trace and that $b$ is non-zero for some $\beta >0$.  We have
    \begin{eqnarray}
\label{eq:semifin}
\tr(\gamma^\dagger \widehat{\tilde C_{2\beta}}_L \gamma ) &=& \sup_{\beta '>0} \langle \tilde C_{\beta'} | \gamma^\dagger \widehat{\tilde C_{2\beta}}_L \gamma | \tilde C_{\beta'} \rangle  \cr  &=& \sup_{\beta '>0} \langle \tilde C_{\beta} | \gamma \widehat{\tilde C_{2\beta'}}_L \gamma^\dagger | \tilde C_{\beta} \rangle \le \langle \tilde C_{\beta} | \gamma   \gamma^\dagger | \tilde C_{\beta} \rangle .
\end{eqnarray}
In writing \eqref{eq:semifin},  we have used \eqref{eq:swapbetas} to pass from the first line to the second.  The final step follows from \eqref{eq:diffpos}. The right-hand side is clearly finite for any $\beta>0$, so this establishes that our $b$ has finite trace.  Furthermore, since $a=\gamma^\dagger \gamma$ is non-zero, and since Corollary \ref{cor:Cto1} (together with Lemma \ref{lemma:Cnorm}) showed the operators $\tilde C_\beta$ to converge in the strong operator topology to the identity as $\beta \rightarrow 0$, for small enough $\beta$ the operator $b= \gamma^\dagger \widehat{\tilde C_{2\beta}}_L \gamma$ must be non-zero as well.  This establishes that $\tr$ is semifinite as claimed.
\end{proof}

\addcontentsline{toc}{section}{References}
\bibliographystyle{JHEP}
\bibliography{references}

\end{document}